\documentclass[acmsmall]{acmart}

\acmDOI{} 
\acmPrice{}

\setcopyright{none}

\usepackage{booktabs}   
\usepackage{subcaption} 
\usepackage{enumerate,paralist}
\usepackage{mathtools}
\usepackage{comment}
\usepackage{parskip}
\usepackage{multirow}
\usepackage{graphicx}
\usepackage{color}
\usepackage[T1]{fontenc}
\usepackage{textcomp}
\usepackage{siunitx}
\usepackage{tikz,pgffor}
\usetikzlibrary{arrows}
\usetikzlibrary{shapes}
\usetikzlibrary{calc}
\usetikzlibrary{automata}
\usetikzlibrary{positioning}

\tikzstyle{proglabel}=[shape=circle,draw,inner sep=0pt,minimum size=5mm]
\tikzstyle{tran}=[draw,->,>=stealth, rounded corners]

\usepackage{amsmath}
\usepackage{amssymb}

\usepackage{stmaryrd}
\usepackage{url}                  

\usepackage{threeparttable}

\usepackage{tabularx}
\usepackage{algorithm}
\usepackage{algorithmic}
\usepackage{dsfont}
\usepackage{listings}
\lstdefinelanguage{prog}
{
morekeywords={prob, if, then, else, fi, while, do, od, true, false, and, or, skip},
sensitive = false
}

\newcommand{\Rset}{\mathbb{R}}
\newcommand{\Nset}{\mathbb{N}}
\newcommand{\Zset}{\mathbb{Z}}

\newcommand{\pvars}{V_\mathrm{p}}
\newcommand{\rvars}{V_{\mathrm{r}}}

\newcommand{\loc}{\ell}
\newcommand{\tertime}{T}

\newcommand{\probm}{\mathbb{P}}
\newcommand{\expv}{\mathbb{E}}

\newcommand{\updf}{F}

\newcommand{\pv}{\mathbf{b}}
\newcommand{\sv}{\mathbf{r}}
\newcommand{\Sat}[1]{{\llbracket}{#1}{\rrbracket}}

\newcommand{\ie}{i.e.}

\newcommand{\rv}{r.v.}

\newcommand{\dist}{\mathfrak{d}}

\newcommand{\outp}{\mathsf{out}}

\newcommand{\convex}[1]{\mathcal{F}^1(\Rset^{#1})}
\newcommand{\cconvex}[2]{\mathcal{F}^{1,1}_{#2}(\Rset^{#1})}

\theoremstyle{acmdefinition}

\newtheorem{remark}{Remark}
\usepackage{times}

\begin{document}

\title[Proving Expected Sensitivity with Randomized Variable-Dependent Termination Time]{Proving Expected Sensitivity of Probabilistic Programs with Randomized Variable-Dependent Termination Time}         


\author{Peixin Wang}
\affiliation{
  \department{Basics Lab}              
  \institution{Shanghai Jiao Tong University}            
  \city{Shanghai}
  \country{China}                    
}
\email{wangpeixin@sjtu.edu.cn}          

\author{Hongfei Fu}
\authornote{Corresponding Author: Hongfei Fu, fuhf@cs.sjtu.edu.cn}          
\affiliation{
  \department{John Hopcroft Center for Computer Science}             
  \institution{Shanghai Jiao Tong University}           
  \city{Shanghai}
  \country{China}                   
}
\email{fuhf@cs.sjtu.edu.cn}         

\author{Krishnendu Chatterjee}
\affiliation{
  \institution{IST Austria (Institute of Science and Technology Austria)}           
  \city{Klosterneuburg}
  \country{Austria}                   
}
\email{krishnendu.chatterjee@ist.ac.at}         

\author{Yuxin Deng}
\affiliation{
  \department{Shanghai Key Laboratory of Trustworthy Computing}             
  \institution{East China Normal University}           
  \city{Shanghai}
}
\affiliation{
  \department{Center for Quantum Computing}             
  \institution{Pengcheng Laboratory}           
  \city{Shenzhen}
  \country{China}                   
}
\email{yxdeng@sei.ecnu.edu.cn}         

\author{Ming Xu}
\affiliation{
  \department{Shanghai Key Laboratory of Trustworthy Computing}             
  \institution{East China Normal University}           
  \city{Shanghai}
  \country{China}                   
}
\email{mxu@cs.ecnu.edu.cn}         

\begin{abstract}
The notion of program sensitivity (aka Lipschitz continuity) specifies that
changes in the program input result in proportional changes to the program
output.
For probabilistic programs the notion is naturally extended to expected
sensitivity.
A previous approach develops a relational program logic framework
for proving expected sensitivity of probabilistic while loops, where the number
of iterations is fixed and bounded.
In this work, we consider probabilistic while loops where the number of iterations is not fixed, but
randomized and depends on the initial input values.
We present a sound approach for proving expected sensitivity of such programs.
Our sound approach is martingale-based and can be automated through existing martingale-synthesis algorithms.
Furthermore, our approach is compositional for sequential composition of while loops under a mild side condition.
We demonstrate the effectiveness of our approach on several classical
examples from Gambler's Ruin, stochastic hybrid systems and stochastic gradient descent.
We also present experimental results showing that our automated approach can handle various probabilistic programs in the literature.
\end{abstract}





\maketitle

\section{Introduction}\label{sec:introduction}

\noindent{\em Continuity properties of systems.}
Continuity property for systems requires that the change in
the output is bounded by a monotone function of the change in
the input.
Analysis of continuity properties are of great interest in
program and reactive system analysis, such as:
(a)~robustness of numerical computations;
(b)~analysis of sensitivity of numerical queries~\cite{DBLP:journals/fttcs/DworkR14} in databases;
(c)~analysis of stability of learning algorithms~\cite{DBLP:journals/jmlr/BousquetE02};
and
(d)~robustness analysis of programs~\cite{DBLP:conf/popl/ChaudhuriGL10}.

\smallskip\noindent{\em Probabilistic systems.}
Continuity analysis is relevant for probabilistic systems in a similar way,
where the notion of continuity is extended with expectation to
average over the probabilistic behaviours of the system.
For example, statistical notions of differential privacy~\cite{DMNS06};
robustness analysis of Markov chains, Markov decision processes, and
stochastic games~\cite{Aldous83,DBLP:conf/icalp/Fu12,DBLP:journals/tcs/BreugelW06,DBLP:journals/tcs/DesharnaisGJP04,DBLP:conf/fossacs/Chatterjee12};
stability analysis of randomized learning algorithms~\cite{DBLP:journals/jmlr/BousquetE02,DBLP:conf/icml/HardtRS16};
all fall under the umbrella of continuity analysis of probabilistic systems.

\smallskip\noindent{\em Program sensitivity.}
A particular interest among continuity is {\em program sensitivity
} which specifies that the change in the output of a program is
proportional to the change in the input.
Formally, there is a constant $L$ (the \emph{Lipschitz constant}) such that
if the input changes by an amount $x$, then the change in the ouput is at most $L\cdot x$.
In this work we consider the expected sensitivity of probabilistic programs given
as (sequential composition of) probabilistic while loops.

\smallskip\noindent{\em Previous results.}
The expected sensitivity analysis of probabilistic programs was first considered in \cite{DBLP:journals/jmlr/BousquetE02,DBLP:conf/icml/HardtRS16} for machine-learning algorithms such as stochastic gradient descent, through manual proofs.
Then ~\cite{DBLP:journals/pacmpl/BartheEGHS18} proposed  an elegant method based on a relational program logic framework.
The heart of the analysis technique is coupling-based methods, and the approach is
shown to work effectively on several examples from machine learning to statistical physics.
A recent result~\cite{DBLP:conf/atva/HuangWM18} implemented a computer-algebra based tool that calculates tight sensitivity bounds for probabilistic programs.
Although these previous approaches address the expected sensitivity analysis well,
they work only on examples of probabilistic while loops whose
number of iterations is fixed and bounded (i.e., the number of iterations
is fixed to a given number $T$).
In reality, many examples of probabilistic while loops do not have fixed number of iterations,
rather the number of iterations is randomized and depends on the input values.
Hence, such examples cannot be handled by the previous approaches.
In this work, we focus on expected sensitivity analysis of such programs.

\noindent{\em Our contributions.}
Our main contributions are as follows:
\begin{compactenum}
\item We present a sound approach for proving expected sensitivity of probabilistic while loops whose number of iterations is randomized and depends on the initial input values.
\item We show that our approach is compositional w.r.t sequential composition.

\item In contrast to the previous coupling and computer-algebra  based approaches, our approach relies on ranking supermartingales (RSMs) (see~\cite{SriramCAV,DBLP:journals/toplas/ChatterjeeFNH18}), a central notion in proving termination properties of probabilistic programs.


\item Since RSM based approaches can be automated through constraint solving (see e.g.~\cite{SriramCAV,DBLP:journals/toplas/ChatterjeeFNH18}),
the same results in conjunction with our sound approach present an automated approach
for sensitivity analysis of probabilistic programs.
\item
We demonstrate the effectiveness of our approach through (i) a case study on stochastic gradient descent
and (ii) experimental results on various probabilistic programs from the literature, including Gambler's Ruin, stochastic hybrid systems, random walks, etc.
\end{compactenum}

\noindent{\em Technical contribution.} In terms of technical contribution there are key differences between our result and the previous results.
The previous approaches are either coupling-based proof rules, or through computer-algebra tools, and all of them are restricted  to loops with a fixed number of loop iterations.
In contrast, our approach is based on RSMs and can handle loops whose number of iterations is randomized and depends on the input.
Moreover, we prove the non-trivial fact that our approach is compositional under sequential composition.
Furthermore, as RSM-synthesis algorithms have been well-established in the literature, our sound approach directly lead to automated algorithms for proving expected sensitivity of probabilistic programs.

\noindent{\em Limitation.} Our approach mainly focuses on (sequential composition of) probabilistic while loops where there is no conditional branch.
Although the exclusion of conditional branches makes our contribution seemingly restrictive, we argue that typically inclusion of conditional branches will break sensitivity properties in general.
Consider a loop of the form
\[
\textbf{while } \Phi \textbf{ do if } b \textbf{ then } P \textbf{ else } Q
\textbf{ od}
\]
where the programs $P,Q$ perform completely different executions.
Then two close-by program inputs $x_1,x_2\models\Phi$ such that (**) $x_1\models b$ but $x_2\not\models b$ will lead to values that differ significantly after just one loop iteration.
Thus, irrespective of analysis methods this type of programs has bad sensitivity property.
Previous results also reflect the difficulty on handling conditional branches.
For example, in previous approaches such as \cite{DBLP:journals/pacmpl/BartheEGHS18,DBLP:journals/corr/abs-1901-06540}, it must be manually ensured that the conditions of all conditional branches are either (i) both satisfied or (ii) both not satisfied by two close-by program valuations (i.e., the situation (**) above not allowed) (see \cite[Figure~3]{DBLP:journals/pacmpl/BartheEGHS18} and \cite[Figure~1]{DBLP:journals/corr/abs-1901-06540}).
Moreover, in all the experimental examples from~\cite{DBLP:conf/atva/HuangWM18}, conditional-branches within for-loops are either restricted to a finite set of values or directly transformed into probabilistic branches.
For a possible extension to conditional branches see Remark~\ref{rmk:condbranch}.

\section{Probabilistic Programs}\label{sect:preliminaries}

We first present the syntax and semantics of our probabilistic programming language, then define the syntactical subclass of \emph{simple} while loops to which  our approach applies. Throughout the paper, we denote by $\Nset$, $\Zset$, and $\Rset$ the sets of all natural numbers, integers, and real numbers, respectively.

\smallskip\noindent{\bf The Syntax.}
Our probabilistic programming language is imperative and consists of statements.
We present a succinct description below (see Appendix~\ref{app:syntax} for the detailed syntax).

\begin{compactitem}
\item \emph{Variables.} Expressions $\langle\mathit{pvar}\rangle$ (resp. $\langle\mathit{rvar}\rangle$) range over program (resp. sampling) variables, respectively. Program variables are normal variables that control the flow of the program, while each sampling variable is a special variable whose value is sampled from a fixed predefined probability distribution each time the variable is accessed in the program.
\item \emph{Constants.} Expressions $\langle\mathit{const}\rangle$ range over decimals.
\item \emph{Arithmetic Expressions.} Expressions $\langle\mathit{expr}\rangle$ (resp. $\langle\mathit{pexpr}\rangle$) range over arithmetic expressions over both program and sampling variables (resp. program variables only).
    For example, if $x,y$ are program variables and $r$ is a sampling variable, then $x+3\cdot y$ is an instance of $\langle\mathit{pexpr}\rangle$ and $x-y+2\cdot r$ is an instance of $\langle\mathit{expr}\rangle$. In this paper, we consider a general setting of arithmetic expressions and  do not fix a detailed syntax for $\langle\mathit{expr}\rangle$ and $\langle\mathit{pexpr}\rangle$.
\item \emph{Boolean Expressions.} Expressions $\langle\mathit{bexpr}\rangle$ are boolean expressions
over program variables, for which atomic propositions are comparisons between expressions from $\langle\mathit{pexpr}\rangle$ and general expressions are built from
atomic propositions and propositional operators.
\item \emph{Statements $\langle \mathit{stmt}\rangle$.} Assignment statements are indicated by `$:=$';
`\textbf{skip}' is the statement that does nothing;
Standard conditional branches are indicated by the keyword `\textbf{if}' with its \textbf{then}- and \textbf{else}-branches, and a boolean expression that serves as the condition for the conditional branch.
Probabilistic choices are modelled as probabilistic branches with the key word ``\textbf{if prob}$(p)$'' that lead to the \textbf{then}-branch with probability $p$ and to the \textbf{else}-branch with probability $1-p$.
While-loops are indicated by the keyword `\textbf{while}' with a boolean expression as the loop guard.
Finally, sequential compositions are indicated by semicolons.
\end{compactitem}
Note that probabilistic branches can be implemented as a sampling of Bernoulli distribution followed by a conditional branch, but for algorithmic purpose we consider probabilistic branches directly.
In this work, we consider probabilistic programs without non-determinism. 



\smallskip\noindent{\bf The Semantics.}
To define the semantics, we first recall several standard notions from probability theory as follows (see e.g. standard textbooks~\cite{probabilitycambridge,Billinsleyprobability} for details).

\noindent{\em Probability Spaces.} A \emph{probability space} is a triple $(\Omega,\mathcal{F},\probm)$, where $\Omega$ is a nonempty set (so-called \emph{sample space}), $\mathcal{F}$ is a \emph{sigma-algebra} over $\Omega$ (i.e., a collection of subsets of $\Omega$ that contains the empty set $\emptyset$ and is closed under complementation and countable union), and $\probm$ is a \emph{probability measure} on $\mathcal{F}$, i.e., a function $\probm\colon \mathcal{F}\rightarrow[0,1]$ such that (i) $\probm(\Omega)=1$ and
(ii) for all set-sequences $A_1,A_2,\dots \in \mathcal{F}$ that are pairwise-disjoint
(i.e., $A_i \cap A_j = \emptyset$ whenever $i\ne j$)
it holds that $\sum_{i=1}^{\infty}\probm(A_i)=\probm\left(\bigcup_{i=1}^{\infty} A_i\right)$.
Elements in $\mathcal{F}$ are called \emph{events}.
An event $A\in\mathcal{F}$ is said to hold \emph{almost surely} (a.s.) if $\probm(A)=1$.

\noindent{\em Random Variables.} A \emph{random variable} (\rv) $X$ on a probability space $(\Omega,\mathcal{F},\probm)$
is an $\mathcal{F}$-measurable function $X\colon \Omega \rightarrow \Rset \cup \{-\infty,+\infty\}$, i.e.,
a function satisfying the condition that for all $d\in \Rset \cup \{-\infty,+\infty\}$, the set $\{\omega\in \Omega\mid X(\omega)<d\}$ belongs to $\mathcal{F}$.
By convention, we abbreviate $+\infty$ as $\infty$.

\noindent{\em Expectation.} The \emph{expected value} of a random variable $X$ on a probability space $(\Omega,\mathcal{F},\probm)$, denoted by $\expv(X)$, is defined as the Lebesgue integral of $X$ w.r.t $\probm$, i.e.,
$\expv(X)\coloneqq\int X\,\mathrm{d}\probm$;
the precise definition of Lebesgue integral is somewhat technical and is
omitted  here (cf.~\cite[Chapter 5]{probabilitycambridge} for a formal definition).
In the case that $\mbox{\sl ran}~X=\{d_0,d_1,\dots,d_k,\dots\}$ is countable with distinct $d_k$'s, we have that
$\expv(X)=\sum_{k=0}^\infty d_k\cdot \probm(X=d_k)$.

To present the semantics, we also need the notion of \emph{valuations}.

\noindent {\em Valuations.} Let $V$ be a finite set of variables with an implicit linear order over its elements.
A \emph{valuation}  on $V$ is a vector $\pv$ in $\Rset^{|V|}$ such that for each $1\le i\le |V|$,
the $i$-th coordinate of $\pv$, denoted by $\pv[i]$, is the value for the $i$-th variable in the implicit linear order on $V$. For the sake of convenience, we write $\pv[y]$ for the value of a variable $y$ in a valuation $\pv$.

\noindent{\em Program and Sampling Valuations.} Let $\pvars$ (resp. $\rvars$) be the set of program (resp. sampling) variables  appearing in a probabilistic program, respectively.
A \emph{program valuation} (or \emph{program state}) is a valuation on $\pvars$. A \emph{sampling valuation} is a valuation on $\rvars$.
Given a program valuation $\pv$ and a boolean expression $\Phi$, the satisfaction relation $\models$ is defined in the standard way so that we have $\pv\models\Phi$ iff $\Phi$ holds when program variables in $\Phi$ are substituted by their corresponding values in $\pv$.

Now we give a brief description of the semantics for probabilistic programs.
We follow the standard operational semantics through Markov chains.
Given a probabilistic program (without non-determinism), its semantics is given as a general state-space Markov chain (GSSMC)~\cite[Chapter 3]{gssmc},
where
(i) the state space consists of all pairs of program counters and program valuations for which the program counter refers to the next command to be executed and the program valuation specifies the current values for the program variables, and
(ii) the kernel function that specifies the stochastic transitions between states is given by the individual commands in the program.
For any initial state $\mathfrak{c}=({\mathsf{in}}, \pv)$ where ${\mathsf{in}}$ is the program counter of the first command and $\pv$ is the input program valuation, each probabilistic program induces a unique probability space through its corresponding GSSMC, where the sample space consists of all infinite sequences of states in the GSSMC (as \emph{runs}), the sigma-algebra is generated by all \emph{cylinder} sets of runs induced by
finite Cartisian products of measurable subsets of the state space, and the probability measure is uniquely determined by the kernel function and the initial state.
The detailed semantics can be found in~\cite{DBLP:journals/toplas/ChatterjeeFNH18,SriramCAV,DBLP:conf/vmcai/FuC19}.

Under our semantics, we denote by $\probm_\pv$ the probability measure for a probabilistic program with the input program valuation $\pv$ (note that the program counter ${\mathsf{in}}$ is determined by the program),
and by $\expv_\pv(-)$ the expectation under the probability measure $\probm_\pv$.

\noindent{\bf Simple While Loops.} In this paper, we focus on (sequential composition of) simple probabilistic while loops and investigate sound approaches for proving expected sensitivity over such programs. A \emph{simple} (probabilistic) while loop is of the form
\begin{equation}\label{eq:swl}
\textbf{while}~\Phi~\textbf{do}~P~\textbf{od}
\end{equation}
where $\Phi$ is the \emph{loop guard}
and the loop body $P$ is a program without nested while loops. As simple while loops are syntactically restricted, we present succinct notions for such programs.

{\em Update Functions.} Given a simple while loop in the form~(\ref{eq:swl}) with the disjoint sets $\pvars$ and $\rvars$ of program and sampling variables, we abstract away detailed executions of the loop body $P$ by an \emph{update function} $\updf:\mathbf{L}\times \Rset^{|\pvars|}\times \Rset^{|\rvars|} \rightarrow \Rset^{|\pvars|}$
that describes the input-output relationship for one iteration of the loop body
as follows.
First, we let $\mathcal{L}$ be the set of all program counters that refer to a probabilistic branch (i.e., $\textbf{if prob}(p)~\dots$) in the loop body of $P$.
Then we define $\mathbf{L}$ to be the set of all functions from $\mathcal{L}$ into the choices of branches (i.e., \textbf{then}- or \textbf{else}-branch).
Informally, such a function specifies for each probabilistic branch in $P$ which branch
is chosen in the current loop iteration.
Finally, the update function $\updf$ simply gives the program valuation $\updf(\loc, \pv,\sv)$ after the current loop iteration given (i) an element $\loc\in\mathbf{L}$ that specifies the probabilistic choices for probabilistic branches, (ii) a program valuation $\pv$ that specifies the values for program variables before the current loop iteration and (iii) a sampling valuation $\sv$ that gives all the sampled values for the sampling variables in the current loop iteration.
In this way, we abstract away the detailed execution within the loop body $P$ and represent it simply by an update function.
Note that as the loop body $P$ does not involve nested while loops, one can compute its update function symbolically through a recursive algorithm on the structure of $P$.

\noindent{\em Runs.} We also simplify the notion of runs over simple while loops. A \emph{run} for a loop in the form~(\ref{eq:swl}) is an infinite sequence $\{\pv_n\}_{n\ge 0}$ of program valuations such that each $\pv_n$ is the program valuation right before the $(n+1)$-th loop iteration. Note that if $\pv_n\models\Phi$, then $\pv_{n+1}=\updf(\loc_n, \pv_n,\sv_n)$ where $\loc_n$ (resp. $\sv_n$) specifies the probabilistic resolution to all the probabilistic branches (resp. the sampled values for the sampling variables) at the $(n+1)$-th loop iteration, respectively; otherwise, $\pv_{n+1}=\pv_n$.

\noindent{\em Notations.} To ease the use of notations, we always use $\pv$ for a program valuation, $\sv$ for a sampling valuation and $\loc$ for an element in $\mathbf{L}$, with possible super-/sub-scripts.
Given a simple while loop $Q$ in the form (\ref{eq:swl}), we always use  $\pvars$ for its set of program variables, $\rvars$ for
sampling variables, $\updf$ for the update function, $\Phi$ for the loop guard, and $P$ for the loop body.
Moreover, we denote by $\Sat{\Phi}$ the set $\{\pv\mid \pv\models\Phi\}$ of program valuations that satisfy the loop guard $\Phi$.

To reason about expected sensitivity of simple while loops, we require that the loop body is Lipschitz continuous.
This requirement is standard and corresponds to the ``$\eta$-expansiveness'' introduced in \cite[Definition 2.3]{DBLP:conf/icml/HardtRS16}.
This continuity condition needs the standard notion of \emph{metrics} that measures the distance between two program valuations, as follows.

\noindent{\emph{Metrics}.} 
A \emph{metric} is a function $\dist:\Rset^{|\pvars|}\times \Rset^{|\pvars|}\rightarrow [0,\infty)$ that satisfies (i) $\dist(\pv,\pv')=0$ iff $\pv=\pv'$, (ii) $\dist(\pv,\pv')=\dist(\pv',\pv)$ (\emph{symmetry}) and (iii) $\dist(\pv,\pv')\le \dist(\pv,\pv'')+\dist(\pv'',\pv')$ (\emph{triangle inequality}).
Informally, $\dist(\pv,\pv')$ is interpreted as the distance between the two program valuations.
For example, we can define $\dist$ either through the max norm by $\dist(\pv,\pv'):={\|}{\pv-\pv'}{\|}_\infty$ where the max norm ${\|}{\centerdot}{\|}_\infty$ is given as
${\|}{\pv''}{\|}_\infty:=\max_{z\in \pvars}|\pv''[z]|$, or through the Euclidean norm by $\dist(\pv,\pv'):={\|}{\pv-\pv'}{\|}_2$ where
${\|}{\centerdot}{\|}_2$ is given as ${\|}{\pv''}{\|}_2 :=\sqrt{(\pv'')^\mathrm{T}\pv''}$.
In this paper, we consider metrics that are comparable with the max norm, \ie, there exist real constants $D_1,D_2>0$ such that
\begin{equation}\label{metric}
D_1\cdot {\|}{\pv-\pv'}{\|}_\infty\le \dist(\pv,\pv')\le D_2\cdot {\|}{\pv-\pv'}{\|}_\infty\enskip.
\end{equation}
Note that the comparability is naturally satisfied for metrics derived from norms of finite dimension.


Below we describe the continuity of the loop body under a metric $\dist$.

\begin{definition}[Lipschitz Continuity $L$ of the Loop Body]\label{def:LipconLB1}
We say that the loop body of a simple while loop in the form (\ref{eq:swl}) is \emph{Lipschitz continuous} if there exists a real constant $L>0$ such that
\begin{compactitem}
\item[(B1)] $\forall\loc\,\forall \mathbf{r}\,\forall \pv,\pv'\colon\left[\pv,\pv'\models \Phi\Rightarrow \dist (\updf(\loc,\pv,\sv),\updf(\loc,\pv',\sv)) \leq L\cdot \dist(\pv,\pv')\right]$\enskip.
\end{compactitem}
If we can choose $L=1$ in (B1), then the loop is \emph{non-expansive}; otherwise it is \emph{expansive} (\ie, the minimum $L$ is greater than $1$).
\end{definition}

\lstset{language=prog}
\lstset{tabsize=3}
\newsavebox{\progb}
\begin{lrbox}{\progb}
\begin{lstlisting}[mathescape]

while $x\le 1000$ do
  $x:=x+r$
od

\end{lstlisting}
\end{lrbox}

\lstset{language=prog}
\lstset{tabsize=3}
\newsavebox{\progrunningexampled}
\begin{lrbox}{\progrunningexampled}
\begin{lstlisting}[mathescape]
while $\Phi$ do
  $i:= \mathsf{unif}[1,\ldots,n]$;
  $\mathbf{w}:= \mathbf{w}-\gamma\cdot\nabla G_i(\mathbf{w})$
od
\end{lstlisting}
\end{lrbox}

\lstset{language=prog}
\lstset{tabsize=3}
\newsavebox{\progloopguard}
\begin{lrbox}{\progloopguard}
\begin{lstlisting}[mathescape]

while $y\le x\wedge y\ge 0$ do
  $y:=y-1$
od

\end{lstlisting}
\end{lrbox}


\begin{figure}
\begin{minipage}{0.25\textwidth}
\centering
\usebox{\progb}
\caption{Running Example}
\label{fig:runningexample1}
\end{minipage}
\quad
\begin{minipage}{0.29\textwidth}
\centering
\usebox{\progrunningexampled}
\caption{An SGD Algorithm}
\label{fig:example4}
\end{minipage}
\begin{minipage}{0.4\textwidth}
\centering
\usebox{\progloopguard}
\caption{A Counterexample for Sensitivity}
\label{fig:loopguardcounterexample}
\end{minipage}
\end{figure}

\begin{remark}[Simple While Loops]
In general, any imperative probabilistic program can be transformed equivalently into a simple while loop by adding a variable for the program counter and then simulating the original program
through transitions between program counters and valuations.
However, the class of simple while loops that can be handled by our approach is restricted to those with Lipschitz-continuous loop body. Thus generally, our approach cannot handle conditional branches that usually breaks the continuity property.
\end{remark}


\begin{example}[The Running Example]\label{ex:runningexample}
Consider the simple while loop in Figure~\ref{fig:runningexample1}.
In the program, $x$ is a program variable and $r$ is a sampling variable.
In every loop iteration, the value of $x$ is increased by a value sampled w.r.t the probability distribution of $r$ until it is greater than $1000$.
There is no probabilistic branch so $\mathbf{L}$ is a singleton set that only contains the empty function.
The update function $\updf$ for the loop body is then given by
$\updf(\loc,\pv,\sv)[x]=\pv[x]+\sv[r]$ for program valuation $\pv$ and sampling valuation $\sv$, where
$\loc$ is the only element in $\mathbf{L}$.
By definition, the loop is non-expansive.
\end{example}

\vspace{-1em}
\section{Expected Sensitivity Analysis of Probabilistic Programs}\label{sect:sensitivitypro}

In this paper, we focus on averaged sensitivity which is one of the most fundamental sensitivity notions in expected sensitivity analysis of probabilistic programs.
Informally, averaged sensitivity compares the distance between the expected outcomes
from two close-by input program valuations.
The notion of averaged sensitivity has an important applicational value in that it can
be used to model algorithmic stability in many machine-learning algorithms (see e.g.~\cite{DBLP:journals/jmlr/BousquetE02}).


Below we illustrate the notion of averaged sensitivity formally. To ensure well-definedness, we only consider probabilistic programs that terminate with probability one (i.e., with \emph{almost-sure termination}~\cite{SriramCAV}) for all input program valuations.
Furthermore, as our approach will rely on ranking supermartingales, we actually require that the probabilistic programs we consider terminate with finite expected termination time~\cite{DBLP:journals/toplas/ChatterjeeFNH18}.
Below we fix a probabilistic program $Q$ 
and a metric $\dist:\Rset^{|\pvars|}\times \Rset^{|\pvars|}\rightarrow [0,\infty)$. 
\begin{definition}[Averaged Sensitivity~\cite{DBLP:journals/jmlr/BousquetE02,DBLP:journals/pacmpl/BartheEGHS18}]\label{def:EAS}
We say that the program $Q$ is \emph{averaged affine-sensitive} over a subset $U\subseteq \Rset^{|\pvars|}$ of input program valuations if there exist real constants $A,B\ge 0$ and $\theta\in (0,\infty]$ such that for all program variables $z$ and $\pv,\pv'\in U$,
\begin{equation}\label{eq:expaffsen}
\mbox{if }\dist(\pv,\pv')\le\theta\mbox{ then }|\expv_{\pv}(Z)-\expv_{\pv'}(Z')|\leq A\cdot{\dist(\pv,\pv')}+B
\end{equation}
where $Z,Z'$ are random variables representing the values of $z$ after the execution of the program $Q$ under the input program valuations $\pv,\pv'$, respectively.
Furthermore, if we can choose $B=0$ in (\ref{eq:expaffsen}), then the program $Q$ is said to be \emph{averaged linear-sensitive} in the program variable $z$.
\end{definition}

In the definition, the constants $A,B$ are sensitivity coefficients, while $\theta$ is the threshold below which the sensitivity is applicable; if $\theta=\infty$ then the sensitivity is applicable regardless of the distance between $\pv,\pv'$.
Informally, a program $Q$ is averaged affine-sensitive
if the difference in the expected value of any program variable $z$ after the termination of $Q$ is bounded by an affine function in the difference of the input program valuations.
Likewise, the program is averaged linear-sensitive if the difference can be bounded by a linear function.
In this way, we consider the expected sensitivity of the return values where each program variable represents an individual return value.
Note that another subtle issue arising from the well-definedness is that the random variables $Z,Z'$ in (\ref{eq:expaffsen})  may not be integrable.
In the following, we will always guarantee that the random variables are integrable.

As we only consider averaged sensitivity, in the rest of the paper we will refer to \emph{averaged affine-/linear-sensitivity} simply  as \emph{expected affine-/linear-sensitivity}.
It is worth noting that in ~\cite{DBLP:journals/pacmpl/BartheEGHS18}, a coupling-based definition for expected sensitivity is proposed.
Compared with their definition, our definition treats expected
sensitivity directly and do not consider couplings.
\section{Motivating Examples}\label{sect:motiv}

In the following, we show several motivating examples for expected sensitivity analysis of probabilistic programs.
We consider in particular probabilistic programs with a randomized number of loop iterations that also depends on the input program valuation.
As existing results~\cite{DBLP:journals/pacmpl/BartheEGHS18,DBLP:conf/atva/HuangWM18,DBLP:conf/icml/HardtRS16} only consider probabilistic loops with a fixed number of loop iterations,
none of the examples in this section can be handled by these approaches.

%
	
\lstset{language=prog}
\lstset{tabsize=3}
\newsavebox{\progexamplea}
\begin{lrbox}{\progexamplea}
\begin{lstlisting}[mathescape]
while $x\geq 1$ do
  if prob($\frac{6}{65}$) then
     $x:=x+1$;$w:=w+2$
  else if prob($\frac{4}{59}$) then
       $x:=x+2$;$w:=w+3$
  else if prob($\frac{3}{55}$) then
       $x:=x+3$;$w:=w+4$
  else if prob($\frac{2}{52}$) then
       $x:=x+5$;$w:=w+5$
  else if prob($\frac{1}{50}$) then
       $x:=x+11$;$w:=w+6$
  else $x:=x-1$
  fi fi fi fi fi od
\end{lstlisting}
\end{lrbox}
\lstset{language=prog}
\lstset{tabsize=3}
\newsavebox{\progexampleaa}
\begin{lrbox}{\progexampleaa}
\begin{lstlisting}[mathescape]
while $x\geq 1$ do
  if prob($\frac{6}{65}$) then
     $x:=x+r_1$;$w:=w+2$
  else if prob($\frac{4}{59}$) then
       $x:=x+r_2$;$w:=w+3$
  else if prob($\frac{3}{55}$) then
       $x:=x+r_3$;$w:=w+4$
  else if prob($\frac{2}{52}$) then
       $x:=x+r_4$;$w:=w+5$
  else if prob($\frac{1}{50}$) then
       $x:=x+r_5$;$w:=w+6$
  else $x:=x-r_6$
  fi fi fi fi fi od
\end{lstlisting}
\end{lrbox}
\begin{figure}[htbp]
\begin{minipage}{0.5\textwidth}
\usebox{\progexamplea}
\end{minipage}\hfill
\begin{minipage}{0.5\textwidth}
\usebox{\progexampleaa}
\end{minipage}
\caption{A Mini-roulette example (left) and its continuous variant (right)}
\label{fig:example1}
\end{figure}

\begin{example}[Mini-roulette]
A particular gambler's-ruin game is called \emph{mini-roulette}, which is a popular casino game based on a 13-slot wheel. A player starts the game with $x$ amount of chips.
He needs one chip to make a bet and he bets as long as he has chips. If he loses a bet, the chip will not be returned, but a winning bet will not consume the chip and results in a specific amount of (monetary) reward, and possibly even more chips. The following types of bets can be placed at each round. (1) \emph{Even-money bets}: In these bets, $6$ specific slots are chosen. Then the ball is rolled and the player wins the bet if it lands in one of the $6$ slots. So the player has a winning probability of $\frac{6}{13}$. Winning them gives a reward of two unit and one extra chip. (2) \emph{2-to-1 bets}: these bets correspond to $4$ chosen slots and winning them gives a reward of $3$ and $2$ extra chips. (3) \emph{3-to-1, 5-to-1 and 11-to-1 bets}: These are defined similarly and have winning probabilities of $\frac{3}{13}$, $\frac{2}{13}$ and $\frac{1}{13}$ respectively. Suppose at each round, the player chooses each type of bets with the same probability (i.e., chooses each type with probability $\frac{1}{5}$). The probabilistic program for this example is shown in Figure \ref{fig:example1}(left), where the program variable $x$ represents the amount of chips and
the program variable $w$ records the accumulated rewards. (In the program we consider that $x$ can take a real value.)
We also consider a continuous variant of the mini-roulette example in Figure \ref{fig:example1}(right), where
we replace increments to the variable $x$ by uniformly-distributed sampling variables $r_i(i=1,\dots,6)$ and
one may choose $r_1\sim \mathsf{unif}(1,2)$, $r_2\sim\mathsf{unif}(2,3)$, $r_3\sim\mathsf{unif}(3,4)$, $r_4\sim\mathsf{unif}(4,5)$, $r_5\sim\mathsf{unif}(8,9)$, $r_6\sim\mathsf{unif}(1,2)$  or other uniform distributions that ensure the termination of the program.
Note that the number of loop iterations in all the programs in Figure~\ref{fig:example1} is randomized and depends on the input program valuation as the loop guard is $x\ge 1$ and the increment/decrement of $x$ is random in each loop iteration.
In both the examples, we consider the expected sensitivity in the output program variable $w$ that records the accumulated reward.
\end{example}

\begin{example}[Multi-room Heating]
We consider a case study on multi-room heating from~\cite{DBLP:journals/ejcon/AbateKLP10}, modelled as a stochastic hybrid system that involves discrete and probabilistic dynamics.
In the case study, there are $n$ rooms each equipped with a heater. The heater can heat the room and the heat can be transferred to another room if the rooms are adjacent.
We follow the setting from~\cite{DBLP:journals/ejcon/AbateKLP10} that
the average temperature of each room, say room $i$, evolves according to the following stochastic difference equation that describes the transition from the $k$-th time step to the $(k+1)$-th time step, with constant time-interval $\Delta t$:
\begin{equation}\label{eq:hybrid}
\textstyle x_i(k+1)=x_i(k)+b_i(x_a-x_i(k))+\sum_{i\neq j}a_{ij}(x_j(k)-x_i(k))+c_i+w_i(k)
\end{equation}
where (i) $x_a$ represents the ambient temperature (assumed to be constant and equal for the whole building), (ii) the quantities $b_i$, $a_{ij}$, $c_i$ are nonnegative constants representing respectively the average heat transfer rate from room $i$ to the ambient (i.e., $b_i$), to adjacent rooms $j\neq i$ (i.e., $a_{ij}$'s), supplied to room $i$ by the heater (i.e., $c_i$), and (iii) $w_i(k)$  is the noise that observes a predefined probability distribution, such as Gaussian, Poisson or uniform distribution, etc.
In this paper, we consider two simplified scenarios. The first is a single-room heating modelled as the probabilistic program in Figure~\ref{fig:hybrida}, where the program variable $x$ represents the current
room temperature, the constants $x_a,b,c$ are as in (\ref{eq:hybrid}) and $w$ is the noise (as a sampling variable); the goal in the first scenario is to raise the room temperature up to \SI{20}{\degreeCelsius}.
We assume that the starting room temperature is between \SI{0}{\degreeCelsius} and \SI{20}{\degreeCelsius}.
The second is a double-room heating modelled in Figure~\ref{fig:hybridb}, where the heater of the main room is on and the heater for the side room is off.
In the figure, the program variable $x_1$ (i.e., $x_2$) represents the temperature for the main room (resp. the side room), respectively;
the constants $x_a,b_i,c_i, a_{ij}$ are as in (\ref{eq:hybrid}); the sampling variables $w_1,w_2$ represent the noises.
In the program, we adopt the succinct form of simultaneous vector assignment for updates to $x_1,x_2$.
In both scenarios, we have a loop counter $n$ that records the number of stages until the (main) room reaches \SI{20}{\degreeCelsius}.
We consider in particular the expected sensitivity w.r.t the total number of stages as recorded in $n$, for which we assume that the value of $n$ always starts with $0$.
\end{example}

\lstset{language=prog}
\lstset{tabsize=3}
\newsavebox{\hybrida}
\begin{lrbox}{\hybrida}
\begin{lstlisting}[mathescape]

while $0\le x\le 20$ do
  $x:=x+b*(x_a-x)+c+w$;
  $n:=n+1$
od

\end{lstlisting}
\end{lrbox}

\lstset{language=prog}
\lstset{tabsize=3}
\newsavebox{\hybridb}
\begin{lrbox}{\hybridb}
\begin{lstlisting}[mathescape]
while $0\le x_1\le 20\wedge 0\le x_2\le 20$ do
  $\begin{pmatrix} x_1\\ x_2\end{pmatrix}:=\begin{pmatrix} x_1+b_1*(x_a-x_1)+a_{12}*(x_{2}-x_1)+c_1+w_1 \\  x_2+b_2*(x_a-x_2)+a_{21}*(x_{1}-x_2)+w_2 \end{pmatrix}$;
  $n:=n+1$
od
\end{lstlisting}
\end{lrbox}

\lstset{language=prog}
\lstset{tabsize=3}
\newsavebox{\hybridc}
\begin{lrbox}{\hybridc}
\begin{lstlisting}[mathescape]
while $x_1\le 30 \wedge x_2\le 30 \wedge x_3 \le 30$ do
  $x_1:=x_1+b_1(x_a-x_1)+a_{12}(x_2-x_1)$
       $+a_{13}(x_3-x_1)+c_1+\mathbf{w}_1$;
  $x_2:=x_2+b_2(x_a-x_2)+a_{21}(x_1-x_2)$
       $+a_{23}(x_3-x_2)+c_2+\mathbf{w}_2$;
  $x_3:=x_3+b_3(x_a-x_3)+a_{31}(x_1-x_3)$
       $+a_{32}(x_2-x_3)+c_3+\mathbf{w}_3$;
od
\end{lstlisting}
\end{lrbox}

\lstset{language=prog}
\lstset{tabsize=3}
\newsavebox{\hybridd}
\begin{lrbox}{\hybridd}
\begin{lstlisting}[mathescape]
while $k\le N$ do
  $q_{last,1}=q_1$;$q_{last,2}=q_2$;$q_{last,3}=q_3$;$q_{last,4}=q_4$;
  $q_1\sim T_{q,1}$;$q_2\sim T_{q,2}$;$q_3\sim T_{q,3}$;$q_4\sim T_{q,4}$;
  $x_1:=x_1+b_1(x_a-x_1)+c_1\cdot\mathds{1}_{\mathbf{Q}_1}(q_{last,1})+\mathbf{w}_1$;
  $x_2:=x_2+b_2(x_a-x_2)+c_2\cdot\mathds{1}_{\mathbf{Q}_2}(q_{last,2})+\mathbf{w}_2$;
  $x_3:=x_3+b_3(x_a-x_3)+c_3\cdot\mathds{1}_{\mathbf{Q}_3}(q_{last,3})+\mathbf{w}_3$;
  $x_4:=x_4+b_4(x_a-x_4)+c_4\cdot\mathds{1}_{\mathbf{Q}_4}(q_{last,4})+\mathbf{w}_4$;
  $k:=k+1$
od
\end{lstlisting}
\end{lrbox}

\begin{figure}[htbp]
\begin{minipage}{0.4\textwidth}
\usebox{\hybrida}
\caption{Single-Room Heating}
\label{fig:hybrida}
\end{minipage}\hfill
\begin{minipage}{0.6\textwidth}
\usebox{\hybridb}
\caption{Double-Room Heating}
\label{fig:hybridb}
\end{minipage}
\end{figure}

\begin{example}[Stochastic Gradient Descent]\label{ex:motivsgd}
The most widely-used method in machine learning is \emph{stochastic gradient descent} (SGD).
The general form of an SGD algorithm is illustrated in Figure~\ref{fig:example4} on Page~\pageref{fig:example4}.
In the figure, an SGD algorithm with $n$ training data is modelled as a simple while loop, where
(i) $\mathsf{unif}[1,\ldots,n]$ is a sampling variable whose value is sampled uniformly from $1,2,\dots, n$, (ii) $\mathbf{w}$ is a vector of program variables that represents parameters to be learned,
(iii) $i$ is a program variable that represents the sampled index of the training data,
and (iv) $\gamma$ is a positive constant that represents the \emph{step size}.
The symbol $\nabla$ represents the \emph{gradient}, while each $G_i$ ($1\le i\le n$) is the loss function for the $i$th training data. By convention, the total loss function $G$ is given as the expected sum of all $G_i$'s, i.e., $G:=\frac{1}{n}\sum_i G_i$.
At each loop iteration, a data $i$ is chosen uniformly from all $n$ training data and the parameters in $\mathbf{w}$ are adjusted by the product of the step size and the gradient of the $i$th loss function $G_i$.
The loop guard $\Phi$ can either be practical so that a fixed number of iterations is performed (as
is analyzed in existing approaches~\cite{DBLP:conf/icml/HardtRS16,DBLP:journals/pacmpl/BartheEGHS18,DBLP:conf/atva/HuangWM18}), or the local criteria that the magnitude  ${\parallel}{\nabla G}{\parallel}_2$
of the gradient of the total loss function $G$  is small enough, or the global criteria that the value of $G$ is small enough. In this paper, we consider the global criteria, \ie, the loop guard is of the form $G(\mathbf{w})\ge \zeta$ where $\zeta$ is the threshold for ``small enough''.
Note that the SGD algorithm with the global criteria has randomized loop iterations which depends on the initial parameters.
\end{example}

\section{Proving Expected Sensitivity for Non-expansive Simple Loops}\label{sect:simplecase}

In this section, we demonstrate a sound approach for proving expected sensitivity over non-expansive simple while loops, whose number of loop iterations is randomized and depends on the input program valuation.
The main difficulty is that when the number of loop iterations depends on both the randomized execution and the input program valuation, the executions from two close-by input program valuations may be \emph{non-synchronous} in the sense that they do not terminate at the same time.
The following example illustrates this situation.

\begin{example}[Non-synchronicity]\label{exx:running}
Consider our running example in Figure~\ref{fig:runningexample1}, where the sampling variable $r$ observes the Dirac distribution such that $\probm(r=1)=1$, so that the program is completely deterministic.
Choose the initial inputs $x^*_1,x^*_2$ by setting $x^*_1=1-\epsilon$ and $x^*_2=1+\epsilon$, where $\epsilon>0$ can be sufficiently small.
Since we add $1$ to the value of $x$ in each loop iteration, the output value $x^\mathrm{out}_2$ under the input $x^*_2$ equals $1000+\epsilon$, while at the same step
the execution from $x^*_2$ stops,
the execution from $x^*_1$ does not terminate as the corresponding value is $1000-\epsilon$. Note that the final output from $x^*_1$ is $1001-\epsilon$.
\end{example}

The non-synchronicity prevents us from inferring the total expected sensitivity from the local sensitivity incurred in each loop iteration. To address this issue, we explore a martingale-based approach. In previous results such as~\cite{DBLP:journals/toplas/ChatterjeeFNH18,SriramCAV}, martingales have been successfully applied to prove termination properties of probabilistic programs. Besides qualitative termination properties, martingales can also derive tight quantitative upper/lower bounds for expected termination time and resource usage~\cite{DBLP:journals/toplas/ChatterjeeFNH18,DBLP:conf/pldi/NgoC018,ijcai18,DBLP:conf/pldi/Wang0GCQS19}. In this paper, we utilize the quantitative feature of martingales to bound the difference caused by non-synchronous situations.

We first recall the notion of \emph{ranking-supermartingale maps} (RSM-maps), a core notion in the application of martingale-based approaches to probabilistic programs. As we consider simple while loops as the basic building block of probabilistic programs, we present a simplified version for simple while loops.
Below we fix a simple while loop $Q$ in the form (\ref{eq:swl}). 

\begin{definition}[RSM-maps~\cite{SriramCAV,DBLP:journals/toplas/ChatterjeeFNH18}]\label{def:RSMmaps1}
A \emph{ranking-supermartingale map} (RSM-map) is a Borel-measurable function $\eta: \Rset^{|\pvars|}\rightarrow \Rset$ 
such that there exist real numbers $\epsilon>0, K\le 0$ satisfying the following conditions:
	\begin{compactitem}
        \item[(A1)] $\forall\pv:\big(\pv\models \Phi \Rightarrow \eta(\pv)\geq 0\big)$;
        \item[(A2)] $\forall\pv\,\forall\loc\,\forall\sv:\big((\pv\models \Phi \wedge \updf(\loc,\pv,\sv)\not\models\Phi) \Rightarrow K\leq \eta(\updf(\loc,\pv,\sv))\leq 0\big)$;
		\item[(A3)] $\forall\pv:\big(\pv\models \Phi\Rightarrow \expv_{\sv,\loc}(\eta(\updf(\loc,\pv,\sv)))\leq \eta(\pv)-\epsilon\big)$;
	\end{compactitem}
where $\expv_{\sv,\loc}(\eta(\updf(\loc,\pv,\sv)))$ is the expected value of $\eta(\updf(\loc,\pv,\sv))$ such that $\pv$ is treated as a constant vector and $\sv$ (resp. $\loc$) observes the joint probability distributions of sampling variables (resp. the probabilities of the probabilistic branches), respectively.
\end{definition}

Informally, (A1) specifies that the RSM-map should be non-negative before program termination, (A2) specifies the condition at loop termination, and (A3) specifies the \emph{ranking} condition that the expected value of the RSM-map should decrease (by the positive amount $\epsilon$) after each loop iteration.

The existence of an RSM-map provides a finite upper bound on the expected termination time of a probabilistic program~\cite{DBLP:journals/toplas/ChatterjeeFNH18} (see Theorem~\ref{thm:rsmmaps1} in Appendix~\ref{app:simplecase}).
In this way, an RSM-map controls the randomized number of loop iterations.
However,
simply having an upper bound for the expected termination time is not enough,
as what we need to bound is the difference between the expected values in non-synchronous situations.
To resolve the non-synchronicity, we need some additional conditions.
The first is the \emph{bounded-update} requiring that the value-change
in one loop iteration is bounded.
The second is the \emph{RSM-continuity} specifying that the RSM-map should be Lipschitz continuous over the loop guard.
Below we fix  a metric $\dist$.

\begin{definition}[Bounded Update $d$]\label{def:bu1}
We say that a simple while loop $Q$ has \emph{bounded update}
if there exists a real constant $d\ge 0$ such that
\begin{compactitem}
\item[(B2)] $\forall\loc\,\forall \pv\,\forall \mathbf{r} :\big(\pv\models \Phi\Rightarrow \dist(\pv,\updf(\loc,\pv,\sv)) \leq d\big)$.
\end{compactitem}
\end{definition}
The bounded-update condition simply bounds the change of values
during each loop iteration.
This condition is standard as it comes from the ``$\sigma$-finiteness'' proposed in the analysis of stochastic gradient descent~\cite[Definition~2.4]{DBLP:conf/icml/HardtRS16}. 

\begin{definition}[RSM-continuity $M$]\label{def:LipconRSM1}
An RSM-map $\eta$ has \emph{RSM-continuity} if there exists a real constant $M>0$ such that
\begin{compactitem}
\item[(B3)] $\forall \pv,\pv':|\eta(\pv)-\eta(\pv')|\leq M\cdot \dist(\pv,\pv')$.
\end{compactitem}
\end{definition}
By definition, the RSM-continuity bounds the difference of the RSM-map value proportionally in the metric $\dist$ 
when the program valuations $\pv,\pv'$ are close.
This condition is used to bound the difference in non-synchronous situations and is naturally satisfied if the RSM-map is linear.
Although this condition seems a bit restrictive, latter we will show that it can be relaxed (see Remark~\ref{rmk:rsmc1} and Remark~\ref{rmk:rsmc2}).
We keep the condition in its current form for the sake of brevity.

Below we first present our result for affine sensitivity, then linear sensitivity.
We fix a metric $\dist$.

\subsection{Proving Expected Affine-Sensitivity}

The main result for proving expected affine-sensitivity of non-expansive simple loops is
as follows. 

\begin{theorem}\label{thm:affine}
A non-expansive simple while loop $Q$ in the form~(\ref{eq:swl}) is expected affine-sensitive
over its loop guard $\Sat{\Phi}$ if we have that
\begin{compactitem}
\item
$Q$ has bounded update,
and
\item
there exists an RSM-map for $Q$ that has RSM-continuity.
\end{compactitem}
In particular, we can choose $\theta=\infty$ and $A=2\cdot\frac{d\cdot M+\epsilon}{  \epsilon\cdot D_1}, B=-2\cdot\frac{d\cdot K}{\epsilon\cdot D_1}$
in (\ref{eq:expaffsen}),
where the parameters
$d,M,\epsilon,K,D_1$ are from Definition~\ref{def:RSMmaps1}, Definition~\ref{def:bu1}, Definition~\ref{def:LipconRSM1} and (\ref{metric}).
\end{theorem}
\begin{proof}[Proof Sketch]
Choose any program variable $z$.
Let $d$ be a bound
from Definition~\ref{def:bu1}, and $\eta$ be an RSM-map with the parameters $\epsilon, K$ from Definition~\ref{def:RSMmaps1} that has RSM-continuity with a constant $M$ from Definition~\ref{def:LipconRSM1}.
Consider input program valuations $\pv,\pv'$ such that $\pv,\pv'\models\Phi$.
Let $\delta := \dist(\pv,\pv')$.
Denote by $T_{\pv''}$ (resp. $Z_{\pv''}$) the random variable for the number of loop iterations (resp. the value of $z$ after the execution of $Q$) from an input program valuation $\pv''$, respectively.
Define $\mathbf{W}_{\pv''}$ as the vector of random variables that represents the program valuation after the execution of $Q$, starting from $\pv''$.
We illustrate the main proof idea through clarifying the relationships between any runs $\omega=\{\pv_n\}_{n\ge 0}$, $\omega'=\{\pv'_n\}_{n\ge 0}$ that start from respectively $\pv,\pv'$ (i.e., $\pv_0=\pv$ and $\pv'_0=\pv'$) and follow the \emph{same} probabilistic branches and sampled values in every loop iteration.
Consider at a step $n$ the event $\min\{T_{\pv}, T_{\pv'}\}\ge n$ holds (i.e., both the executions do not terminate before the $n$th loop iteration). We have the following cases:
\begin{compactitem}
\item[{\em Case 1.}]
	Both $\pv_n$ and $\pv'_n$ violate the loop guard $\Phi$, i.e., $\pv_n,\pv'_n\models \neg\Phi$.
	This case describes that the loop $Q$ terminates exactly after the $n$th loop iteration for both the executions. From the non-expansiveness, we obtain directly that
	$\dist(\pv_n,\pv'_n)\leq \delta$. Hence $|\pv_n[z]-\pv'_n[z]|\le \frac{\dist(\pv_n,\pv'_n)}{D_1}\le \frac{\delta}{D_1}$.
\item[{\em Case 2.}] Exactly one of $\pv_n, \pv'_n$ violates the loop guard $\Phi$. This is the non-synchronous situation that needs to be addressed through martingales.
W.l.o.g., we can assume that $\pv_n\models\Phi$ and $\pv'_n\models \neg\Phi$. From the upper-bound property of RSM-maps (see Theorem~\ref{thm:rsmmaps1} in Appendix~\ref{app:simplecase}), we derive that $\expv_{\pv_n}(T_{\pv_n})\leq\dfrac{\eta(\pv_n) -K}{\epsilon}$.
From the bounded-update condition (B2) and the triangle inequality of metrics, we have that $|\pv_n[z]-Z_{\pv_n}|\le \frac{1}{D_1}\cdot\dist(\pv_n,\mathbf{W}_{\pv_n})\le \frac{d}{D_1}\cdot T_{\pv_n}$.
Thus, we obtain that
\begin{equation}
|\expv_{\pv_n}(Z_{\pv_n})-\pv_n[z]| \leq \expv_{\pv_n}(|\pv_n[z]-Z_{\pv_n}|) \leq  \expv_{\pv_n}\left(\frac{d}{D_1}\cdot T_{\pv_n}\right)   \leq \frac{d}{D_1}\cdot \dfrac{\eta(\pv_n) -K}{\epsilon}\enskip.
\end{equation}
By the non-expansiveness, we have $\dist(\pv_n,\pv'_n)\leq \delta$. Then by the RSM-continuity (B3), we have $|\eta(\pv_n)-\eta(\pv'_n)|\leq M\cdot \delta$.
    Furthermore, from  (A2) we have $\eta(\pv'_n)\le 0$. So we obtain that $\eta(\pv_n)\leq M\cdot \delta$.
It follows that
\begin{eqnarray*}
|\expv_{\pv_n}(Z_{\pv_n})-\expv_{\pv'_n}(Z_{\pv'_n})| &=& |\expv_{\pv_n}(Z_{\pv_n})-\pv'_n[z]|  \\
&\leq &  |\expv_{\pv_n}(Z_{\pv_n})-\pv_n[z]|+ |\pv_n[z]-\pv'_n[z] |  \\
&\leq & \frac{d}{D_1}\cdot\dfrac{M\cdot  \delta-K}{\epsilon}+ \frac{\delta}{D_1} = \frac{d\cdot M+\epsilon}{\epsilon\cdot D_1}\cdot   \delta- \frac{d\cdot K}{\epsilon\cdot D_1}\enskip.
\end{eqnarray*}
\item[{\em Case 3.}] Neither $\pv_n$ nor $\pv'_n$ violates the loop guard $\Phi$. In this case, the loop $Q$ will continue from both $\pv_n$ and $\pv'_n$.
Then in the next iteration, the same analysis can be carried out for the next program valuations $\pv_{n+1},\pv'_{n+1}$, and so forth.
\end{compactitem}
From Theorem~\ref{thm:rsmmaps1} (in Appendix~\ref{app:simplecase}),
the probability that the third case happens infinitely often equals zero.
Thus, the sensitivity analysis eventually reduces to the first two cases.
In these two cases, the difference contributed to the total expected sensitivity $|\expv_{\pv}(Z_{\pv})-\expv_{\pv'}(Z_{\pv'})|$ when one of the runs terminates
after the $n$th loop iteration is at most
$\probm(T_\pv=n\vee T_{\pv'}=n)\cdot \left(\frac{d\cdot M+\epsilon}{  \epsilon\cdot D_1}\cdot   \delta- \frac{d\cdot K}{\epsilon\cdot D_1}\right)$
which is no greater than
$(\probm(T_\pv=n)+ \probm(T_{\pv'}=n))\cdot \left(\frac{d\cdot M+\epsilon}{  \epsilon\cdot D_1}\cdot   \delta- \frac{d\cdot K}{\epsilon\cdot D_1}\right)$.
Then by a summation over all $n$'s, we derive the desired result that
$|\expv_{\pv}(Z_{\pv})-\expv_{\pv'}(Z_{\pv'})| \le  A\cdot\delta+ B$
where $A:=2\cdot\frac{d\cdot M+\epsilon}{  \epsilon\cdot D_1}$ and $B:=-2\cdot\frac{d\cdot K}{\epsilon\cdot D_1}$.
The detailed proof requires an explicit representation of the expected values through Lebesgue integral (see Appendix~\ref{app:simplecase}).
In particular, the integral representation allows us to consider the \emph{same} probabilistic branches and sampled values in each loop iteration.
Another subtle point is that the integrability of the random variables $Z,Z'$ in (\ref{eq:expaffsen}) are guaranteed by the bounded-update condition and finite expected termination time from Theorem~\ref{thm:rsmmaps1}.
\end{proof}

\begin{remark}
Although in the statement of Theorem~\ref{thm:affine} we do not bound the constant $B$, the result is non-trivial as it applies to all program valuations that satisfies
the loop guard.
This is because input program valuations in the satisfaction set may lead to unbounded expected outcome as the expected number of loop iterations depend on input program valuations. Thus, simply raising the value of $B$ does not suffice to bound the unbounded expected outcomes.
\end{remark}

\begin{remark}\label{rmk:lpg}
In Theorem~\ref{thm:affine} we only consider the sensitivity over the loop guard.
The reason is that since we consider loops with randomized and input-dependent loop iterations, sensitivity usually applies to the loop guard only.
Consider the (nonprobabilistic) loop in Figure~\ref{fig:loopguardcounterexample} on Page~\pageref{fig:loopguardcounterexample}. We can construct an RSM-map $\eta$ by $\eta(x,y):=y$ with parameters $\epsilon=1, K=-1$. Then by Theorem~\ref{thm:affine} (where we set $d=M=1$ and $\dist$ to be the max-norm) we derive that the loop is expected affine-sensitive
over its loop guard.
However, it is straightforward to observe that if we choose two input program valuations $\pv,\pv'$ such that $\pv[x]=n$, $\pv'[x]=n-\epsilon$ and $\pv[y]=\pv'[y]=n$, where $n$ is a natural number that can be arbitrarily large and $\epsilon$ is a positive real number that can be arbitrarily close to zero (so that $\pv\models\Phi$ and $\pv'\not\models\Phi$), then the expected affine-sentivity does not hold.
The reason is that the execution from $\pv$ enters the loop and ends with $-1$ for $y$, and that from $\pv'$ does not enter the loop and hence keeps its input value $n$.
\end{remark}

\begin{example}[The Running Example]\label{ex:running}
Consider our running example in Figure~\ref{fig:runningexample1}.
We choose the sampling variable $r$ to observe the Bernoulli distribution $\probm(r=0)=\probm(r=1)=\frac{1}{2}$ and
the metric $\dist$ as the max norm. 
Then we can construct an RSM-map $\eta(x)=1000-x$ with $\epsilon=\frac{1}{2}, K=-1$.
Moreover, the RSM-map $\eta$ has the RSM-continuity with $M=1$, and the loop has bounded update with $d=1$.
Hence by Theorem~\ref{thm:affine}, the loop is expected affine-sensitive
over its loop guard.
\end{example}

\begin{example}[Mini-roulette]\label{ex:roulette2}
We show that the Mini-roulette example in Figure~\ref{fig:example1}(left) is expected affine-sensitive in the program variable $w$ over its loop guard.
To show this, we construct the function $\eta(x,w)=13\cdot x-13$ with $\epsilon=1, K=-13$.
We also clarify the following points.
\begin{compactenum}
\item For any values $x_1,x_2$ to the program variable $x$ before a loop iteration and any $\loc\in\mathbf{L}$ that resolves the probabilistic branches, we have that $|(x_1+a)-(x_2+a)|=|x_1-x_2|$ after the loop iteration where the value of $a$ is determined by the probabilistic branch (i.e for branch $5$, $a=11$).
    The same applies to the program variable $w$. Thus the loop is non-expansive.
\item All increments to $x$ and $w$ are bounded, hence the loop has bounded update, which ensures (B2).
\item The loop guard $x\ge 1$ implies $\eta(x,w)=13\cdot x-13\ge 0$, thus (A1) is satisfied. When $x\ge 1$, $\loc
\in\mathbf{L}$ and $\updf(\loc, (x,w), -)<1$, we have $-13\le \eta(\updf(\loc, (x,w), -))\le 0$, ensuring (A2). When $x\ge 1$, $\loc
\in\mathbf{L}$, we have $\expv_{\loc}(\eta(\updf(\loc,(x,w),-)))\leq \eta(x,w)-1$.
Thus $\eta$ is an RSM-map.
\item Given any values $x_1,x_2\ge 1$ and $w_1,w_2$ to the program variables $x,w$, we have $|\eta(x_1,w_1)-\eta(x_2,w_2)|=13\cdot |x_1-x_2|$. Thus $\eta$ has \emph{RSM-continuity}.
\end{compactenum}
By Theorem \ref{thm:affine}, we obtain that the program is expected affine-sensitive over its loop guard.
\end{example}

\begin{remark}\label{rmk:rsmc1}
For proving expected affine sensitivity, one can relax the RSM-continuity to the condition that
$\exists C>0.\forall \pv,\pv'. \left[(\updf(\loc,\pv,\sv)\models\Phi\wedge \updf(\loc,\pv',\sv)\not\models\Phi)\Rightarrow \eta(\pv)\le C\right]$,
so that the difference in non-synchronous situations is guaranteed to be bounded by $C$.
\end{remark}

\begin{remark}[Conditional Branches]\label{rmk:condbranch}
In certain scenarios, it is possible to extend Theorem~\ref{thm:affine} to conditional branches at which it may happen that one program valuation satisfies the condition of the branch and another close-by valuation does not.
Consider a scenario where we are to use piecewise-linear functions to approximate a complex loop body.
If the approximation is sufficiently strong to ensure that neighbouring pieces of functions behave similarly, then our approach can handle the scenario as follows.
We first relax the Lipschitz continuity in Definition~\ref{def:LipconLB1} to ``$\le L\cdot\mathfrak{d}(\pv,\pv')+V$'' where $V$ is a nonnegative constant that bounds the difference between neighbouring pieces of functions in each loop iteration (which is small for sufficiently strong approximation).
This ensures that the final sensitivity would be
$A\cdot\mathfrak{d}(\pv,\pv')+(B+V\cdot \expv(T))$ in the non-expansive case where $\expv(T)$ is the expected termination time of the loop (that depends on the initial input).
\end{remark}

\subsection{Proving Expected Linear-Sensitivity}\label{sect:avsen}

To develop a sound approach for proving expected linear-sensitivity, one possible way is to extend the approach for expected affine-sensitivity.
However, simply extending the approach is not correct, as is shown by the following example.


\begin{example}\label{ex:linearcounter}
Consider our running example in Figure~\ref{fig:runningexample1}. We first consider that the sampling variable $r$ observes the Dirac distribution such that $\probm(r=1)=1$, the same as in Example
~\ref{exx:running}.
By choosing the same initial values $x^*_1$ and $x^*_2$ from Example~\ref{exx:running}, we have that the outcomes satisfy $|x^\mathrm{out}_1-x^\mathrm{out}_2|=1-2\cdot\epsilon$.
Hence, we could not find a constant $A$ such that $|x^\mathrm{out}_1-x^\mathrm{out}_2|\le A\cdot |x^*_1-x^*_2|=2\cdot A\cdot \epsilon$ when $\epsilon\rightarrow 0$.
Similar situation happens even if we have non-Dirac discrete probability distributions. For example, consider now that the sampling variable $r$ observes the distribution such that $\probm(r=0)=\probm(r=1)=0.5$.
Then with the same initial values $x^*_1$ and $x^*_2$, as the increment to the program variable $x$ is either $0$ or $1$, we have the same outputs $x^\mathrm{out}_1, x^\mathrm{out}_2$, refuting 
expected linear-sensitive.
\end{example}

The reason why we have such a situation in Example~\ref{ex:linearcounter} is again due to the non-synchronous situation where the number of loop iterations depends on the input program valuation.
While proving expected affine-sensitivity we can use a constant $B$ (cf. Definition~\ref{def:EAS}) to bound the difference caused by non-synchronous situations, in proving linear-sensitivity we need to set $B=0$, leading to a difficulty that cannot be resolved by the technique developed for affine-sensitivity.
To address this issue, we introduce
another Lipschitz continuity
w.r.t a given metric $\dist$.

\begin{definition}[Lipschitz Continuity in Next-step Termination $L'$]\label{def:LipconT1}
We say that a simple while loop $Q$ in the form (\ref{eq:swl}) is
\emph{Lipschitz continuous in next-step termination}
if there exists a constant $L'>0$ such that
\begin{compactitem}
\item[(B4)] $\forall\loc\,\forall \pv,\pv': \left(\pv,\pv' \models \Phi\Rightarrow \probm_{\sv}(\updf(\loc,\pv,\sv)\models \Phi\wedge \updf(\loc,\pv',\sv)\models \neg\Phi) \leq L'\cdot \dist(\pv,\pv') \right)$
\end{compactitem}
where given the program valuations $\pv,\pv'$ before the loop iteration and the resolution $\loc$ for the probabilistic branches,
the value $\probm_{\sv}(\updf(\loc,\pv,\sv)\models \Phi\wedge \updf(\loc,\pv',\sv)\models \neg\Phi)$ is the probability regarding the sampled values that after one loop iteration we have $\updf(\loc,\pv,\sv)$ can still enter the loop, while $\updf(\loc,\pv',\sv)$ violates the loop guard.
\end{definition}

The condition (B4) specifies that when the program valuations $\pv,\pv'$ are close, the probability that after the current loop iteration one of them stays in the loop while the other jumps out of the loop is
small as it is proportional to the distance between $\pv,\pv'$.
This condition handles the non-synchronous situation in the sense that the probability of non-synchronous situations is bounded linearly by the distance between the program valuations before a loop iteration.
For simple while loops with only \emph{discrete} probability distributions, this condition is usually not met.
This is because in discrete probability distributions there often exists a vector $\sv$ of sampled values with a minimum probability $p>0$ that $\updf(\loc,\pv,\sv)\models\Phi$ and $\updf(\loc,\pv',\sv)\not\models\Phi$. In some cases such probability $p$ may be very large, e.g., in our running example (Example~\ref{ex:linearcounter}), with $\probm(r=1)=1$ we have that $\probm_{r}(\updf(\loc,999-\epsilon,r)\models \Phi\wedge \updf(\loc,999+\epsilon,r)\models \neg\Phi)=1$, and with $\probm(r=0)=\probm(r=1)=\frac{1}{2}$ we have that
the same probability is
$\frac{1}{2}$,
when $\epsilon\rightarrow 0$.
In contrast, loops with \emph{continuous} distributions often satisfy this condition. For example, consider again our running example where $r$ now observes the uniform distribution over the interval $[0,1]$. Then for any initial values $x''\le x'\le 1000$ for the program variable $x$, the probability that $x'+r>1000$ but $x''+r\le 1000$ equals the chance that the sampled value of $r$ falls in $(1000-x', 1000-x'']$, which is no greater than $|x-x'|$ as the probability density function of $r$ is $1$ over the interval $[0,1]$.

In the following, we show that a large class of \emph{affine} simple while loops with continuous distributions guarantees the (B4) condition.
Below we say that an update function $\updf$ is \emph{affine} if for all $\loc\in\mathbf{L}$, we have that $\updf(\loc, \pv,\sv)=\mathbf{B}\cdot \pv+ \mathbf{C}\cdot \sv +\mathbf{c}$ for constant matrices $\mathbf{B},\mathbf{C}$ and vector $\mathbf{c}$.
Moreover, a boolean expression $\Phi$ is said to be \emph{affine} if $\Phi$ can be equivalently rewritten into a disjunctive normal form (DNF) $\bigvee_{i\in\mathcal{I}} (\mathbf{A}_i\cdot\pv\le \mathbf{d}_i)$ with constant matrices $\mathbf{A}_i$ and vectors $\mathbf{d}_i$ so that for all program valuations $\pv$, we have $\pv\models\Phi$ iff the disjunctive formula $\bigvee_{i\in\mathcal{I}} (\mathbf{A}_i\cdot \pv\le \mathbf{d}_i)$ holds.
The class of simple while loops that guarantees (B4) is as follows.

\begin{lemma}
\label{le:LipconT1}
Consider a simple while loop $Q$ in the form (\ref{eq:swl}) that satisfies the following conditions: 
\begin{compactenum}
\item both $\updf$ and $\Phi$ are affine and $\Phi$ is equivalent to some DNF $\bigvee_{i\in\mathcal{I}} (\mathbf{A}_i\cdot \pv\le \mathbf{d}_i)$;
\item all sampling variables are continuously-distributed whose probability density functions have bounded values;
\item for all $i\in\mathcal{I}$, $\loc\in\mathbf{L}$ and program valuations $\pv\models\Phi$, the coefficients for the sampling variables $\sv$ in $\mathbf{A}_i\cdot \updf(\loc, \pv, \sv)$  are not all zero at each row, \ie, the truth value of each disjunctive clause in $\Phi$ for $\updf(\loc,\pv,\sv)$ depends on $\sv$ at every row.
\end{compactenum}
Then the loop $Q$
is Lipschitz continuous in next-step termination
w.r.t any metric $\dist$.
\end{lemma}
Informally, the lemma guarantees the (B4) condition by requiring that (i) both the update function and the loop guard are affine, (ii) all sampling variables are continuously-distributed, and (iii) the truth value of every linear inequality in the loop guard after the current loop iteration depends on the sampled values.
The proof of Lemma~\ref{le:LipconT1} is elementary and is put in Appendix~\ref{app:avsen}.

\begin{remark}
We note that Lemma~\ref{le:LipconT1} serves as a sound condition only, and there are situations where the prerequsite of Lemman~\ref{le:LipconT1} fails but the condition (B4) still holds.
For example, consider that a program variable $x$ is assigned to $1$ every time in a loop iteration, and the loop guard involves the condition  $x\le 2$.
Then the condition $x\le 2$ does not affect the truth value of the loop guard since it is always satisfied, but has zero coefficients for all sampling variables.
While to derive weaker conditions is possible, in this work we consider Lemma~\ref{le:LipconT1} as a simple guarantee for ensuring (B4).
\end{remark}

Now we demonstrate our approach for proving expected linear-sensitivity.
Our first result is a sound approach for proving \emph{local} linear-sensitivity. Below given a program valuation $\pv\models\Phi$, a radius $\rho>0$ and a metric $\dist$, we denote by $U_{\Phi,\dist}(\pv,\rho)$ the neighbourhood $\{\pv'\mid \pv'\models\Phi\wedge \dist(\pv,\pv')\le\rho\}$.

\begin{proposition}\label{prop:linear}
A non-expansive simple while loop $Q$ in the form~(\ref{eq:swl}) has expected linear-sensitivity
over some neighbourhood $U_{\Phi,\dist}(\pv^*, \rho)$ of any given $\pv^*\in\Sat{\Phi}$ if $Q$ has (i) bounded update,
(ii) an RSM-map with RSM-continuity and (iii) the \emph{Lipschitz continuity in next-step termination}.
\end{proposition}
The proof resembles the one for expected affine-sensitivity (Theorem~\ref{thm:affine}). The obtained
the coefficient $A$ (cf. Definition~\ref{def:EAS}) depends on the expected termination time from the input program valuation $\pv^*$. See Appendix~\ref{app:avsen} for details.

Proposition~\ref{prop:linear} gives a sound approach for proving \emph{local} linear-sensitivity in that the coefficient $A$
only works for a small neighbourhood of a given input.
A natural question arises whether we can obtain \emph{global} linear-sensitivity so that the coefficient $A$ works for all program valuations that satisfy the loop guard.
The major barrier in the proof of Proposition~\ref{prop:linear} to obtain global linear sensitivity
is that in general we can only treat every program valuation uniformly, without distinguishing between program valuations with large and small RSM-map values.
To overcome this difficulty, we partition program valuations into \emph{finitely} many classes so that each class shares a common coefficient, but different classes may have different coefficients. Based on the partition, we utilize the inter-relationship between different classes to prove the existence of a collection of coefficients for the expected linear-sensitivity.
The partition relies on the \emph{difference-bounded} condition for RSM-maps proposed for proving concentration properties of termination time~\cite{DBLP:journals/toplas/ChatterjeeFNH18}.


\begin{definition}[The Difference-bounded Condition~\cite{DBLP:journals/toplas/ChatterjeeFNH18}]\label{def:diffbounded}
We say that an RSM-map $\eta$ (for a simple while loop) is \emph{difference-bounded} if there exists a constant $c\ge 0$ such that
\begin{compactitem}
\item[(A4)] $\forall\pv\,\forall\loc\,\forall \sv: \left(\pv\models \Phi \Rightarrow |\eta(\updf(\loc,\pv,\sv))-\eta(\pv)|\le c\right)$\enskip.
\end{compactitem}
\end{definition}
The difference-bounded condition ensures the concentration property for program termination~\cite{DBLP:journals/toplas/ChatterjeeFNH18}, see Theorem~\ref{thm:rsmmaps} in Appendix~\ref{app:sect6}.
Below we demonstrate how this condition helps to partition the program valuations into finitely-many regions by their corresponding RSM-map values.
First we show that this condition derives a minimum positive probability that the value of an RSM-map decreases by a minimum positive amount.

\begin{lemma}\label{lemm:pd}
If $\eta$ is a difference-bounded RSM-map with the parameters $\epsilon,c$ specified in Definition~\ref{def:RSMmaps1} and Definition~\ref{def:diffbounded}, then there exists a constant $p\in (0,1]$ such that
\begin{compactitem}
\item[\emph{(\dag)}] $\forall\pv: \left(\pv\models\Phi\Rightarrow \probm_{\sv,\loc}(\eta(\updf(\loc,\pv,\sv))-\eta(\pv)\leq -\frac{1}{2}\cdot\epsilon)\ge p\right)$
\end{compactitem}
where the probability $\probm_{\sv,\loc}(-)$ is taken w.r.t the sampled valuation $\sv$ and the resolution $\loc$ for probabilistic branches, and treats the program valuation $\pv$ as constant.
In particular, we can take $p:=\frac{\epsilon}{2\cdot c-\epsilon}$.
\end{lemma}
\begin{proof}[Proof Sketch]
The proof is through Markov's inequality. See Appendix~\ref{app:avsen} for details.
\end{proof}

\noindent{\em The finite partition.}
Based on Lemma~\ref{lemm:pd}, we partition the satisfaction set $\Sat{\Phi}$ into finitely many regions.
Our aim is to have a partition $R_1,\dots, R_{n^*}, R_\infty$ based on which we find individual sensitivity coefficients on each $R_k$.
First we choose a smallest natural number $n^*$ such that $(n^*-1)\cdot \frac{1}{2}\cdot \epsilon\le c+1 < n^*\cdot  \frac{1}{2}\cdot \epsilon$. (Note that  $n^*\ge 3$ as $\epsilon\le c$.)
Then we define the region $R_k:=\{\pv\in \Sat{\Phi}\mid (k-1)\cdot\frac{1}{2}\cdot\epsilon\le \eta(\pv)< k\cdot\frac{1}{2}\cdot\epsilon\}$ for natural numbers $1\le k\le n^*$, and
$R_{\infty}:=\{\pv\in\Sat{\Phi}\mid \eta(\pv)\ge n^*\cdot \frac{1}{2}\cdot\epsilon\}$.
It follows that $\Sat{\Phi}$ is a disjoint union of all $R_k$'s. 
Especially, we treat $R_{\infty}$ as the region for program valuations with ``large enough'' RSM-map values.
After the partitioning, we are to prove that for each $R_k$ there is a coefficient $A_k$ for Definition~\ref{def:EAS}, and for different $R_k$'s there may be different $A_k$'s.
The following result presents the first step of the proof, where each $A_k$ ($1\le k\le n^*$) represents  the coefficient for $R_k$ and $A_\infty$ for $R_{\infty}$.

\begin{proposition}\label{prop:Ak}
For any natural number $n\ge 1$, real numbers $C,D\ge 0$ and probability value $p\in (0,1]$, the following system of linear inequalities (with real variables $A_k$'s ($0\le k\le n$) and $A_\infty$)
\begin{eqnarray*}
(1-p)\cdot A_\infty+ C+ p\cdot A_0  \le  A_1 \\
(1-p)\cdot A_\infty+ C+ p\cdot A_1  \le  A_2 \\
\qquad\qquad\qquad \vdots \qquad\qquad\qquad\\
(1-p)\cdot A_\infty+ C+ p\cdot A_{n-1}  \le  A_n \\
D=A_0\le A_1\le  \dots \le A_n \le A_\infty \\
\end{eqnarray*}
has a solution.
\end{proposition}

\begin{proof}[Proof Sketch]
First, we equate all inequalities but the last line above (i.e., $D=A_0\le A_1\le  \dots \le A_n \le A_\infty$), so that we directly get the solution as follows:
\begin{compactitem}
\item $A_\infty= \frac{1}{p^{n+1}}\cdot (\sum_{m=1}^{n+1} p^{m-1})\cdot C+D= \frac{1}{p^{n+1}}\cdot \frac{1-p^{n+1}}{1-p}\cdot C+D$,
\item $A_k= (\frac{1-p^{k}}{p^{n+1}}\cdot (\sum_{m=1}^{n+1} p^{m-1})+ (\sum_{m=1}^k p^{m-1}))\cdot C + D=  (\frac{1-p^{k}}{p^{n+1}}\cdot \frac{1-p^{n+1}}{1-p} +  \frac{1-p^{k}}{1-p})\cdot C + D$ for $1\le k\le n$.
\end{compactitem}
Then we check that the inequalities in the last line above hold.
See Appendix~\ref{app:avsen} for details.
\end{proof}

Now we state our main result for proving global linear sensitivity on non-expansive simple loops.

\begin{theorem}\label{thm:linear}
A non-expansive simple while loop $Q$ in the form~(\ref{eq:swl}) has expected linear-sensitivity
over its loop guard $\Sat{\Phi}$ if $Q$ has (i) bounded update,
(ii) a difference-bounded RSM-map with RSM-continuity and (iii) the
Lipschitz continuity in next-step termination.
In particular, we can choose $\theta=\frac{1}{M}$ in (\ref{eq:expaffsen}) where the parameter $M$ is from the RSM-continuity (Definition~\ref{def:LipconRSM1}).
\end{theorem}
\begin{proof}[Proof Sketch]
Choose any program variable $z$.
Denote by $T,T'$ (resp. $Z_n,Z'_n$) the random variables for the number of loop iterations (resp. the value of $z$  at the $n$-th step),
from two close-by input program valuations $\pv,\pv'$,
respectively.
For each natural number $n\ge 0$, we define $\delta_n(\pv,\pv'):=\expv_{\pv}(Z_{T\wedge n})-\expv_{\pv'}(Z'_{T'\wedge n})$,
where the random variable $T\wedge n$ is defined as $\min\{T,n\}$ and $T'\wedge n$ likewise. We also define $\delta(\pv,\pv'):=\expv_{\pv}(Z_{T})-\expv_{\pv'}(Z'_{T'})$.
First, we prove from the Dominated Convergence Theorem that  $\lim\limits_{n\rightarrow \infty}\delta_n(\pv,\pv')=\delta(\pv,\pv')$.
Second, given a difference-bounded RSM-map $\eta$ with RSM-continuity, we construct the regions $R_k$'s ($1\le k\le n^*$) and $R_\infty$ as in the paragraph below Lemma~\ref{lemm:pd}, and solve $A_k$'s and $A_\infty$ from Proposition~\ref{prop:Ak}.
Third, based on the solved $A_k$'s and $A_\infty$, we prove by induction on $n\ge 0$ that for all $k\in\{1,\dots, n^*,\infty\}$ and all program valuations $\pv,\pv'$, we have $\delta_n(\pv,\pv')\le A_k\cdot \dist(\pv,\pv')$ when $\pv,\pv'\models \Phi$ and $\pv\in R_k$.
In the inductive proof, (i) we apply Lemma~\ref{lemm:pd} to tackle the regions $R_k$ ($1\le k\le n^*$)
and use the inequality $(1-p)\cdot A_\infty+ C+ p\cdot A_{k-1} \le  A_k$ from Proposition~\ref{prop:Ak} to prove the inductive case, and (ii) for $R_\infty$ we ensure the fact that starting from two close-by program valuations in $R_\infty$, the loop will not terminate after the current loop iteration from both the program valuations, as is guaranteed by (A1), (A2), (A4) and the RSM-continuity.
Finally, the result follows from taking the limit $n\rightarrow\infty$ and the fact that we have finitely many regions.
The detailed proof is put in Appendix~\ref{app:avsen}.
\end{proof}

\begin{remark}\label{rmk:rsmc2}
For expected linear sensitivity, one can relax the RSM-continuity as follows.
First, we require the relaxed condition in Remark \ref{rmk:rsmc1} to tackle the non-synchronous situation.
Second, we need the relaxed condition $\exists D>0.\forall \pv,\pv'. \left[(\pv,\pv'\models \Phi\wedge \dist(\pv,\pv')\le \theta\wedge \eta(\pv)>D) \Rightarrow \eta(\pv')>c\right]$,
so that the neighourbood around $\pv$ w.r.t the threshold $\theta$ will not lead to termination after one loop iteration.
\end{remark}

\begin{example}\label{va-roulette}
We now show that the mini-roulette variant in Figure~\ref{fig:example1}(right) is expected linear-sensitive in the program variable $w$ over its loop guard.
To show this, we construct the function $\eta(x,w)=2.45\cdot x-2.45$ with $\epsilon=1, K=-4.91$.
We also clarify the following points.
\begin{compactenum}
\item For any values $x_1,x_2$ to the program variable $x$ before a loop iteration and any $\loc\in\mathbf{L}$ that resolves the probabilistic branches, we have that $|(x_1+r_i)-(x_2+r_i)|=|x_1-x_2|$ after the loop iteration where the value of $r_i$ is decided by the executed branch and its distribution(i.e for branch $5$, $r_i:=r_5\sim unif(8,9)$). The same applies to the program variable $w$. Thus the loop body is non-expansive.
\item All increments to $x$ and $w$ are bounded, hence the loop has bounded update, which ensures (B2).
\item The loop guard $x\ge 1$ implies $\eta(x,w)=2.45\cdot x-2.45\ge 0$, thus (A1) is satisfied. When $x\ge 1$, $\loc
\in\mathbf{L}$ and $\updf(\loc, (x,w), -)<1$, we have $-4.91\le \eta(\updf(\loc, (x,w), -))\le 0$, ensuring (A2). When $x\ge 1$, $\loc
\in\mathbf{L}$, we have $\expv_{\sv,\loc}(\eta(\updf(\loc,(x,w),\sv)))\leq \eta(x,w)-1$, ensuring (A3).
Thus, $\eta$ is an RSM-map.
\item Given any values $x_1,x_2\ge 1$ and $w_1,w_2$ to the program variables $x,w$, we have $|\eta(x_1,w_1)-\eta(x_2,w_2)|=2.45\cdot |x_1-x_2|$. Thus $\eta$ has \emph{RSM-continuity}.
\item When $x\ge 1$, we have $|\eta(\updf(\loc,(x,w),\sv))-\eta(x,w)|\le |\eta(\updf(\loc,(x,w),r_5))-\eta(x,w)|= |2.45\cdot (x+r_5)-2.45-2.45\cdot x+2.45|\le 2.45\cdot 9=22.05$, ensuring (A4). Thus, $\eta$ is \emph{difference-bounded}.
\item Due to the fact that both the update function and the loop guard are affine, all sampling variables are bounded continuously-distributed, and the coefficients for the current sampling variables are not all zero in the loop guard of the next iteration, we can verify that the loop has the
    Lipschitz continuity in next-step termination
    by Lemma \ref{le:LipconT1}.
\end{compactenum}
Then by Theorem \ref{thm:linear}, we can conclude that this probabilistic program is expected linear-sensitive over its loop guard.
\end{example}

\section{Proving Expected Sensitivity for Expansive Simple While Loops}\label{sect:generalsensi}

In this section, we show how our sound approach for proving expected sensitivity of non-expansive loops can be extended to expansive simple while loops.
We first illustrate the main difficulty,
and then enhance RSM-maps to be difference-bounded and show how they can address the difficulty.

The main difficulty to handle expansive loops is that the difference between two program valuations may tend to infinity as the number of loop iterations increases.
For example, consider a simple while loop where at every loop iteration (i) the value of a program variable $z$ is tripled and (ii) the loop terminates  immediately after the current loop iteration with probability $\frac{1}{2}$.
Then given two different initial values $z',z''$ for $z$, we have that
\[
\textstyle \expv_{z'}(Z')-\expv_{z''}(Z'')=\sum_{n=1}^\infty \probm(T=n)\cdot 3^n\cdot |z'-z''|=\sum_{n=1}^\infty \left(\frac{3}{2}\right)^n\cdot |z'-z''|=\infty.
\]
where $Z',Z''$ are given by the same way of $Z,Z'$ as in~(\ref{eq:expaffsen}) and $T$ is the termination time random variable.
Thus the expected-sensitivity properties do not hold for this example, as the increasing speed of $z$ is higher than that for program termination.
To cope with this point, we consider again RSM-maps to be difference-bounded, as in Definition~\ref{def:diffbounded}.
The main idea is to use the exponential decrease from difference-bounded RSM-maps (Theorem~\ref{thm:rsmmaps} in Appendix~\ref{app:sect6}) to counteract the unbounded increase in the difference between input program valuations.

Below we illustrate the main result of this section. Recall that given a program valuation $\pv$, a radius $\rho>0$ and a metric $\dist$, we denote by $U_{\Phi,\dist}(\pv,\rho)$ the neighbourhood $\{\pv'\mid \pv'\models\Phi\wedge \dist(\pv,\pv')\le\rho\}$.


\begin{theorem}\label{thm:exp}
Consider a simple while loop $Q$ in the form (\ref{eq:swl})
that satisfies the following conditions:
\begin{compactitem}
\item the loop body $P$ is Lipschitz continuous with a constant $L$ specified in Definition~\ref{def:LipconLB1}, and has bounded update;
\item there exists a difference-bounded RSM-map $\eta$ for $Q$ with RSM-continuity and parameters $\epsilon, K, c$ from Definition~\ref{def:RSMmaps1} and Definition~\ref{def:diffbounded} such that $L<\mathrm{exp}({\frac{3\cdot\epsilon^2}{8\cdot c^2}})$.
\end{compactitem}
Then for any program valuation $\pv^*$ such that $\pv^*\models \Phi$ and $\eta(\pv^*)>0$, there exists a radius $\rho>0$ such that the loop $Q$ is expected affine-sensitive
over $U_{\Phi,\dist}(\pv^*,\rho)$.
In particular, we can choose in Definition~\ref{def:EAS} that
\begin{eqnarray*}
& & \textstyle A:= 2\cdot  A'\cdot L^N+2\cdot A'\cdot L^N\cdot \exp\left(-\frac{\epsilon\cdot \eta(\pv^*)}{8\cdot c^2}\right)\cdot\sum_{n=1}^\infty \left(L\cdot \exp\left(-\frac{3\cdot \epsilon^2}{8\cdot c^2}\right)\right)^{n}\\
& & \textstyle B:= 2\cdot B'+ 2\cdot B'\cdot \exp\left(-\frac{\epsilon\cdot \eta(\pv^*)}{8\cdot c^2}\right)\cdot \sum_{n=1}^\infty \exp\left(-\frac{3\cdot \epsilon^2}{8\cdot c^2}\cdot n\right)
\end{eqnarray*}
where $A'=\frac{d\cdot M+\epsilon}{D_1\cdot\epsilon}$, $B'=-\frac{d\cdot K}{D_1\cdot\epsilon}$ and $N=\lfloor 4\cdot\frac{\eta(\pv^*)}{\epsilon}\rfloor+1$, for which the parameters
$d,M,\epsilon,K,D_1$ are from Definition~\ref{def:RSMmaps1}, Definition~\ref{def:bu1}, Definition~\ref{def:LipconRSM1} and (\ref{metric}).
\end{theorem}

The proof resembles the one for Theorem~\ref{thm:affine} and compares $L$ with the exponential-decreasing factor $\mathrm{exp}({\frac{3\cdot\epsilon^2}{8\cdot c^2}})$, see Appendix~\ref{app:sect6} for the detailed proof. Note that in the statement of the theorem we do not care for $\theta$, this is because we have already restricted the threshold to the neighbourhood $U_{\Phi,\dist}(\pv^*,\rho)$.



Theorem~\ref{thm:exp} cannot be directly extended to linear sensitivity as the technique to derive linear sensitivity (e.g. Theorem~\ref{thm:linear}) requires non-expansiveness as an important prerequisite.
We leave a more detailed investigation of the expansive case (including the linear sensitivity) as a future work.

\noindent{\bf Summary of Prerequisites for Expected Sensitivity.} In Table~\ref{ex:requirements}, we summarize the prerequisites for expected sensitivity. The first column specifies the program type (i.e. non-expansive/expansive loops), the second column specifies the sensitivity type (i.e. expected affine/linear-sensitive) for the program, the third column specifies the related theorem for this expected sensitivity of the program, and the last column contains all the prerequisites of this expected sensitivity.

\begin{remark}
All our results cover the degenerate case where the number of loop iterations is fixed and bounded.
To see this, suppose that the number of loop iterations is fixed to be $n$, and there is a program variable $i$ that serves as the loop counter. Then we can choose $n-i$ as a (difference-bounded) RSM-map that is independent of the program variables other than $i$, so that our result for expected affine-sensitivity (Theorem~\ref{thm:affine}) holds directly for this degenerate case.
Furthermore, the condition (B4) is satisfied directly as the termination depends only on the loop counter $i$, so that our linear-sensitivity result (Theorem~\ref{thm:linear}) holds for the degenerate case;
for expansive loops, we even do not need to check whether $L<\exp(\frac{3\cdot\epsilon^2}{8\cdot c^2})$ in Theorem~\ref{thm:exp} as the number of loop iterations is bounded.
\end{remark}

\begin{table*}
\caption{Overview of Expected Sensitivity Results and Their Prerequisites}
\label{ex:requirements}
\centering
\begin{threeparttable}
\begin{tabular}{c|c|c|c}
\hline
\hline
Program & Type & Theorem &  Prerequisites \\
\hline
non-expansive & expected & Theorem~\ref{thm:affine}& basic prerequisites\tnote{*}  \\
& affine-sensitivity & & \\
\hline
non-expansive & expected   &Theorem~\ref{thm:linear} &  basic prerequisites\tnote{*}   \\
&linear-sensitivity & & additional prerequisites\tnote{**}  \\
\hline
 & expected  & &  $L<\mathrm{exp}({\frac{3\cdot\epsilon^2}{8\cdot c^2}})$ (Definition~\ref{def:LipconLB1})\\
expansive &affine-sensitivity &  Theorem~\ref{thm:exp} &  basic prerequisites\tnote{*}   \\
 &&&  difference-bounded condition (Definition~\ref{def:diffbounded})  \\
\hline
\hline
\end{tabular}
\begin{tablenotes}
\footnotesize
\item[*] basic:
bounded update (Definition~\ref{def:bu1}), RSM-map (Definition~\ref{def:RSMmaps1}), RSM-continuity (Definition~\ref{def:LipconRSM1})
\item[**] additional:
difference-bounded condition (Definition~\ref{def:diffbounded}),
Lipschitz continuity in next-step termination
(Definition~\ref{def:LipconT1})
\end{tablenotes}
\end{threeparttable}
\end{table*}

\section{Sequential Composition of Simple While Loops}\label{sect:seqcom}

In this section, we demonstrate the compositionality of our martingale-based approach for proving expected sensitivity of probabilistic programs. We follow the previous work~\cite{DBLP:journals/pacmpl/BartheEGHS18} to consider the sequential composition of probabilistic programs. We show that under the same side condition from ~\cite{DBLP:journals/pacmpl/BartheEGHS18}, our approach is compositional under sequential composition.

We first show that the compositionality under sequential composition does not hold in general. The main point is that if the output range of a preceding program $Q$ does not match the input range
over which the latter program $Q'$ is expected sensitive, then the global expected sensitivity for the sequential composition $Q;Q'$ may not hold.
A detailed example is as follows.

\begin{example}
Consider the sequential composition $Q=\mathbf{skip}; Q'$ where $Q'$ is the simple while loop from Remark~\ref{rmk:lpg}. We know that $\mathbf{skip}$ is expected sensitive over all input program valuations.
From Theorem~\ref{thm:affine} and Remark~\ref{rmk:lpg}, we have $Q'$ is expected affine-sensitive only over its loop guard.
Then the program $Q$ is not expected affine-sensitive over all input program valuations, as one can choose two input program valuations such that one satisfies the loop guard of $Q'$ but the other does not.
\end{example}

Thus, in order to ensure compositionality, we need to require that the output range of the preceding program should match the input sensitivity range of the latter program,
as is also required in ~\cite{DBLP:journals/pacmpl/BartheEGHS18}.
Under this side condition, we prove that our approach is compositional over sequential composition of non-expansive simple while loops.
Below for a probabilistic program $Q$, we denote by $\outp(Q)$ the set of all possible outcome program valuations after the execution of $Q$ under some input program valuation in
the satisfaction set of its loop guard.

\begin{theorem}
\label{thm:seqcom}
Consider a non-expansive simple while loop $Q$ with  bounded-update
and an RSM-map with RSM-continuity, and a general program $Q'$ that has expected  affine-sensitivity
over a subset $U$ of input program valuations with threshold $\theta$ in (\ref{eq:expaffsen}).
If $\outp(Q)\subseteq U$ and assuming integrability in (\ref{eq:expaffsen}), then the sequential composition $Q;Q'$ is expected affine-sensitive
over the satisfaction set of the loop guard of $Q$ with threshold $\theta$. 
\end{theorem}
\begin{proof}[Proof Sketch]
The proof is basically an extension to the previous proof for Theorem~\ref{thm:affine}. We consider the same three cases from the previous proof, and use the expected affine-sensitivity from $Q'$ to derive the new sensitivity coefficients. The detailed proof is put in Appendix~\ref{app:seqcom}.
\end{proof}

Theorem~\ref{thm:seqcom} presents a general compositional result where the program $Q'$ can be an arbitrary probabilistic program.
By extending the proof for Theorem~\ref{thm:linear} in a similar way as from Theorem~\ref{thm:affine} to Theorem~\ref{thm:seqcom}, we can also derive a compositional result for expected linear-sensitivity.
However, we now need to consider $Q'$ as a sequential composition of simple while loops and impose \emph{linearity} on the RSM-maps.
(An RSM-map is \emph{linear} if it can be expressed as a linear combination of program variables and a possible constant term.)
Then for a sequential composition $Q=Q_1;\dots ;Q_n$ of simple while loops,
we require the condition (\ddag) that
(i) each $Q_i$ having bounded update and a \emph{linear} RSM-map $\eta_i$ that witnesses its expected linear-sensitivity (i.e., that satisfies the conditions (A1) -- (A4).) and (ii) $\outp(Q_i)\subseteq \Sat{\Phi_{i+1}}$ for all $i$, where
$\Phi_{i+1}$ is the loop guard of $Q_{i+1}$.
By extending the proof for Theorem~\ref{thm:linear} (see Appendix~\ref{app:seqcom} for the detailed proof), we establish the following theorem.

\begin{theorem}
\label{thm:seqcomlinear}
Consider a non-expansive simple while loop $Q$ with loop guard $\Phi$ that has (i) bounded-update,
(ii) a difference-bounded linear RSM-map with RSM-continuity, and (iii) the 
Lipschitz continuity in next-step termination.
Then for any sequential composition $Q'$ of simple while loops that (a) satisfies the condition (\ddag) (defined right before the theorem) and (b) has expected linear-sensitivity
over a subset $U$ of input program valuations,
if $\Sat{\Phi}\cup\outp(Q)\subseteq U$, then the sequential composition $Q;Q'$ is expected linear-sensitive
over the satisfaction set of the loop guard of $Q$. 
\end{theorem}

By an iterated application of Theorem~\ref{thm:seqcom} and Theorem~\ref{thm:seqcomlinear} (i.e., loop-by-loop), we obtain directly the compositionality over sequential composition of non-expansive simple  while loops.

\begin{remark}[Expansive Loops]
Up till now we only consider non-expansive loops. The main issue arising from expansive loops is that the expected sensitivity is restricted to a small neighbourhood of a fixed input program valuation, and the sensitivity coefficients often depend on the input program valuation. 
More precisely, these coefficients may be exponential in general (see Theorem~\ref{thm:exp}). Thus, compositionality for expansive loops 
depends on the exact post probability distribution after the execution of the preceding loops.
To overcome this difficulty, new technique needs to be developed and we plan it as a future work.
\end{remark}

\begin{remark}[Comparison with~\cite{DBLP:journals/pacmpl/BartheEGHS18}]
A similar compositional result is established in~\cite[Proposition 4.3]{DBLP:journals/pacmpl/BartheEGHS18} for loops with a fixed number of loop iterations.
The approach is similar to ours as it also considers sequential composition and requires that the output range of the preceding loop should match the input sensitivity range of the latter loop, and
treats each individual program in the sequential composition separately.
The only difference is that they prove directly that their coupling-based sensitivity has the compositional property regardless of the detailed program structure, while our approach requires an explicit RSM-map for each loop.
This is however due to the fact that our approach considers the more complex situation that the loop iterations are randomized and depend on the input program valuation.
\end{remark}

\begin{remark}[Compositionality]
We would like to note that the level of compositionality depends on the \emph{side condition} in an approach.
Some authors insist that compositionality should require no side condition, while other authors allow side conditions~\cite{DBLP:conf/compos/KupfermanV97}.
Our approach, like the approach in~\cite{DBLP:journals/pacmpl/BartheEGHS18}, has the least side condition, as we only require that the output range of the preceding program matches the input sensitivity range of the latter.
Our result is also different from the original one in~\cite{DBLP:journals/pacmpl/BartheEGHS18} as in our case we need to tackle the non-trivial point of non-synchronicity (see Example~\ref{exx:running}).
\end{remark}

\newcommand{\fnames}{\mathit{F}}
\newcommand{\fn}[1]{\mathsf{#1}}

\newcommand{\Val}{\mbox{\sl Val}}
\newcommand{\PV}{V}
\newcommand{\RV}{U}
\newcommand{\ndc}{\ell}

\newcommand{\locs}{\mathit{L}}
\newcommand{\blocs}{\mathit{L}_{\mathrm{b}}}
\newcommand{\alocs}{\mathit{L}_{\mathrm{a}}}
\newcommand{\plocs}{\mathit{L}_{\mathrm{p}}}
\newcommand{\tlocs}{\mathit{L}_{\mathrm{t}}}
\newcommand{\Alocs}[1]{\mathit{L}_{\mathrm{A}}^{\mathsf{#1}}}
\newcommand{\Dlocs}{\mathit{L}_{\mathrm{nd}}}
\newcommand{\transitions}{{\rightarrow}}

\newcommand{\assgn}[2]{\left[#1/#2\right]}
\newcommand{\lin}{\loc_\mathrm{in}}
\newcommand{\lout}{\loc_\mathrm{out}}
\newcommand{\val}[1]{\mbox{\sl Val}_{#1}}
\newcommand{\samples}{\val{}^\mathrm{r}}
\newcommand{\id}{\mbox{\sl id}}

\newcommand{\dpd}{q}
\newcommand{\supp}[1]{{\mathrm{supp}}{\left(#1\right)}}

\newcommand{\condexpv}[2]{{\expv}{\left({#1}{\mid}{#2}\right)}}
\newcommand{\condvar}[2]{{\mbox{\sl Var}}{\left({#1}{\mid}{#2}\right)}}
\newcommand{\infruns}{\Lambda}

\newcommand{\last}[1]{{#1}{\downarrow}}
\newcommand{\sat}[1]{\langle #1 \rangle}
\newcommand{\monoid}{\mbox{\sl Monoid}}

\newcommand{\configs}{\mathcal{C}}
\newcommand{\stackelems}{\mathcal{E}}

\newcommand{\sampfunc}{\Upsilon}
\newcommand{\sampdpd}{\overline{\Upsilon}}

\newcommand{\enabled}[1]{\mathrm{En}(#1)}
\newcommand{\cyl}{\mbox{\sl Cyl}}

\newcommand{\POLYS}[1]{{\mathfrak{R}}{\left[#1\right]}}
\newcommand{\DNF}[1]{\mathrm{DNF}(#1)}
\newcommand{\MS}[1]{\mathrm{MS}(#1)}
\newcommand{\TRUE}{\mbox{\sl true}}
\newcommand{\FALSE}{\mbox{\sl false}}
\newcommand{\SAT}[1]{\mathsf{Sat}\left({#1}\right)}
\newcommand{\pre}{\mathrm{pre}}

\newcommand{\infval}{\mbox{\sl infval}}
\newcommand{\supval}{\mbox{\sl supval}}
\newcommand{\handelmanformat}{(\dagger)}

\newcommand{\pspace}{(\Omega,\mathcal{F},\probm)}
\newcommand{\setR}{\mathbb{R}}
\newcommand{\setN}{\mathbb{N}}
\newcommand{\setZ}{\mathbb{Z}}
\newcommand{\run}{\{ (\loc_n, \pv_n) \}_{n=0}^\infty}

\newcommand{\px}[1]{{\color{blue} Peixin}:~{\color{red}#1}}
\newcommand{\rd}[1]{{#1}}

\section{An Automated Approach through RSM-synthesis Algorithms} \label{sec:computational}


In this section, we describe an automated algorithm that, given a non-expansive probabilistic loop $Q$ in the form~(\ref{eq:swl}),
synthesizes an RSM-map  with extra conditions required for proving expected sensitivity.
We consider affine programs whose loop guard and update function are affine,
and linear templates for an RSM-map.
Our algorithm runs in polynomial time and reduces the problem of RSM-map synthesis to linear programming by applying Farkas' Lemma.
We strictly follow the framework from previous synthesis algorithms~\cite{DBLP:journals/toplas/ChatterjeeFNH18,SriramCAV,DBLP:conf/cav/ChatterjeeFG16,ijcai18,ChatterjeeNZ2017,DBLP:conf/pldi/Wang0GCQS19,DBLP:conf/atva/FengZJZX17}.
As our synthesis framework is not novel, we describe only the essential details of the algorithm.

We first recall Farkas' Lemma.

\smallskip
\begin{theorem}[Farkas' Lemma~\cite{FarkasLemma}]\label{thm:farkas}
Let $\mathbf{A}\in\mathbb{R}^{m\times n}$, $\mathbf{b}\in\mathbb{R}^m$, $\mathbf{c}\in\mathbb{R}^{n}$ and $d\in\mathbb{R}$.
Suppose that $\{\mathbf{x}\in\mathbb{R}^n\mid \mathbf{A}\mathbf{x}\le \mathbf{b}\}\ne\emptyset$.
Then
$\{\mathbf{x}\in\mathbb{R}^n\mid \mathbf{A}\mathbf{x}\le \mathbf{b}\}\subseteq \{\mathbf{x}\in\mathbb{R}^n\mid \mathbf{c}^{\mathrm{T}}\mathbf{x}\le d\}$
iff there exists $\mathbf{y}\in\mathbb{R}^m$ such that $\mathbf{y}\ge \mathbf{0}$, $\mathbf{A}^\mathrm{T}\mathbf{y}=\mathbf{c}$ and $\mathbf{b}^{\mathrm{T}}\mathbf{y}\le d$.
\end{theorem}
Intuitively, Farkas' Lemma transforms the inclusion problem of a nonempty polyhedron within a halfspace into a feasibility problem of a system of linear inequalities.
As a result, one can decide the inclusion problem in polynomial time through linear programming.

\noindent{\bf The RSM-synthesis Algorithm.} Our algorithm has the following four steps:
\begin{compactenum}
\item {\em Template.} The algorithm sets up a column vector $\mathbf{a}$ of $|\pvars{}|$ fresh variables and a fresh scalar variable $b$
such that the template for an RSM-map $\eta$ is
$\eta(\pv)=
\mathbf{a}^{\mathrm{T}}\cdot\pv+b$. Note that since we use linear templates, the RSM-continuity condition is naturally satisfied.
\item {\em Constraints on $\mathbf{a}$ and $b$.} The algorithm first encodes the condition (A1) for the template $\eta$ as the inclusion assertion
$\{\pv\mid \pv\models\Phi\}\subseteq \{\pv\mid \mathbf{c}_1^{\mathrm{T}}\cdot\pv\le d_1\}$
where $\mathbf{c}_1, d_1$ are unique linear combinations of unknown coefficients $\mathbf{a},b$ satisfying that $\mathbf{c}_1^{\mathrm{T}}\cdot\pv\le d_1\Leftrightarrow\eta(\pv)\geq 0$.
Next, the algorithm encodes the condition (A2) as the inclusion assertion
$\{(\pv,\sv)\mid \pv\models \Phi \wedge \updf(\loc,\pv,\sv)\not\models\Phi\} \subseteq \{(\pv,\sv)\mid K\leq \eta(\updf(\loc,\pv,\sv))\leq 0\}$
parameterized with $\mathbf{a},b,K$ for every $\ell\in \mathbf{L}$,
where $K$ is a fresh unknown constant.
Then the algorithm encodes (A3) as
$\{\pv\mid \pv\models\Phi\}\subseteq \{\pv\mid \mathbf{c}_2^{\mathrm{T}}\cdot\pv\le d_2\}$
where $\mathbf{c}_2, d_2$ are unique linear combinations of unknown coefficients $\mathbf{a},b$ satisfying that $\mathbf{c}_2^{\mathrm{T}}\cdot\pv\le d_2\Leftrightarrow  \expv_{\sv,\loc}(\eta(\updf(\loc,\pv,\sv)))\leq \eta(\pv)-\epsilon$.
The algorithm can also encode (A4) as
$\{(\pv,\sv)\mid \pv\models\Phi\}\subseteq \{(\pv,\sv)\mid |\eta(\updf(\loc,\pv,\sv))-\eta(\pv)|\le c\}$
parameterized with $\mathbf{a},b,c$ for every $\ell\in \mathbf{L}$,
where $c$ is a fresh unknown constant.
All the inclusion assertions (with parameters $\mathbf{a},b,K,\epsilon,c$) are grouped \emph{conjunctively} so that these inclusions should all hold.
\item {\em Applying Farkas' Lemma.} The algorithm applies Farkas' Lemma to all the inclusion assertions 
from the previous step and obtains a system of linear inequalities
with the parameters $\mathbf{a},b,K,\epsilon,c$, where we over-approximate all strict inequalities (with `$<$') by non-strict ones (with `$\le$').
\item {\em Constraint Solving.} The algorithm calls a linear programming (LP) solver on the linear program consisting of the system of linear inequalities generated in the previous step.
\end{compactenum}
Besides the RSM-synthesis, we guarantee the non-expansiveness either directly from the structure of the program or by manually inspection (note that it can also be verified automatically through SMT solvers on the first order theory of reals). We check the bounded-update condition by a similar application of Farkas' Lemma (but without unknown parameters), and the 
Lipschitz continuity in next-step termination
by Lemma~\ref{le:LipconT1}.
If the output of the algorithm is successful, i.e.~if the obtained system of linear inequalities is feasible, then the solution to the LP obtains concrete values for $\mathbf{a},b,K,\epsilon,c$ and leads to a concrete (difference-bounded) RSM-map $\eta$.

As our algorithm is based on LP solvers, we obtain polynomial-time complexity of our algorithm.

\begin{theorem}
Our RSM-synthesis algorithm has polynomial-time complexity. 
\end{theorem}

\begin{example}\label{al:simple}
Consider the mini-roulette example showed in Figure~\ref{fig:example1}(left) (Page~\pageref{fig:example1}).
\begin{compactenum}
\item The algorithm sets a linear template $\eta(x,w):=a_{1}\cdot x+a_2\cdot w+a_3$.
\item The algorithm encodes the conditions (A1)--(A3) as the inclusion assertions:
$$
\begin{matrix*}[l]
(A1) & \{(x,w)\mid x\ge 1\wedge w\ge 0\}\subseteq \{(x,w)\mid -a_1\cdot x-a_2\cdot w\le a_3\} \\
(A2) & \{(x,w)\mid x\ge 1\wedge w\ge 0\wedge x<2\}\subseteq \{(x,w)\mid K\le a_1\cdot (x-1)+a_2\cdot w+a_3\le 0\} \\
(A3) & \{(x,w)\mid x\ge 1\wedge w\ge 0\}\subseteq \{(x,w)\mid 0\le\frac{1}{13}a_1-\frac{4}{5}a_2-\epsilon\} \\
\end{matrix*}
$$
\item The algorithm applies Farkas' Lemma to all the inclusion assertions generated in the previous step and obtains a system of linear inequalities
involving the parameters $a_1,a_2,a_3,K,\epsilon$, where we over-approximate all strict inequalities (with `$<$') by non-strict ones (with `$\le$').
\item The algorithm calls a linear programming (LP) solver on the linear program consisting of the system of linear inequalities generated in the previous step.
\end{compactenum}
Finally, the algorithm outputs an optimal answer $\eta(x)=13\cdot x-13$ with $\epsilon=1, K=-13$(see Example~\ref{ex:roulette2}). We can verify this $\eta$ is an RSM-map with RSM-continuity. Due to the fact that this loop is non-expansive and has bounded-update, we can conclude that this loop is expected affine-sensitive over its loop guard by Theorem \ref{thm:affine}.
\end{example}

\begin{example}
Consider the mini-roulette variant example showed in Figure~\ref{fig:example1}(right) (Page~\pageref{fig:example1}).
\begin{compactenum}
\item The algorithm sets a linear template $\eta(x,w):=a_{1}\cdot x+a_2\cdot w+a_3$.
\item The algorithm encodes the conditions (A1)--(A4) as the inclusion assertions:
$$
\begin{matrix*}[l]
(A1) & \{(x,w)\mid x\ge 1\wedge w\ge 0\}\subseteq \{(x,w)\mid -a_1\cdot x-a_2\cdot w\le a_3\} \\
(A2) & \{\left((x,w),r_6\right)\mid x\ge 1\wedge w\ge 0\wedge x<3\} \subseteq \{\left((x,w),r_6\right)\mid K\le a_1\cdot (x-r_6)+a_2\cdot w+a_3\le 0\} \\
(A3) & \{(x,w)\mid x\ge 1\wedge w\ge 0\}\subseteq \{(x,w)\mid 0\le\frac{53}{130}a_1-\frac{4}{5}a_2-\epsilon\} \\
\end{matrix*}
$$\vspace{-2mm} $$
\begin{matrix*}[l]
(A4)&\{(x,w)\mid x\ge 1\wedge w\ge 0\}\subseteq &\{(x,w)\mid |a_1\cdot r_1+2a_2|\le c\wedge|a_2\cdot r_2+3a_2|\le c\wedge |a_3\cdot r_3+4a_2|\le c \\
   &  &\wedge |a_4\cdot r_4+5a_2|\le c\wedge |a_5\cdot r_5+6a_2|\le c\wedge |a_6\cdot r_6|\le c\} \\
\end{matrix*}
$$
\item The algorithm applies Farkas' Lemma to all the inclusion assertions generated in the previous step and obtains a system of linear inequalities
involving the parameters $a_1,a_2,a_3,K,\epsilon$, where we over-approximate all strict inequalities (with `$<$') by non-strict ones (with `$\le$').
\item The algorithm calls a linear programming (LP) solver on the linear program consisting of the system of linear inequalities generated in the previous step.
\end{compactenum}
Finally, the algorithm outputs an optimal answer $\eta(x,w)=2.45\cdot x-2.45$ with $\epsilon=1, K=-4.91$(see Example~\ref{va-roulette}). We can verify this $\eta$ a difference-bounded RSM-map with RSM-continuity. In this example, we find both the update function and the loop guard are affine, all sampling variables are bounded continuously-distributed, and the coefficients for the current sampling variables are not all zero in the loop guard of the next iteration, so we can conclude this loop has the
Lipschitz continuity in next-step termination
by Lemma \ref{le:LipconT1}. Therefore, by Theorem \ref{thm:linear}, we can obtained that this loop is expected linear-sensitive over its loop guard.
\end{example}

\begin{remark}[Scalability]
As our approach is based on martingale synthesis, the scalability of our approach relies on the efficiency of martingale synthesis algorithms~\cite{DBLP:journals/toplas/ChatterjeeFNH18,SriramCAV,DBLP:conf/cav/ChatterjeeFG16,ijcai18,ChatterjeeNZ2017,DBLP:conf/pldi/Wang0GCQS19,DBLP:conf/atva/FengZJZX17}.
\end{remark}

\vspace{-1em}
\section{Case Studies and Experimental Results}\label{sect:algorithm}
\vspace{0.5em}

We demonstrate the effectiveness of our approach through case studies and experimental results. First, we consider two case studies of the SGD algorithm.
Then in the experimental results, we use existing RSM-synthesis algorithms ~\cite{DBLP:journals/toplas/ChatterjeeFNH18,SriramCAV,DBLP:conf/cav/ChatterjeeFG16} (as is illustrated in the previous section) to synthesize linear RSM-maps for non-expansive simple loops such as the mini-roulette examples, the heating examples and many other examples from the literature.
Note that by Theorem~\ref{thm:affine} and Theorem~\ref{thm:linear}, the existence of RSM-maps with proper conditions (such as (A4), (B1)--(B4), etc.) leads to the expected sensitivity of all these examples.

\subsection{Case Studies on Stochastic Gradient Descent}\label{sect:sgdcasestudies}

For the general SGD algorithm from Figure~\ref{fig:example4}, we view each value $i$ from $1,2,\dots,n$ a probabilistic branch with probability $\frac{1}{n}$.
By viewing so, we have that the value after one loop iteration is $\updf(\mathbf{w}, i)$ where $\updf$ is the update function, $\mathbf{w}$ is the program valuation before the loop iteration and $i$ is the random integer sampled uniformly from $1,2,\dots,n$.
In the case studies, we consider the metric defined from the Euclidean distance.
To prove the expected affine-sensitivity of the SGD algorithm, we recall several properties for smooth convex functions (see~e.g.~\cite{nesterovbook}).

\begin{definition}[Smooth Convex Functions]
A continuous differentiable function $f:\Rset^{n}\rightarrow \Rset$ is \emph{convex} if for all $u,v\in\Rset^n$, we have that $f(v)\ge f(u)+\langle (\nabla f)(u),v-u\rangle$.
\end{definition}

In the following, we denote by $\convex{n}$ the class of convex continuous differentiable Lipschitz-continuous functions, and by $\cconvex{n}{\beta}$ the subclass of convex continuous differentiable Lipschitz-continuous functions whose gradient is Lipschitz-continuous with a Lipschitz constant $\beta$.
We also denote by $\mathbf{0}$ the zero vector.
The following results can be found in ~\cite{nesterovbook,DBLP:conf/icml/HardtRS16}.

\begin{proposition}\label{prop:glminimum} 
If $f\in\convex{n}$ and $(\nabla f)(u^*)=\mathbf{0}$, then $u^*$ is the global minimum of $f$ on $\Rset^n$.
\end{proposition}


\begin{proposition}\label{prop:nonexp}
For any function $f\in\cconvex{n}{\beta}$ and real number $0<\gamma\le \frac{2}{\beta}$, we have that the function $g$ defined by $g(u):= u-\gamma\cdot (\nabla f)(u)$ is Lipschitz-continuous with Lipschitz constant $1$.
\end{proposition}

We consider the SGD algorithm in Figure~\ref{fig:example4} whose execution time is randomized and depends on the input program valuation. Recall that we consider the loop guard $\Phi$ of the SGD algorithm to be $G(\mathbf{w})\ge \zeta$ so that $\zeta$ is the acceptable threshold of the (non-negative) loss function $G$ and the aim of the algorithm is to find a solution $\mathbf{w}^*$ satisfying $G(\mathbf{w}^*)<\zeta$, as described in Example~\ref{ex:motivsgd}.
Following~\cite{DBLP:conf/icml/HardtRS16}, we assume that each loss function $G_i$ lies in $\cconvex{n}{\beta}$
(so that the Euclidean magnitude of its gradient will always be bounded by a constant $M$), with a single Lipschitz constant $\beta$ for all $i$.  
For practical purpose, we further consider the following assumptions:
\begin{compactenum}
\item the function $G$ has a global minimum $\zeta_{\min}$ at $\mathbf{w}_{\min}$ so that $G(\mathbf{w}_{\min})=\zeta_{\min}$, which is satisfied by many convex functions;
\item the parameters $\mathbf{w}$ will always be bounded during the execution of the SGD algorithm, i.e., $\|\mathbf{w}\|_2\le R$ for some radius $R>0$.
This can be achieved by e.g. regularization techniques where a penalty term is added to each $G_i$ to prevent the parameters from growing too large;
\item in order to ensure termination of the algorithm, we consider that the threshold value $\zeta$ is strictly greater than the global minimum $\zeta_{\min}$.
\end{compactenum}
Below we illustrate the case studies,
showing the expected (approximately linear) sensitivity of the SGD algorithm.
The first fixes the vector $\vartheta$ of $n$ training data that results in a non-expansive loop, while the second considers sensitivity w.r.t the training data and leads to an expansive loop in general.

\noindent{\bf Case Study A: sensitivity w.r.t initial parameters.} In this case study, we fix the training data $\vartheta$.
By Proposition~\ref{prop:nonexp}, we have that the loop is non-expansive w.r.t the Euclidean distance if the step size $\gamma$ is small enough,
for all $i$.
Moreover, from the bound $M$, we have that the loop body has bounded update in $\mathbf{w}$, with a bound $d=M''\cdot \gamma$ for some constant $M''$ determined by $M$.

Define the RSM-map $\eta$ by $\eta(\mathbf{w}):=G(\mathbf{w})-\zeta$. We show that $\eta$ is indeed an RSM-map when the step size $\gamma$ is sufficiently small.
Since $G$ is smooth convex, we have from Proposition~\ref{prop:glminimum} that $\nabla G(\mathbf{w})=\mathbf{0}$ iff $\mathbf{w}$ is the global minimum on $\Rset^{n}$.
Thus by the compactness from the radius $R$, we have that $G(\mathbf{w})\ge \zeta>\zeta_{\min}$ implies ${\parallel}{\nabla G(\mathbf{w})}{\parallel}^2_2 \ge \delta$ for a fixed constant $\delta>0$.
Then by the Mean-Value Theorem, there is a vector $\mathbf{w'}$ on the line segment between $\mathbf{w}$ and $\mathbf{w}-\gamma \cdot \nabla G_i(\mathbf{w})$ such that
\begin{eqnarray*}
\eta(\mathbf{w}-\gamma \cdot \nabla G_i(\mathbf{w}))-\eta(\mathbf{w}) &=& (\nabla G(\mathbf{w}'))^\mathrm{T}\cdot (-\gamma \cdot \nabla G_i(\mathbf{w})) \\
&=& (\nabla G(\mathbf{w}))^\mathrm{T}\cdot (-\gamma \cdot \nabla G_i(\mathbf{w}))+(\nabla G(\mathbf{w})-\nabla G(\mathbf{w}'))^\mathrm{T}\cdot \gamma \cdot \nabla G_i(\mathbf{w})\enskip.\\
\end{eqnarray*}
By the smoothness and the Cauchy-Schwarz's Inequality, we have that
\begin{eqnarray*}
|(\nabla G(\mathbf{w})-\nabla G(\mathbf{w}'))^\mathrm{T}\cdot \gamma \cdot \nabla G_i(\mathbf{w})| &\le & \gamma\cdot \|\nabla G(\mathbf{w})-\nabla G(\mathbf{w}')\|_2\cdot \|\nabla G_i(\mathbf{w})\|_2\\
&\le & \gamma\cdot\beta \cdot \|\mathbf{w}-\mathbf{w}'\|_2\cdot \|\nabla G_i(\mathbf{w})\|_2\\
&\le & \gamma^2\cdot\beta \cdot \|\nabla G_i(\mathbf{w})\|_2^2\\
&\le & \gamma^2\cdot \beta\cdot M^2
\end{eqnarray*}
It follows that
\[
\textstyle \expv_i(\eta(\mathbf{w}-\gamma \cdot \nabla G_i(\mathbf{w})))-\eta(\mathbf{w}) = \frac{1}{n}\cdot \sum_{i=1}^n \left[\eta(\mathbf{w}-\gamma \cdot \nabla G_i(\mathbf{w}))-\eta(\mathbf{w})\right] = -\frac{\gamma}{n}\cdot {\parallel}{\nabla G(\mathbf{w})}{\parallel}^2_2 + C
\]
where $|C|\le \gamma^2\cdot \beta\cdot M^2$.
Hence by choosing a step size $\gamma$ small enough, whenever $G(\mathbf{w})\ge\zeta$ and the SGD algorithm enters the loop, we have
$\expv_i(\eta(\mathbf{w}-\gamma \cdot \nabla G_i(\mathbf{w})))\le \eta(\mathbf{w})- \frac{\delta\cdot\gamma}{2\cdot n}$.
Thus, we can choose  $\epsilon=\frac{\delta\cdot\gamma}{2\cdot n}$ to fulfill the condition (A3).
Moreover, by choosing $K=-\gamma\cdot M'$ for some positive constant $M'$ determined by $M$, we have that when the SGD algorithm terminates, it is guaranteed that $K\le \eta(\mathbf{w})\le 0$.
Hence, $\eta$ is an RSM-map for the SGD algorithm. Then by Theorem~\ref{thm:affine}, we obtain the desired result that the SGD algorithm is expected affine-sensitive
w.r.t the initial input parameters. In detail, we have the coefficients $A_\gamma,B_\gamma$ from Theorem~\ref{thm:affine} that
$A_\gamma=2\cdot\frac{d\cdot M+\epsilon}{\epsilon\cdot D_1}, B_\gamma=-2\cdot\frac{d\cdot K}{\epsilon\cdot D_1}$. As both $d, K, \epsilon$ are proportional to the step size $\gamma$, when $\gamma\rightarrow 0$, we have that $A_\gamma$ remains bounded and $B_\gamma\rightarrow 0$.
Thus, our approach derives that the SGD algorithm is approximately expected linear sensitive (over all $G(\mathbf{w})\ge \zeta$)
when the step size tends to zero.

\noindent{\bf Case Study B: sensitivity w.r.t both initial parameters and training data.} Second, we consider the expected affine-sensitivity
around a neighbourhood of initial parameters $\mathbf{w}^*$ and the training data $\vartheta^*$.
Similar to the first case study, we consider that the values of the parameters $\mathbf{w}$ are always bounded in some radius and the magnitude of each individual training data is also bounded.
A major difference in this case study is that we cannot ensure the non-expansiveness of the loop as the variation in the training data may cause the loop to be expansive.
Instead, we consider the general case that the loop is expansive with the Lipschitz constant $L_\gamma=1+\gamma\cdot C$ for some constant $C>0$. 

Define the RSM-map $\eta$ again as $\eta(\mathbf{w},\vartheta):=G(\mathbf{w},\vartheta)-\zeta$. Similarly, we can show that $\eta$ is an RSM-map when the step size $\gamma$ is sufficiently small,
with parameters $\epsilon, K,d,M$ derived in the same way as in the first case study;
in particular, both $d, K, \epsilon,c$ are proportional to the step size $\gamma$ and we denote $\epsilon=M_1\cdot\gamma$. 
Moreover, the RSM-map $\eta$ is difference bounded with bound $c=M_2\cdot \gamma$, where $M_2$ is a constant determined by $M$.
Then by Theorem~\ref{thm:exp}, we obtain that the SGD algorithm is expected affine-sensitive w.r.t both the initial input parameters and the training data.
By a detailed calculation, we have the coefficients $A_\gamma,B_\gamma$ from Theorem~\ref{thm:exp} that
\begin{eqnarray*}
& & \textstyle A_\gamma:= 2\cdot  A'_\gamma\cdot L_\gamma^{N_\gamma}+2\cdot A'\cdot L_\gamma^{N_\gamma}\cdot \exp\left(-\frac{\epsilon\cdot \eta(\mathbf{w}^*,\vartheta^*)}{8\cdot c^2}\right)\cdot\sum_{n=1}^\infty \left(L_\gamma\cdot \exp\left(-\frac{3\cdot \epsilon^2}{8\cdot c^2}\right)\right)^{n}\\
& & \textstyle B_\gamma:= 2\cdot B'_\gamma+ 2\cdot B'_\gamma\cdot \exp\left(-\frac{\epsilon\cdot \eta(\mathbf{w}^*,\vartheta^*)}{8\cdot c^2}\right)\cdot \sum_{n=1}^\infty \exp\left(-\frac{3\cdot \epsilon^2}{8\cdot c^2}\cdot n\right)
\end{eqnarray*}
where $A'_\gamma=\frac{d\cdot M+\epsilon}{D_1\cdot\epsilon}$, $B'_\gamma=-\frac{d\cdot K}{D_1\cdot\epsilon}$ and $N_\gamma=\lfloor 4\cdot\frac{\eta(\mathbf{w}^*,\vartheta^*)}{\epsilon}\rfloor+1$.
As both $d, K, \epsilon,c$ are proportional to the step size $\gamma$, we have that
$A'_\gamma$ remains bounded and $B'_\gamma\rightarrow 0$.
Moreover, we have that
\[
L_\gamma^{N_\gamma}=(1+C\cdot\gamma)^{\lfloor 4\cdot\frac{\eta(\mathbf{w}^*,\vartheta^*)}{M_1\cdot \gamma}\rfloor+1}\le (1+C\cdot\gamma)^{4\cdot\frac{\eta(\mathbf{w}^*,\vartheta^*)}{M_1\cdot \gamma}+1}\le
e^{4\cdot \frac{C}{M_1}\cdot \eta(\mathbf{w}^*,\vartheta^*)}\cdot (1+C\cdot \gamma)\enskip.
\]
where we recall that $e$ is the base for natural logarithm.
Hence, $L_\gamma^{N_\gamma}$ remains bounded when $\gamma\rightarrow 0$.
Furthermore, as $\frac{\epsilon}{c^2}\rightarrow\infty$ and $\frac{\epsilon^2}{c^2}$ is constant when $\gamma\rightarrow 0$, we obtain that $A_\gamma$ remain bounded and $B_\gamma\rightarrow 0$
when the step size tends to zero.
Thus, we can also assert in this case that the SGD algorithm is approximately expected linear-sensitive (around a neighbourhood of the given input parameters and training data)
when the step size tends to zero.


\subsection{Experimental Results}

We implemented our
approach in Section~\ref{sec:computational} and obtained experimental results on a variety of programs.
Recall we use Lemma~\ref{le:LipconT1} to ensure the 
Lipschitz continuity in next-step termination, and manually check whether a loop is non-expansive (which as mentioned in Section~\ref{sec:computational} can also be automated).

\noindent{\bf Experimental Examples.} We consider examples and their variants from the literature~\cite{DBLP:journals/ejcon/AbateKLP10,DBLP:conf/ijcai/ChatterjeeFGO18,DBLP:journals/toplas/ChatterjeeFNH18,DBLP:conf/pldi/NgoC018,DBLP:conf/pldi/Wang0GCQS19}.
Single/double-room heating is from~\cite{DBLP:journals/ejcon/AbateKLP10}. Mini-roulette, American roulette are from \cite{DBLP:conf/ijcai/ChatterjeeFGO18}. Ad rdwalk 1D, Ad rdwalk 2D are from \cite{DBLP:journals/toplas/ChatterjeeFNH18}. rdwalk, prdwalk, prspeed and race are from \cite{DBLP:conf/pldi/NgoC018}. Simple-while-loop, pollutant-disposal are from \cite{DBLP:conf/pldi/Wang0GCQS19}.
See Appendix~\ref{app:running} for these detailed examples.
All the assignment statements in these examples (except for single/double-room heating) are of the form $x:=x+r$, where $r$ is a constant or a random variable, so we can find these examples are non-expansive.

For single/double-room heating (Figure~\ref{fig:hybrida} and Figure \ref{fig:hybridb}), we choose the values of the parameters $b,b_1,b_2,a_{12},a_{21}$ be small enough, so that we can find the two examples non-expansive by manually inspection. In our experiments, we choose $x_a=10,b=b_1=0.03,c=c_1=1.5,w~\sim\mathsf{unif}(-0.3,0.3)$, $b_2=0.02,a_{12}=a_{21}=0.04,w_1~\sim\mathsf{unif}(-0.3,0.3),w_2~\sim\mathsf{unif}(-0.2,0.2)$. We use the max-norm as the metric.
In this example, as the loop counter $n$ starts always with $0$, our experimental results show that the programs are expected affine/linear sensitive in $n$ (i.e., the number of loop iterations).

\noindent{\bf Implementations and Results.}
We implemented our approach in Matlab R2018b. The results were obtained on a Windows machine with an Intel Core i5 2.9GHz processor and 8GB of RAM.
Examples of expected affine-sensitivity are illustrated in Table~\ref{ex:affineresults}, where the first column specifies the example and the program variable of concern, the second is the running time (in seconds) for the example, the third column is the RSM-map, and the last columns specify its related constants. Examples of expected linear-sensitivity are illustrated in Table~\ref{ex:linearresults} with similar layout.

\begin{table}
\caption{Experimental Results for Expected Affine-sensitivity (with $\epsilon=1$, $L=1$)}
\label{ex:affineresults}
\centering
\begin{tabular}{c|c|c|c|c|c}
\hline
\hline
Example & Time/sec & $\eta(\pv)$ & $K$    &$d$ &$M$ \\
\hline
mini-Roulette &5.97&$13\cdot x-13$&-13 &11&13 \\
\hline
rdwalk & 3.91  &$-5\cdot x+5000$ &-5   & 1&5 \\
\hline
prdwalk variant & 4.88 &$-0.2857\cdot x+285.7$ &-1.429   & 5&$0.2857$ \\
\hline
prspeed & 4.31 &$-1.7143\cdot x+1714.3$ &$-5.143$  &3 &$1.7143$ \\
\hline
race variant & 5.43  &$-1.43\cdot h+1.43\cdot t$ &-4.29    &4 &2.86 \\
\hline
ad. rdwalk 2D & 4.55  &$-0.77\cdot x+0.77\cdot y$ &-2.31    &3 &1.54 \\
\hline
ad. rdwalk 1D Variant & 4.49 &$2.86\cdot x$ &-2.86    &2 & 2.86\\
\hline
American Roulette & 7.66 & $20.27\cdot x-20.27$ &-20.27    &35 &20.27 \\
\hline
\hline
\end{tabular}
\end{table}

\begin{table}
\caption{Experimental Results for Expected Linear-sensitivity (with $\epsilon=1$, $L=1$)}
\label{ex:linearresults}
\centering
\begin{tabular}{c|c|c|c|c|c|c}
\hline
\hline
Example & Time/sec & $\eta(\pv)$ & $K$   &$d$ &$M$ &$c$\\
\hline
mini-roulette variant  & 12.69  &$2.45\cdot x-2.45$ &-4.91    &9 &2.45 &22.08\\
\hline
single-room heating  & 5.21 & $-0.833\cdot x+16.67$&-1.25 & 2.1 &0.833 &1.75\\
\hline
double-room heating &5.64 &$-2.27\cdot x_1+45.45$ &-3.41 & 2.87& 2.27&6.518\\
\hline
rdwalk variant & 4.44 &$-2.5\cdot x+2500$ &-7.5   &3 & 2.5 &7.5\\
\hline
prdwalk & 4.46 &$-0.5714\cdot x+571.4$ &-2.86  &5 &$0.5714$ &2.86\\
\hline
prspeed variant & 5.00  &$-0.5333\cdot x+533.3$ &-2.67   &5 &$0.5333$ &2.67\\
\hline
race & 5.17  &$-2\cdot h+2\cdot t$ &-6   &4 & 4& 6\\
\hline
simple while loop &3.59  &$-2\cdot x+2000$ &-2   &1 & 2& 2\\
\hline
pollutant disposal & 4.87 &$n-5$ &-3   &3 &1 &3\\
\hline
ad. rdwalk 2D variant & 6.07 &$-0.606\cdot x+0.606\cdot y$ &-2.424   &4 &1.212 &2.424 \\
\hline
ad. rdwalk 1D & 5.16 &$1.11\cdot x$ &-2.22   &2 &1.11 &2.22 \\
\hline
American roulette variant &  14.95 & $2.08\cdot x-2.08$ &-4.15  &35 &2.08 & 72.63\\
\hline
\hline
\end{tabular}
\end{table}

\begin{remark}
In this work, we only consider the synthesis of linear RSM-maps to prove expected sensitivity. 
Given that algorithms for synthesis of polynomial RSM-maps are also present (see \cite{DBLP:conf/cav/ChatterjeeFG16,DBLP:conf/atva/FengZJZX17}), it is also possible to tackle the case studies in Section~\ref{sect:sgdcasestudies} if the number of training data is fixed (so that the number of program variables is fixed) and the loss function $G$ is polynomial.
We plan the further investigation of automated synthesis of complex RSM-maps for proving expected sensitivity as a future work. 
\end{remark}

\section{Related Work}\label{sect:relatedwork}
In program verification
Lipschitz continuity has been studied extensively:
a SMT-based method for proving programs robust for a core imperative language is presented in~\cite{DBLP:conf/popl/ChaudhuriGL10};
a linear type system for proving sensitivity has been developed in~\cite{DBLP:conf/icfp/ReedP10};
approaches for differential privacy in higher-order languages have also been considered~\cite{DBLP:conf/popl/AmorimGHKC17,DBLP:conf/popl/GaboardiHHNP13,DBLP:journals/pacmpl/Winograd-CortHR17}.

For probabilistic programs computing expectation properties have been studied over the decades,
such as, influential works on PPDL~\cite{DBLP:journals/jcss/Kozen85} and {\scriptsize P}GCL~\cite{DBLP:journals/toplas/MorganMS96}.
Various approaches have been developed to reason about expected termination time of
probabilistic programs~\cite{DBLP:conf/esop/KaminskiKMO16,DBLP:conf/vmcai/FuC19,DBLP:journals/toplas/ChatterjeeFNH18} as well as to reason about whether a probabilistic program terminates with
probability~one~\cite{mciver2017new,DBLP:conf/aplas/HuangFC18,AgrawalCN18,ChatterjeeNZ2017}.
However, these works focus on non-relational properties, such as, upper bounds expected termination
time, whereas expected sensitivity is intrinsically relational.
To the best of our knowledge while RSMs have been used for non-relational properties, we are the first
 to apply for relational properties.

There is also a great body of literature on relational analysis of probabilistic programs, such as,
relational program logics
~\cite{DBLP:conf/popl/BartheGB09}
and differential privacy of algorithms~\cite{DBLP:conf/popl/BartheKOB12}.
However, this line of works does not consider relational expectation properties.
There have also been several works on relational expectation properties in several specific area, e.g.,
in the area of masking implementations in cryptography, quantitative masking~\cite{DBLP:journals/tcad/EldibWTS15}
and bounded moment model~\cite{DBLP:journals/iacr/BartheDFGSS16}.

The general framework to consider probabilistic program sensitivity was first considered in~\cite{DBLP:conf/popl/BartheGHS17}, and later improved in~\cite{DBLP:journals/pacmpl/BartheEGHS18}.
Several classical examples such as stochastic gradient descent, population dynamics or Glauber dynamics
can be analyzed in the framework of~\cite{DBLP:journals/pacmpl/BartheEGHS18}.
Another method for the sensitivity analysis of probabilistic programs has been proposed in~\cite{DBLP:conf/atva/HuangWM18} and they analysed a linear-regression example derived from the SGD algorithm in~\cite{DBLP:journals/pacmpl/BartheEGHS18}.
For details of literature on relational analysis of probabilistic programs leading to the work of \cite{DBLP:journals/pacmpl/BartheEGHS18} see \cite[Section 9]{DBLP:journals/pacmpl/BartheEGHS18}.



Below we compare our result in detail with the most related results, i.e., \cite{DBLP:journals/pacmpl/BartheEGHS18}, \cite{DBLP:conf/atva/HuangWM18} and also a recent arXiv submission~\cite{DBLP:journals/corr/abs-1901-06540}.
Recall that we have discussed the issue of conditional branches at the end of Section~\ref{sec:introduction}. Here we compare other technical aspects.

\emph{Comparison with \cite{DBLP:journals/pacmpl/BartheEGHS18}.}
The result of \cite{DBLP:journals/pacmpl/BartheEGHS18} is based on the classical notion of couplings.
Coupling is a powerful probabilistic proof technique to compare two distributions $X$ and $Y$ by creating a random distribution $W$ who marginal distributions correspond to $X$ and $Y$. Given a program with two different inputs $x$ and $y$, let $X_i$ and $Y_i$ denote the respective probability distribution after the $i$-th iteration. If the number of iterations is fixed, then coupling can be constructed for each $i$-th step. However, if the number of iterations is randomized and variable-dependent, then in one case termination could be achieved while the other case still continues with the loop. In such situation, a suitable coupling is cumbersome to obtain. Thus, while coupling present an elegant technique for fixed number of iterations, our approach applies directly to the situation where the number of iterations is randomized as well as variable dependent. The advantage of coupling-based proof rules is that such approach can handle variable-dependent sampling and complex data structures through manual proofs. This leads to the fact that their approach can prove rapid mixing of population dynamics and glauber dynamics, while our approach in its current form cannot handle such examples. A strengthening of our approach to handle these type of examples (through e.g. an integration with coupling) is an interesting future direction.

\emph{Comparison with  \cite{DBLP:conf/atva/HuangWM18}.}
The result of \cite{DBLP:conf/atva/HuangWM18} develops an automated tool based on computer algebra that calculates tight sensitivity bounds for probabilistic programs. As computer algebra requires to unroll
every loop into its sequential counterpart without looping, the approach is suitable only for programs with a bounded number of loop iterations and
cannot handle variable-dependent randomized loop iterations that typically lead to unbounded loop iterations. In contrast, our approach can handle unbounded variable-dependent randomized loop iterations.

\emph{Comparison with a recent arXiv submission.} Recently, there is an arXiv submission~\cite{DBLP:journals/corr/abs-1901-06540} that also involves sensitivity analysis of probabilistic programs. 
Although their approach can handle randomized loop iterations to some extent, the approach requires that the executions from two close-by inputs should \emph{synchronize} strictly, and simply assigns $\infty$ to \emph{non-synchronous} situations such as one program valuation enters the loop while the other does not (see the definition for if-branch and while loop in ~\cite[Figure 1]{DBLP:journals/corr/abs-1901-06540}). Moreover, all the examples for illustrating their approach have fixed and bounded number of loop iterations, while all our examples have variable-dependent randomized loop iterations.

\section{Conclusion}\label{sect:conclu}

In this work we studied expected sensitivity analysis of probabilistic programs,
and presented sound approaches for the analysis of (sequential composition of) probabilistic while loops whose termination time is randomized and depends on the input values,
rather than being fixed and bounded.
Our approach can be automated and can handle a variety of programs from the literature.
An interesting future direction is to extend our approach to a wider class of programs (e.g., programs with expansive loops, conditional branches and variable-dependent sampling, synthesis of complex RSM-maps, etc.).
Another important direction is to consider methods for generating tight bounds for sensitivity. 
Besides, integration with coupling-based approaches and practical issues arising from e.g. floating-point arithmetics would also be worthwhile to address.


\begin{acks}                            
We thank anonymous reviewers for helpful comments, especially for pointing to us a scenario of piecewise-linear approximation (Remark~\ref{rmk:condbranch}).
We are grateful to Prof. Yuxi Fu, director of the BASICS Lab at Shanghai Jiao Tong University, for his support.
\end{acks}

\bibliography{sensitivity}


\begin{thebibliography}{45}


\ifx \showCODEN    \undefined \def \showCODEN     #1{\unskip}     \fi
\ifx \showDOI      \undefined \def \showDOI       #1{#1}\fi
\ifx \showISBNx    \undefined \def \showISBNx     #1{\unskip}     \fi
\ifx \showISBNxiii \undefined \def \showISBNxiii  #1{\unskip}     \fi
\ifx \showISSN     \undefined \def \showISSN      #1{\unskip}     \fi
\ifx \showLCCN     \undefined \def \showLCCN      #1{\unskip}     \fi
\ifx \shownote     \undefined \def \shownote      #1{#1}          \fi
\ifx \showarticletitle \undefined \def \showarticletitle #1{#1}   \fi
\ifx \showURL      \undefined \def \showURL       {\relax}        \fi
\providecommand\bibfield[2]{#2}
\providecommand\bibinfo[2]{#2}
\providecommand\natexlab[1]{#1}
\providecommand\showeprint[2][]{arXiv:#2}

\bibitem[\protect\citeauthoryear{Abate, Katoen, Lygeros, and Prandini}{Abate
  et~al\mbox{.}}{2010}]%
        {DBLP:journals/ejcon/AbateKLP10}
\bibfield{author}{\bibinfo{person}{Alessandro Abate},
  \bibinfo{person}{Joost{-}Pieter Katoen}, \bibinfo{person}{John Lygeros},
  {and} \bibinfo{person}{Maria Prandini}.} \bibinfo{year}{2010}\natexlab{}.
\newblock \showarticletitle{Approximate Model Checking of Stochastic Hybrid
  Systems}.
\newblock \bibinfo{journal}{\emph{Eur. J. Control}} \bibinfo{volume}{16},
  \bibinfo{number}{6} (\bibinfo{year}{2010}), \bibinfo{pages}{624--641}.
\newblock
\urldef\tempurl%
\url{https://doi.org/10.3166/ejc.16.624-641}
\showDOI{\tempurl}


\bibitem[\protect\citeauthoryear{Agrawal, Chatterjee, and
  Novotn{\'{y}}}{Agrawal et~al\mbox{.}}{2018}]%
        {AgrawalCN18}
\bibfield{author}{\bibinfo{person}{Sheshansh Agrawal},
  \bibinfo{person}{Krishnendu Chatterjee}, {and} \bibinfo{person}{Petr
  Novotn{\'{y}}}.} \bibinfo{year}{2018}\natexlab{}.
\newblock \showarticletitle{Lexicographic ranking supermartingales: an
  efficient approach to termination of probabilistic programs}.
\newblock \bibinfo{journal}{\emph{{PACMPL}}} \bibinfo{volume}{2},
  \bibinfo{number}{{POPL}} (\bibinfo{year}{2018}),
  \bibinfo{pages}{34:1--34:32}.
\newblock
\urldef\tempurl%
\url{https://doi.org/10.1145/3158122}
\showDOI{\tempurl}


\bibitem[\protect\citeauthoryear{Aguirre, Barthe, Hsu, Kaminski, Katoen, and
  Matheja}{Aguirre et~al\mbox{.}}{2019}]%
        {DBLP:journals/corr/abs-1901-06540}
\bibfield{author}{\bibinfo{person}{Alejandro Aguirre}, \bibinfo{person}{Gilles
  Barthe}, \bibinfo{person}{Justin Hsu}, \bibinfo{person}{Benjamin~Lucien
  Kaminski}, \bibinfo{person}{Joost{-}Pieter Katoen}, {and}
  \bibinfo{person}{Christoph Matheja}.} \bibinfo{year}{2019}\natexlab{}.
\newblock \showarticletitle{Kantorovich Continuity of Probabilistic Programs}.
\newblock \bibinfo{journal}{\emph{CoRR}}  \bibinfo{volume}{abs/1901.06540}
  (\bibinfo{year}{2019}).
\newblock
\showeprint[arxiv]{1901.06540}
\urldef\tempurl%
\url{http://arxiv.org/abs/1901.06540}
\showURL{%
\tempurl}


\bibitem[\protect\citeauthoryear{Aldous}{Aldous}{1983}]%
        {Aldous83}
\bibfield{author}{\bibinfo{person}{David~J. Aldous}.}
  \bibinfo{year}{1983}\natexlab{}.
\newblock \showarticletitle{Random walks on finite groups and rapidly mixing
  Markov chains}.
\newblock \bibinfo{journal}{\emph{S\'eminaire de probabilit\'es de Strasbourg}}
   \bibinfo{volume}{17} (\bibinfo{year}{1983}), \bibinfo{pages}{243--297}.
\newblock
\urldef\tempurl%
\url{http://www.numdam.org/item/SPS_1983__17__243_0}
\showURL{%
\tempurl}


\bibitem[\protect\citeauthoryear{Barthe, Dupressoir, Faust, Gr{\'{e}}goire,
  Standaert, and Strub}{Barthe et~al\mbox{.}}{2016}]%
        {DBLP:journals/iacr/BartheDFGSS16}
\bibfield{author}{\bibinfo{person}{Gilles Barthe},
  \bibinfo{person}{Fran{\c{c}}ois Dupressoir}, \bibinfo{person}{Sebastian
  Faust}, \bibinfo{person}{Benjamin Gr{\'{e}}goire},
  \bibinfo{person}{Fran{\c{c}}ois{-}Xavier Standaert}, {and}
  \bibinfo{person}{Pierre{-}Yves Strub}.} \bibinfo{year}{2016}\natexlab{}.
\newblock \showarticletitle{Parallel Implementations of Masking Schemes and the
  Bounded Moment Leakage Model}.
\newblock \bibinfo{journal}{\emph{{IACR} Cryptology ePrint Archive}}
  \bibinfo{volume}{2016} (\bibinfo{year}{2016}), \bibinfo{pages}{912}.
\newblock
\urldef\tempurl%
\url{http://eprint.iacr.org/2016/912}
\showURL{%
\tempurl}


\bibitem[\protect\citeauthoryear{Barthe, Espitau, Gr{\'{e}}goire, Hsu, and
  Strub}{Barthe et~al\mbox{.}}{2018}]%
        {DBLP:journals/pacmpl/BartheEGHS18}
\bibfield{author}{\bibinfo{person}{Gilles Barthe}, \bibinfo{person}{Thomas
  Espitau}, \bibinfo{person}{Benjamin Gr{\'{e}}goire}, \bibinfo{person}{Justin
  Hsu}, {and} \bibinfo{person}{Pierre{-}Yves Strub}.}
  \bibinfo{year}{2018}\natexlab{}.
\newblock \showarticletitle{Proving expected sensitivity of probabilistic
  programs}.
\newblock \bibinfo{journal}{\emph{{PACMPL}}} \bibinfo{volume}{2},
  \bibinfo{number}{{POPL}} (\bibinfo{year}{2018}),
  \bibinfo{pages}{57:1--57:29}.
\newblock
\urldef\tempurl%
\url{https://doi.org/10.1145/3158145}
\showDOI{\tempurl}


\bibitem[\protect\citeauthoryear{Barthe, Gr{\'{e}}goire, and
  B{\'{e}}guelin}{Barthe et~al\mbox{.}}{2009}]%
        {DBLP:conf/popl/BartheGB09}
\bibfield{author}{\bibinfo{person}{Gilles Barthe}, \bibinfo{person}{Benjamin
  Gr{\'{e}}goire}, {and} \bibinfo{person}{Santiago~Zanella B{\'{e}}guelin}.}
  \bibinfo{year}{2009}\natexlab{}.
\newblock \showarticletitle{Formal certification of code-based cryptographic
  proofs}. In \bibinfo{booktitle}{\emph{Proceedings of the 36th {ACM}
  {SIGPLAN-SIGACT} Symposium on Principles of Programming Languages, {POPL}
  2009, Savannah, GA, USA, January 21-23, 2009}}. \bibinfo{pages}{90--101}.
\newblock
\urldef\tempurl%
\url{https://doi.org/10.1145/1480881.1480894}
\showDOI{\tempurl}


\bibitem[\protect\citeauthoryear{Barthe, Gr{\'{e}}goire, Hsu, and Strub}{Barthe
  et~al\mbox{.}}{2017}]%
        {DBLP:conf/popl/BartheGHS17}
\bibfield{author}{\bibinfo{person}{Gilles Barthe}, \bibinfo{person}{Benjamin
  Gr{\'{e}}goire}, \bibinfo{person}{Justin Hsu}, {and}
  \bibinfo{person}{Pierre{-}Yves Strub}.} \bibinfo{year}{2017}\natexlab{}.
\newblock \showarticletitle{Coupling proofs are probabilistic product
  programs}. In \bibinfo{booktitle}{\emph{Proceedings of the 44th {ACM}
  {SIGPLAN} Symposium on Principles of Programming Languages, {POPL} 2017,
  Paris, France, January 18-20, 2017}}. \bibinfo{pages}{161--174}.
\newblock
\urldef\tempurl%
\url{http://dl.acm.org/citation.cfm?id=3009896}
\showURL{%
\tempurl}


\bibitem[\protect\citeauthoryear{Barthe, K{\"{o}}pf, Olmedo, and
  B{\'{e}}guelin}{Barthe et~al\mbox{.}}{2012}]%
        {DBLP:conf/popl/BartheKOB12}
\bibfield{author}{\bibinfo{person}{Gilles Barthe}, \bibinfo{person}{Boris
  K{\"{o}}pf}, \bibinfo{person}{Federico Olmedo}, {and}
  \bibinfo{person}{Santiago~Zanella B{\'{e}}guelin}.}
  \bibinfo{year}{2012}\natexlab{}.
\newblock \showarticletitle{Probabilistic relational reasoning for differential
  privacy}. In \bibinfo{booktitle}{\emph{Proceedings of the 39th {ACM}
  {SIGPLAN-SIGACT} Symposium on Principles of Programming Languages, {POPL}
  2012, Philadelphia, Pennsylvania, USA, January 22-28, 2012}}.
  \bibinfo{pages}{97--110}.
\newblock
\urldef\tempurl%
\url{https://doi.org/10.1145/2103656.2103670}
\showDOI{\tempurl}


\bibitem[\protect\citeauthoryear{Billingsley}{Billingsley}{1995}]%
        {Billinsleyprobability}
\bibfield{author}{\bibinfo{person}{Patrick Billingsley}.}
  \bibinfo{year}{1995}\natexlab{}.
\newblock \bibinfo{booktitle}{\emph{{P}robability and {M}easure}}.
\newblock \bibinfo{publisher}{JOHN WILEY \& SONS}.
\newblock


\bibitem[\protect\citeauthoryear{Bousquet and Elisseeff}{Bousquet and
  Elisseeff}{2002}]%
        {DBLP:journals/jmlr/BousquetE02}
\bibfield{author}{\bibinfo{person}{Olivier Bousquet} {and}
  \bibinfo{person}{Andr{\'{e}} Elisseeff}.} \bibinfo{year}{2002}\natexlab{}.
\newblock \showarticletitle{Stability and Generalization}.
\newblock \bibinfo{journal}{\emph{Journal of Machine Learning Research}}
  \bibinfo{volume}{2} (\bibinfo{year}{2002}), \bibinfo{pages}{499--526}.
\newblock
\urldef\tempurl%
\url{http://www.jmlr.org/papers/v2/bousquet02a.html}
\showURL{%
\tempurl}


\bibitem[\protect\citeauthoryear{Chakarov and Sankaranarayanan}{Chakarov and
  Sankaranarayanan}{2013}]%
        {SriramCAV}
\bibfield{author}{\bibinfo{person}{Aleksandar Chakarov} {and}
  \bibinfo{person}{Sriram Sankaranarayanan}.} \bibinfo{year}{2013}\natexlab{}.
\newblock \showarticletitle{Probabilistic Program Analysis with Martingales}.
  In \bibinfo{booktitle}{\emph{CAV 2013}}. \bibinfo{pages}{511--526}.
\newblock


\bibitem[\protect\citeauthoryear{Chatterjee}{Chatterjee}{2012}]%
        {DBLP:conf/fossacs/Chatterjee12}
\bibfield{author}{\bibinfo{person}{Krishnendu Chatterjee}.}
  \bibinfo{year}{2012}\natexlab{}.
\newblock \showarticletitle{Robustness of Structurally Equivalent Concurrent
  Parity Games}. In \bibinfo{booktitle}{\emph{Foundations of Software Science
  and Computational Structures - 15th International Conference, {FOSSACS} 2012,
  Held as Part of the European Joint Conferences on Theory and Practice of
  Software, {ETAPS} 2012, Tallinn, Estonia, March 24 - April 1, 2012.
  Proceedings}}. \bibinfo{pages}{270--285}.
\newblock
\urldef\tempurl%
\url{https://doi.org/10.1007/978-3-642-28729-9\_18}
\showDOI{\tempurl}


\bibitem[\protect\citeauthoryear{Chatterjee, Fu, and Goharshady}{Chatterjee
  et~al\mbox{.}}{2016}]%
        {DBLP:conf/cav/ChatterjeeFG16}
\bibfield{author}{\bibinfo{person}{Krishnendu Chatterjee},
  \bibinfo{person}{Hongfei Fu}, {and} \bibinfo{person}{Amir~Kafshdar
  Goharshady}.} \bibinfo{year}{2016}\natexlab{}.
\newblock \showarticletitle{Termination Analysis of Probabilistic Programs
  Through Positivstellensatz's}. In \bibinfo{booktitle}{\emph{Computer Aided
  Verification - 28th International Conference, {CAV} 2016, Toronto, ON,
  Canada, July 17-23, 2016, Proceedings, Part {I}}}
  \emph{(\bibinfo{series}{Lecture Notes in Computer Science})},
  \bibfield{editor}{\bibinfo{person}{Swarat Chaudhuri} {and}
  \bibinfo{person}{Azadeh Farzan}} (Eds.), Vol.~\bibinfo{volume}{9779}.
  \bibinfo{publisher}{Springer}, \bibinfo{pages}{3--22}.
\newblock
\showISBNx{978-3-319-41527-7}
\urldef\tempurl%
\url{https://doi.org/10.1007/978-3-319-41528-4\_1}
\showDOI{\tempurl}


\bibitem[\protect\citeauthoryear{Chatterjee, Fu, Goharshady, and
  Okati}{Chatterjee et~al\mbox{.}}{2018a}]%
        {ijcai18}
\bibfield{author}{\bibinfo{person}{Krishnendu Chatterjee},
  \bibinfo{person}{Hongfei Fu}, \bibinfo{person}{Amir~Kafshdar Goharshady},
  {and} \bibinfo{person}{Nastaran Okati}.} \bibinfo{year}{2018}\natexlab{a}.
\newblock \showarticletitle{Computational Approaches for Stochastic Shortest
  Path on Succinct MDPs}. In \bibinfo{booktitle}{\emph{IJCAI 2018}}.
  \bibinfo{pages}{4700--4707}.
\newblock


\bibitem[\protect\citeauthoryear{Chatterjee, Fu, Goharshady, and
  Okati}{Chatterjee et~al\mbox{.}}{2018b}]%
        {DBLP:conf/ijcai/ChatterjeeFGO18}
\bibfield{author}{\bibinfo{person}{Krishnendu Chatterjee},
  \bibinfo{person}{Hongfei Fu}, \bibinfo{person}{Amir~Kafshdar Goharshady},
  {and} \bibinfo{person}{Nastaran Okati}.} \bibinfo{year}{2018}\natexlab{b}.
\newblock \showarticletitle{Computational Approaches for Stochastic Shortest
  Path on Succinct MDPs}. In \bibinfo{booktitle}{\emph{Proceedings of the
  Twenty-Seventh International Joint Conference on Artificial Intelligence,
  {IJCAI} 2018, July 13-19, 2018, Stockholm, Sweden.}}
  \bibinfo{pages}{4700--4707}.
\newblock
\urldef\tempurl%
\url{https://doi.org/10.24963/ijcai.2018/653}
\showDOI{\tempurl}


\bibitem[\protect\citeauthoryear{Chatterjee, Fu, Novotn{\'{y}}, and
  Hasheminezhad}{Chatterjee et~al\mbox{.}}{2018c}]%
        {DBLP:journals/toplas/ChatterjeeFNH18}
\bibfield{author}{\bibinfo{person}{Krishnendu Chatterjee},
  \bibinfo{person}{Hongfei Fu}, \bibinfo{person}{Petr Novotn{\'{y}}}, {and}
  \bibinfo{person}{Rouzbeh Hasheminezhad}.} \bibinfo{year}{2018}\natexlab{c}.
\newblock \showarticletitle{Algorithmic Analysis of Qualitative and
  Quantitative Termination Problems for Affine Probabilistic Programs}.
\newblock \bibinfo{journal}{\emph{{ACM} Trans. Program. Lang. Syst.}}
  \bibinfo{volume}{40}, \bibinfo{number}{2} (\bibinfo{year}{2018}),
  \bibinfo{pages}{7:1--7:45}.
\newblock
\urldef\tempurl%
\url{https://doi.org/10.1145/3174800}
\showDOI{\tempurl}


\bibitem[\protect\citeauthoryear{Chatterjee, Novotn{\'{y}}, and
  \v{Z}ikeli\'{c}}{Chatterjee et~al\mbox{.}}{2017}]%
        {ChatterjeeNZ2017}
\bibfield{author}{\bibinfo{person}{Krishnendu Chatterjee},
  \bibinfo{person}{Petr Novotn{\'{y}}}, {and} \bibinfo{person}{{\DJ}or{\dj}e
  \v{Z}ikeli\'{c}}.} \bibinfo{year}{2017}\natexlab{}.
\newblock \showarticletitle{Stochastic invariants for probabilistic
  termination}. In \bibinfo{booktitle}{\emph{{POPL} 2017}}.
  \bibinfo{pages}{145--160}.
\newblock


\bibitem[\protect\citeauthoryear{Chaudhuri, Gulwani, and Lublinerman}{Chaudhuri
  et~al\mbox{.}}{2010}]%
        {DBLP:conf/popl/ChaudhuriGL10}
\bibfield{author}{\bibinfo{person}{Swarat Chaudhuri}, \bibinfo{person}{Sumit
  Gulwani}, {and} \bibinfo{person}{Roberto Lublinerman}.}
  \bibinfo{year}{2010}\natexlab{}.
\newblock \showarticletitle{Continuity analysis of programs}. In
  \bibinfo{booktitle}{\emph{Proceedings of the 37th {ACM} {SIGPLAN-SIGACT}
  Symposium on Principles of Programming Languages, {POPL} 2010, Madrid, Spain,
  January 17-23, 2010}}. \bibinfo{pages}{57--70}.
\newblock
\urldef\tempurl%
\url{https://doi.org/10.1145/1706299.1706308}
\showDOI{\tempurl}


\bibitem[\protect\citeauthoryear{de~Amorim, Gaboardi, Hsu, Katsumata, and
  Cherigui}{de~Amorim et~al\mbox{.}}{2017}]%
        {DBLP:conf/popl/AmorimGHKC17}
\bibfield{author}{\bibinfo{person}{Arthur~Azevedo de Amorim},
  \bibinfo{person}{Marco Gaboardi}, \bibinfo{person}{Justin Hsu},
  \bibinfo{person}{Shin{-}ya Katsumata}, {and} \bibinfo{person}{Ikram
  Cherigui}.} \bibinfo{year}{2017}\natexlab{}.
\newblock \showarticletitle{A semantic account of metric preservation}. In
  \bibinfo{booktitle}{\emph{Proceedings of the 44th {ACM} {SIGPLAN} Symposium
  on Principles of Programming Languages, {POPL} 2017, Paris, France, January
  18-20, 2017}}. \bibinfo{pages}{545--556}.
\newblock
\urldef\tempurl%
\url{http://dl.acm.org/citation.cfm?id=3009890}
\showURL{%
\tempurl}


\bibitem[\protect\citeauthoryear{Desharnais, Gupta, Jagadeesan, and
  Panangaden}{Desharnais et~al\mbox{.}}{2004}]%
        {DBLP:journals/tcs/DesharnaisGJP04}
\bibfield{author}{\bibinfo{person}{Josee Desharnais}, \bibinfo{person}{Vineet
  Gupta}, \bibinfo{person}{Radha Jagadeesan}, {and} \bibinfo{person}{Prakash
  Panangaden}.} \bibinfo{year}{2004}\natexlab{}.
\newblock \showarticletitle{Metrics for labelled Markov processes}.
\newblock \bibinfo{journal}{\emph{Theor. Comput. Sci.}} \bibinfo{volume}{318},
  \bibinfo{number}{3} (\bibinfo{year}{2004}), \bibinfo{pages}{323--354}.
\newblock
\urldef\tempurl%
\url{https://doi.org/10.1016/j.tcs.2003.09.013}
\showDOI{\tempurl}


\bibitem[\protect\citeauthoryear{Dwork, McSherry, Nissim, and Smith}{Dwork
  et~al\mbox{.}}{2006}]%
        {DMNS06}
\bibfield{author}{\bibinfo{person}{Cynthia Dwork}, \bibinfo{person}{Frank
  McSherry}, \bibinfo{person}{Kobbi Nissim}, {and} \bibinfo{person}{Adam
  Smith}.} \bibinfo{year}{2006}\natexlab{}.
\newblock \showarticletitle{Calibrating Noise to Sensitivity in Private Data
  Analysis}. In \bibinfo{booktitle}{\emph{Proceedings of the Third Conference
  on Theory of Cryptography}} \emph{(\bibinfo{series}{TCC'06})}.
  \bibinfo{publisher}{Springer-Verlag}, \bibinfo{address}{Berlin, Heidelberg},
  \bibinfo{pages}{265--284}.
\newblock


\bibitem[\protect\citeauthoryear{Dwork and Roth}{Dwork and Roth}{2014}]%
        {DBLP:journals/fttcs/DworkR14}
\bibfield{author}{\bibinfo{person}{Cynthia Dwork} {and} \bibinfo{person}{Aaron
  Roth}.} \bibinfo{year}{2014}\natexlab{}.
\newblock \showarticletitle{The Algorithmic Foundations of Differential
  Privacy}.
\newblock \bibinfo{journal}{\emph{Foundations and Trends in Theoretical
  Computer Science}} \bibinfo{volume}{9}, \bibinfo{number}{3-4}
  (\bibinfo{year}{2014}), \bibinfo{pages}{211--407}.
\newblock
\urldef\tempurl%
\url{https://doi.org/10.1561/0400000042}
\showDOI{\tempurl}


\bibitem[\protect\citeauthoryear{Eldib, Wang, Taha, and Schaumont}{Eldib
  et~al\mbox{.}}{2015}]%
        {DBLP:journals/tcad/EldibWTS15}
\bibfield{author}{\bibinfo{person}{Hassan Eldib}, \bibinfo{person}{Chao Wang},
  \bibinfo{person}{Mostafa M.~I. Taha}, {and} \bibinfo{person}{Patrick
  Schaumont}.} \bibinfo{year}{2015}\natexlab{}.
\newblock \showarticletitle{Quantitative Masking Strength: Quantifying the
  Power Side-Channel Resistance of Software Code}.
\newblock \bibinfo{journal}{\emph{{IEEE} Trans. on {CAD} of Integrated Circuits
  and Systems}} \bibinfo{volume}{34}, \bibinfo{number}{10}
  (\bibinfo{year}{2015}), \bibinfo{pages}{1558--1568}.
\newblock
\urldef\tempurl%
\url{https://doi.org/10.1109/TCAD.2015.2424951}
\showDOI{\tempurl}


\bibitem[\protect\citeauthoryear{Farkas}{Farkas}{1894}]%
        {FarkasLemma}
\bibfield{author}{\bibinfo{person}{J. Farkas}.}
  \bibinfo{year}{1894}\natexlab{}.
\newblock \showarticletitle{A Fourier-f\'{e}le mechanikai elv alkalmaz\'{a}sai
  ({H}ungarian)}.
\newblock \bibinfo{journal}{\emph{Mathematikai\'{e}s Term\'{e}szettudom\'{a}nyi
  \'{E}rtesit\"{o}}}  \bibinfo{volume}{12} (\bibinfo{year}{1894}),
  \bibinfo{pages}{457--472}.
\newblock


\bibitem[\protect\citeauthoryear{Feng, Zhang, Jansen, Zhan, and Xia}{Feng
  et~al\mbox{.}}{2017}]%
        {DBLP:conf/atva/FengZJZX17}
\bibfield{author}{\bibinfo{person}{Yijun Feng}, \bibinfo{person}{Lijun Zhang},
  \bibinfo{person}{David~N. Jansen}, \bibinfo{person}{Naijun Zhan}, {and}
  \bibinfo{person}{Bican Xia}.} \bibinfo{year}{2017}\natexlab{}.
\newblock \showarticletitle{Finding Polynomial Loop Invariants for
  Probabilistic Programs}. In \bibinfo{booktitle}{\emph{Automated Technology
  for Verification and Analysis - 15th International Symposium, {ATVA} 2017,
  Pune, India, October 3-6, 2017, Proceedings}} \emph{(\bibinfo{series}{Lecture
  Notes in Computer Science})}, \bibfield{editor}{\bibinfo{person}{Deepak
  D'Souza} {and} \bibinfo{person}{K.~Narayan Kumar}} (Eds.),
  Vol.~\bibinfo{volume}{10482}. \bibinfo{publisher}{Springer},
  \bibinfo{pages}{400--416}.
\newblock
\showISBNx{978-3-319-68166-5}
\urldef\tempurl%
\url{https://doi.org/10.1007/978-3-319-68167-2\_26}
\showDOI{\tempurl}


\bibitem[\protect\citeauthoryear{Fu}{Fu}{2012}]%
        {DBLP:conf/icalp/Fu12}
\bibfield{author}{\bibinfo{person}{Hongfei Fu}.}
  \bibinfo{year}{2012}\natexlab{}.
\newblock \showarticletitle{Computing Game Metrics on Markov Decision
  Processes}. In \bibinfo{booktitle}{\emph{Automata, Languages, and Programming
  - 39th International Colloquium, {ICALP} 2012, Warwick, UK, July 9-13, 2012,
  Proceedings, Part {II}}} \emph{(\bibinfo{series}{Lecture Notes in Computer
  Science})}, \bibfield{editor}{\bibinfo{person}{Artur Czumaj},
  \bibinfo{person}{Kurt Mehlhorn}, \bibinfo{person}{Andrew~M. Pitts}, {and}
  \bibinfo{person}{Roger Wattenhofer}} (Eds.), Vol.~\bibinfo{volume}{7392}.
  \bibinfo{publisher}{Springer}, \bibinfo{pages}{227--238}.
\newblock
\showISBNx{978-3-642-31584-8}
\urldef\tempurl%
\url{https://doi.org/10.1007/978-3-642-31585-5\_23}
\showDOI{\tempurl}


\bibitem[\protect\citeauthoryear{Fu and Chatterjee}{Fu and Chatterjee}{2019}]%
        {DBLP:conf/vmcai/FuC19}
\bibfield{author}{\bibinfo{person}{Hongfei Fu} {and}
  \bibinfo{person}{Krishnendu Chatterjee}.} \bibinfo{year}{2019}\natexlab{}.
\newblock \showarticletitle{Termination of Nondeterministic Probabilistic
  Programs}. In \bibinfo{booktitle}{\emph{Verification, Model Checking, and
  Abstract Interpretation - 20th International Conference, {VMCAI} 2019,
  Cascais, Portugal, January 13-15, 2019, Proceedings}}
  \emph{(\bibinfo{series}{Lecture Notes in Computer Science})},
  \bibfield{editor}{\bibinfo{person}{Constantin Enea} {and}
  \bibinfo{person}{Ruzica Piskac}} (Eds.), Vol.~\bibinfo{volume}{11388}.
  \bibinfo{publisher}{Springer}, \bibinfo{pages}{468--490}.
\newblock
\showISBNx{978-3-030-11244-8}
\urldef\tempurl%
\url{https://doi.org/10.1007/978-3-030-11245-5\_22}
\showDOI{\tempurl}


\bibitem[\protect\citeauthoryear{Gaboardi, Haeberlen, Hsu, Narayan, and
  Pierce}{Gaboardi et~al\mbox{.}}{2013}]%
        {DBLP:conf/popl/GaboardiHHNP13}
\bibfield{author}{\bibinfo{person}{Marco Gaboardi}, \bibinfo{person}{Andreas
  Haeberlen}, \bibinfo{person}{Justin Hsu}, \bibinfo{person}{Arjun Narayan},
  {and} \bibinfo{person}{Benjamin~C. Pierce}.} \bibinfo{year}{2013}\natexlab{}.
\newblock \showarticletitle{Linear dependent types for differential privacy}.
  In \bibinfo{booktitle}{\emph{The 40th Annual {ACM} {SIGPLAN-SIGACT} Symposium
  on Principles of Programming Languages, {POPL} '13, Rome, Italy - January 23
  - 25, 2013}}. \bibinfo{pages}{357--370}.
\newblock
\urldef\tempurl%
\url{https://doi.org/10.1145/2429069.2429113}
\showDOI{\tempurl}


\bibitem[\protect\citeauthoryear{Hardt, Recht, and Singer}{Hardt
  et~al\mbox{.}}{2016}]%
        {DBLP:conf/icml/HardtRS16}
\bibfield{author}{\bibinfo{person}{Moritz Hardt}, \bibinfo{person}{Ben Recht},
  {and} \bibinfo{person}{Yoram Singer}.} \bibinfo{year}{2016}\natexlab{}.
\newblock \showarticletitle{Train faster, generalize better: Stability of
  stochastic gradient descent}. In \bibinfo{booktitle}{\emph{Proceedings of the
  33nd International Conference on Machine Learning, {ICML} 2016, New York
  City, NY, USA, June 19-24, 2016}}. \bibinfo{pages}{1225--1234}.
\newblock
\urldef\tempurl%
\url{http://jmlr.org/proceedings/papers/v48/hardt16.html}
\showURL{%
\tempurl}


\bibitem[\protect\citeauthoryear{Huang, Fu, and Chatterjee}{Huang
  et~al\mbox{.}}{2018a}]%
        {DBLP:conf/aplas/HuangFC18}
\bibfield{author}{\bibinfo{person}{Mingzhang Huang}, \bibinfo{person}{Hongfei
  Fu}, {and} \bibinfo{person}{Krishnendu Chatterjee}.}
  \bibinfo{year}{2018}\natexlab{a}.
\newblock \showarticletitle{New Approaches for Almost-Sure Termination of
  Probabilistic Programs}. In \bibinfo{booktitle}{\emph{Programming Languages
  and Systems - 16th Asian Symposium, {APLAS} 2018, Wellington, New Zealand,
  December 2-6, 2018, Proceedings}} \emph{(\bibinfo{series}{Lecture Notes in
  Computer Science})}, \bibfield{editor}{\bibinfo{person}{Sukyoung Ryu}} (Ed.),
  Vol.~\bibinfo{volume}{11275}. \bibinfo{publisher}{Springer},
  \bibinfo{pages}{181--201}.
\newblock
\showISBNx{978-3-030-02767-4}
\urldef\tempurl%
\url{https://doi.org/10.1007/978-3-030-02768-1\_11}
\showDOI{\tempurl}


\bibitem[\protect\citeauthoryear{Huang, Wang, and Misailovic}{Huang
  et~al\mbox{.}}{2018b}]%
        {DBLP:conf/atva/HuangWM18}
\bibfield{author}{\bibinfo{person}{Zixin Huang}, \bibinfo{person}{Zhenbang
  Wang}, {and} \bibinfo{person}{Sasa Misailovic}.}
  \bibinfo{year}{2018}\natexlab{b}.
\newblock \showarticletitle{PSense: Automatic Sensitivity Analysis for
  Probabilistic Programs}. In \bibinfo{booktitle}{\emph{Automated Technology
  for Verification and Analysis - 16th International Symposium, {ATVA} 2018,
  Los Angeles, CA, USA, October 7-10, 2018, Proceedings}}.
  \bibinfo{pages}{387--403}.
\newblock
\urldef\tempurl%
\url{https://doi.org/10.1007/978-3-030-01090-4\_23}
\showDOI{\tempurl}


\bibitem[\protect\citeauthoryear{Kaminski, Katoen, Matheja, and
  Olmedo}{Kaminski et~al\mbox{.}}{2016}]%
        {DBLP:conf/esop/KaminskiKMO16}
\bibfield{author}{\bibinfo{person}{Benjamin~Lucien Kaminski},
  \bibinfo{person}{Joost{-}Pieter Katoen}, \bibinfo{person}{Christoph Matheja},
  {and} \bibinfo{person}{Federico Olmedo}.} \bibinfo{year}{2016}\natexlab{}.
\newblock \showarticletitle{Weakest Precondition Reasoning for Expected
  Run-Times of Probabilistic Programs}. In
  \bibinfo{booktitle}{\emph{Programming Languages and Systems - 25th European
  Symposium on Programming, {ESOP} 2016, Held as Part of the European Joint
  Conferences on Theory and Practice of Software, {ETAPS} 2016, Eindhoven, The
  Netherlands, April 2-8, 2016, Proceedings}}. \bibinfo{pages}{364--389}.
\newblock
\urldef\tempurl%
\url{https://doi.org/10.1007/978-3-662-49498-1\_15}
\showDOI{\tempurl}


\bibitem[\protect\citeauthoryear{Kozen}{Kozen}{1985}]%
        {DBLP:journals/jcss/Kozen85}
\bibfield{author}{\bibinfo{person}{Dexter Kozen}.}
  \bibinfo{year}{1985}\natexlab{}.
\newblock \showarticletitle{A Probabilistic {PDL}}.
\newblock \bibinfo{journal}{\emph{J. Comput. Syst. Sci.}} \bibinfo{volume}{30},
  \bibinfo{number}{2} (\bibinfo{year}{1985}), \bibinfo{pages}{162--178}.
\newblock
\urldef\tempurl%
\url{https://doi.org/10.1016/0022-0000(85)90012-1}
\showDOI{\tempurl}


\bibitem[\protect\citeauthoryear{Kupferman and Vardi}{Kupferman and
  Vardi}{1997}]%
        {DBLP:conf/compos/KupfermanV97}
\bibfield{author}{\bibinfo{person}{Orna Kupferman} {and}
  \bibinfo{person}{Moshe~Y. Vardi}.} \bibinfo{year}{1997}\natexlab{}.
\newblock \showarticletitle{Modular Model Checking}. In
  \bibinfo{booktitle}{\emph{Compositionality: The Significant Difference,
  International Symposium, COMPOS'97, Bad Malente, Germany, September 8-12,
  1997. Revised Lectures}} \emph{(\bibinfo{series}{Lecture Notes in Computer
  Science})}, \bibfield{editor}{\bibinfo{person}{Willem~P. de~Roever},
  \bibinfo{person}{Hans Langmaack}, {and} \bibinfo{person}{Amir Pnueli}}
  (Eds.), Vol.~\bibinfo{volume}{1536}. \bibinfo{publisher}{Springer},
  \bibinfo{pages}{381--401}.
\newblock
\showISBNx{3-540-65493-3}
\urldef\tempurl%
\url{https://doi.org/10.1007/3-540-49213-5\_14}
\showDOI{\tempurl}


\bibitem[\protect\citeauthoryear{McIver, Morgan, Kaminski, and Katoen}{McIver
  et~al\mbox{.}}{2017}]%
        {mciver2017new}
\bibfield{author}{\bibinfo{person}{Annabelle McIver}, \bibinfo{person}{Carroll
  Morgan}, \bibinfo{person}{Benjamin~Lucien Kaminski}, {and}
  \bibinfo{person}{Joost-Pieter Katoen}.} \bibinfo{year}{2017}\natexlab{}.
\newblock \showarticletitle{A new proof rule for almost-sure termination}.
\newblock \bibinfo{journal}{\emph{Proceedings of the ACM on Programming
  Languages}} \bibinfo{volume}{2}, \bibinfo{number}{POPL}
  (\bibinfo{year}{2017}), \bibinfo{pages}{33}.
\newblock


\bibitem[\protect\citeauthoryear{Meyn and Tweedie}{Meyn and Tweedie}{1993}]%
        {gssmc}
\bibfield{author}{\bibinfo{person}{S.P. Meyn} {and} \bibinfo{person}{R.L.
  Tweedie}.} \bibinfo{year}{1993}\natexlab{}.
\newblock \bibinfo{booktitle}{\emph{{M}arkov {C}hains and {S}tochastic
  {S}tability}}.
\newblock \bibinfo{publisher}{Springer-Verlag, London}.
\newblock
\newblock
\shownote{available at: probability.ca/MT.}


\bibitem[\protect\citeauthoryear{Morgan, McIver, and Seidel}{Morgan
  et~al\mbox{.}}{1996}]%
        {DBLP:journals/toplas/MorganMS96}
\bibfield{author}{\bibinfo{person}{Carroll Morgan}, \bibinfo{person}{Annabelle
  McIver}, {and} \bibinfo{person}{Karen Seidel}.}
  \bibinfo{year}{1996}\natexlab{}.
\newblock \showarticletitle{Probabilistic Predicate Transformers}.
\newblock \bibinfo{journal}{\emph{{ACM} Trans. Program. Lang. Syst.}}
  \bibinfo{volume}{18}, \bibinfo{number}{3} (\bibinfo{year}{1996}),
  \bibinfo{pages}{325--353}.
\newblock
\urldef\tempurl%
\url{https://doi.org/10.1145/229542.229547}
\showDOI{\tempurl}


\bibitem[\protect\citeauthoryear{Nesterov}{Nesterov}{2004}]%
        {nesterovbook}
\bibfield{author}{\bibinfo{person}{Yurii Nesterov}.}
  \bibinfo{year}{2004}\natexlab{}.
\newblock \bibinfo{booktitle}{\emph{{I}ntroductory {L}ectures on {C}onvex
  {O}ptimization}}. \bibinfo{series}{{A}pplied {O}ptimization},
  Vol.~\bibinfo{volume}{87}.
\newblock \bibinfo{publisher}{Springer-Verlag US}.
\newblock
\showISBNx{978-1-4613-4691-3}
\urldef\tempurl%
\url{https://doi.org/10.1007/978-1-4419-8853-9}
\showDOI{\tempurl}


\bibitem[\protect\citeauthoryear{Ngo, Carbonneaux, and Hoffmann}{Ngo
  et~al\mbox{.}}{2018}]%
        {DBLP:conf/pldi/NgoC018}
\bibfield{author}{\bibinfo{person}{Van~Chan Ngo}, \bibinfo{person}{Quentin
  Carbonneaux}, {and} \bibinfo{person}{Jan Hoffmann}.}
  \bibinfo{year}{2018}\natexlab{}.
\newblock \showarticletitle{Bounded expectations: resource analysis for
  probabilistic programs}. In \bibinfo{booktitle}{\emph{Proceedings of the 39th
  {ACM} {SIGPLAN} Conference on Programming Language Design and Implementation,
  {PLDI} 2018, Philadelphia, PA, USA, June 18-22, 2018}}.
  \bibinfo{pages}{496--512}.
\newblock
\urldef\tempurl%
\url{https://doi.org/10.1145/3192366.3192394}
\showDOI{\tempurl}


\bibitem[\protect\citeauthoryear{Reed and Pierce}{Reed and Pierce}{2010}]%
        {DBLP:conf/icfp/ReedP10}
\bibfield{author}{\bibinfo{person}{Jason Reed} {and}
  \bibinfo{person}{Benjamin~C. Pierce}.} \bibinfo{year}{2010}\natexlab{}.
\newblock \showarticletitle{Distance makes the types grow stronger: a calculus
  for differential privacy}. In \bibinfo{booktitle}{\emph{Proceeding of the
  15th {ACM} {SIGPLAN} international conference on Functional programming,
  {ICFP} 2010, Baltimore, Maryland, USA, September 27-29, 2010}}.
  \bibinfo{pages}{157--168}.
\newblock
\urldef\tempurl%
\url{https://doi.org/10.1145/1863543.1863568}
\showDOI{\tempurl}


\bibitem[\protect\citeauthoryear{van Breugel and Worrell}{van Breugel and
  Worrell}{2006}]%
        {DBLP:journals/tcs/BreugelW06}
\bibfield{author}{\bibinfo{person}{Franck van Breugel} {and}
  \bibinfo{person}{James Worrell}.} \bibinfo{year}{2006}\natexlab{}.
\newblock \showarticletitle{Approximating and computing behavioural distances
  in probabilistic transition systems}.
\newblock \bibinfo{journal}{\emph{Theor. Comput. Sci.}} \bibinfo{volume}{360},
  \bibinfo{number}{1-3} (\bibinfo{year}{2006}), \bibinfo{pages}{373--385}.
\newblock
\urldef\tempurl%
\url{https://doi.org/10.1016/j.tcs.2006.05.021}
\showDOI{\tempurl}


\bibitem[\protect\citeauthoryear{Wang, Fu, Goharshady, Chatterjee, Qin, and
  Shi}{Wang et~al\mbox{.}}{2019}]%
        {DBLP:conf/pldi/Wang0GCQS19}
\bibfield{author}{\bibinfo{person}{Peixin Wang}, \bibinfo{person}{Hongfei Fu},
  \bibinfo{person}{Amir~Kafshdar Goharshady}, \bibinfo{person}{Krishnendu
  Chatterjee}, \bibinfo{person}{Xudong Qin}, {and} \bibinfo{person}{Wenjun
  Shi}.} \bibinfo{year}{2019}\natexlab{}.
\newblock \showarticletitle{Cost analysis of nondeterministic probabilistic
  programs}. In \bibinfo{booktitle}{\emph{Proceedings of the 40th {ACM}
  {SIGPLAN} Conference on Programming Language Design and Implementation,
  {PLDI} 2019, Phoenix, AZ, USA, June 22-26, 2019.}},
  \bibfield{editor}{\bibinfo{person}{Kathryn~S. McKinley} {and}
  \bibinfo{person}{Kathleen Fisher}} (Eds.). \bibinfo{publisher}{{ACM}},
  \bibinfo{pages}{204--220}.
\newblock
\showISBNx{978-1-4503-6712-7}
\urldef\tempurl%
\url{https://doi.org/10.1145/3314221.3314581}
\showDOI{\tempurl}


\bibitem[\protect\citeauthoryear{Williams}{Williams}{1991}]%
        {probabilitycambridge}
\bibfield{author}{\bibinfo{person}{David Williams}.}
  \bibinfo{year}{1991}\natexlab{}.
\newblock \bibinfo{booktitle}{\emph{{P}robability with {M}artingales}}.
\newblock \bibinfo{publisher}{Cambridge University Press}.
\newblock


\bibitem[\protect\citeauthoryear{Winograd{-}Cort, Haeberlen, Roth, and
  Pierce}{Winograd{-}Cort et~al\mbox{.}}{2017}]%
        {DBLP:journals/pacmpl/Winograd-CortHR17}
\bibfield{author}{\bibinfo{person}{Daniel Winograd{-}Cort},
  \bibinfo{person}{Andreas Haeberlen}, \bibinfo{person}{Aaron Roth}, {and}
  \bibinfo{person}{Benjamin~C. Pierce}.} \bibinfo{year}{2017}\natexlab{}.
\newblock \showarticletitle{A framework for adaptive differential privacy}.
\newblock \bibinfo{journal}{\emph{{PACMPL}}} \bibinfo{volume}{1},
  \bibinfo{number}{{ICFP}} (\bibinfo{year}{2017}),
  \bibinfo{pages}{10:1--10:29}.
\newblock
\urldef\tempurl%
\url{https://doi.org/10.1145/3110254}
\showDOI{\tempurl}


\end{thebibliography}


\clearpage

\appendix
\section{The Detailed Syntax}\label{app:syntax}

The detailed syntax is in Figure~\ref{fig:syntax}.

\begin{figure}[h]
\begin{align*}
\langle \mathit{stmt}\rangle &::= \langle\mathit{pvar}\rangle \,\mbox{`$:=$'}\, \langle\mathit{expr} \rangle \mid \mbox{`\textbf{skip}'}
\mid \mbox{`\textbf{if}'} \, \langle\mathit{bexpr}\rangle\,\mbox{`\textbf{then}'} \, \langle \mathit{stmt}\rangle \, \mbox{`\textbf{else}'} \, \langle \mathit{stmt}\rangle \,\mbox{`\textbf{fi}'}
\\
& \mid \langle\mathit{stmt}\rangle \, \text{`;'} \, \langle \mathit{stmt}\rangle \mid \mbox{`\textbf{while}'}\, \langle\mathit{bexpr}\rangle \, \text{`\textbf{do}'} \, \langle \mathit{stmt}\rangle \, \text{`\textbf{od}'}
\vspace{\baselineskip}
\\
& \mid \mbox{`\textbf{if}'} \, \mbox{`\textbf{prob}' `('$p$`)'}\,\mbox{`\textbf{then}'} \, \langle \mathit{stmt}\rangle \, \mbox{`\textbf{else}'} \, \langle \mathit{stmt}\rangle \,\mbox{`\textbf{fi}'}\\
\langle \mathit{bexpr}\rangle &::=  \langle\mathit{pexpr} \rangle\, \mbox{`$\leq$'} \,\langle\mathit{pexpr} \rangle ~\mid~ \langle\mathit{pexpr} \rangle\, \mbox{`$\geq$'} \,\langle\mathit{pexpr} \rangle ~\mid~ \text{`}\neg\text{'} \langle\mathit{bexpr}\rangle
\\
& \mid\langle \mathit{bexpr} \rangle \, \mbox{`\textbf{or}'} \, \langle\mathit{bexpr}\rangle~\mid~ \langle \mathit{bexpr} \rangle \, \mbox{`\textbf{and}'} \, \langle\mathit{bexpr}\rangle
\end{align*}
\caption{The Syntax of Probabilistic Programs}
\label{fig:syntax}
\end{figure}

\section{Proof for the Integral Expansion}\label{pf:integral}

\begin{theorem}[Integral Expansion]\label{thm:intexp}
Let $Q$ be a simple probabilistic while loop in the form~(\ref{eq:swl}) and $z$ be a program variable.
For any input program valuation $\pv$ such that $\pv\models\Phi$, we have
\[
\expv_\pv(Z_{\pv})
= \int \sum_{\loc\in\mathbf{L}} p_\loc \cdot \expv_{\updf(\loc,\pv,\mathbf{r})}(Z_{\updf(\loc,\pv,\mathbf{r})}) \,\mathrm{d}\mathbf{r}
\]
where $Z_{\pv'}$ is the random variable for the value of $z$ after the execution of $Q$ from
the input program valuation $\pv'$, and
$p_\loc$ is the probability that the probabilistic branches follows the choices in $\loc$.
\end{theorem}

\begin{proof}
The result follows from the following derivations:
\[
\begin{aligned}
\expv_{\pv}(Z_{\pv})
= &\int_{\omega\in\Omega} Z_\pv(\omega) \,\probm_\pv(\mathrm{d}\omega)\\
&(\mbox{by the definition of expectation}) \\
= &\int_{(\loc,\sv)\circ\omega'\in\Omega}
Z_\pv((\loc,\sv)\circ\omega') \,\probm_\pv(\mathrm{d}(\loc,\sv)\circ\omega')\\
&(\mbox{by unrolling the run }\omega\mbox{ into }(\loc,\mathbf{r})\circ \omega') \\
= &\int_{(\loc,\sv)} \left[\int_{\omega'\in\Omega}
Z_{\updf(\loc,\pv,\sv)}(\omega')\, \probm_{\updf(\loc,\pv,\sv)}(\mathrm{d}\omega') \right]\,\probm(\mathrm{d}(\loc,\sv))\\
&(\mbox{by Fubini's Theorem}) \\
= &\int_{(\loc,\sv)} \mathbb{E}_{\updf(\loc,\pv,\sv)}(Z_{\updf(\loc,\pv,\sv)})\,\probm(\mathrm{d}(\loc,\sv))\\
& (\mbox{by the definition of expectation})\\
= & \int_{\sv} \left[\mathbb{E}_{\updf(\loc,\pv,\sv)}(Z_{\updf(\loc,\pv,\sv)})\,\probm(\mathrm{d}\loc)\right]\,\probm(\mathrm{d}\sv)\\
& (\mbox{by Fubini's Theorem})\\
= & \int_{\sv} \left[\sum_{\loc\in \mathbf{L}}p_\loc\cdot\mathbb{E}_{\updf(\loc,\pv,\sv)}(Z_{\updf(\loc,\pv,\sv)})\right]\,\probm(\mathrm{d}\sv)\\
& (\mbox{as $\mathbf{L}$ is finite and discrete})\enskip.\\
\end{aligned}
\]
\end{proof}


\section{Proofs for Section~\ref{sect:simplecase}}\label{app:simplecase}

To prove Theorem~\ref{thm:affine} we need the following known result.

\begin{theorem}\cite{DBLP:journals/toplas/ChatterjeeFNH18,DBLP:conf/vmcai/FuC19}\label{thm:rsmmaps1}
If there exists an RSM-map $\eta$ with $\epsilon,K$ given as in Definition~\ref{def:RSMmaps1} for a simple probabilistic while loop $Q$ in the form~(\ref{eq:swl}), then for any input program valuation $\pv$ we have $\expv_{\pv}(T)\leq \frac{\eta(\pv)-K}{\epsilon}$, where $T$ is the random variable for the number of loop iterations.
\end{theorem}

\noindent{\bf Theorem~\ref{thm:affine}.}
A simple non-expansive while loop $Q$ in the form~(\ref{eq:swl}) is expected affine-sensitive
over its loop guard $\Sat{\Phi}$ if we have that
\begin{compactitem}
\item
$Q$ has bounded update,
and
\item
there exists an RSM-map for $Q$ that has RSM-continuity.
\end{compactitem}
In particular, we can choose $\theta=\infty$ and $A=2\cdot\frac{d\cdot M+\epsilon}{  \epsilon\cdot D_1}, B=-2\cdot\frac{d\cdot K}{\epsilon\cdot D_1}$
in (\ref{eq:expaffsen}),
where the parameters
$d,M,\epsilon,K,D_1$ are from Definition~\ref{def:RSMmaps1}, Definition~\ref{def:bu1}, Definition~\ref{def:LipconRSM1} and (\ref{metric}).
\begin{proof}
Consider any program variable $z$.
Let $d$ be a bound from Definition~\ref{def:bu1}, and $\eta$ be an RSM-continuous RSM-maps with $\epsilon, K$ from Definition~\ref{def:RSMmaps1} and the constant $M$ from Definition~\ref{def:LipconRSM1}.
Consider any input program valuations $\pv,\pv'$ such that $\pv,\pv'\models\Phi$.
Denote $\delta :=\dist(\pv,\pv')$.
We use $T_{\pv''}$ to denote the random variable for the number of loop iterations of the executions starting from an input program valuation $\pv''$.
We also use $Z_{\pv''}$ to denote the random variable for the value of $z$ after the execution of $Q$ from $\pv''$.
We illustrate the main idea through clarifying the relationships between program valuations $\pv_n,\pv'_n$ in any runs $\omega=\{\pv_n\}_{n\ge 0}$, $\omega'=\{\pv'_n\}_{n\ge 0}$ that start from respectively $\pv,\pv'$ and use the same sampled values in each loop iteration.
Consider that the event $\min\{T_{\pv}, T_{\pv'}\}\ge n$ holds (i.e., both the executions do not terminate before the $n$th step). We have the following cases:
\begin{compactitem}
\item[{\em Case 1.}]
	Both $\pv_n$ and $\pv'_n$ violate the loop guard $\Phi$, i.e., $\pv_n,\pv'_n\models \neg\Phi$.
	This case describes that the loop $Q$ terminates exactly after the $n$-th iteration of the loop for both the initial valuations. From the condition (B1), we obtain directly that
	$\dist(\pv_n,\pv'_n)\leq \delta$. Hence $|\pv_n[z]-\pv'_n[z]|\le \frac{\dist(\pv_n,\pv'_n)}{D_1}\le \frac{\delta}{D_1}$.
\item[{\em Case 2.}]  Exactly one of $\pv_n, \pv'_n$ violates the loop guard $\Phi$. This is the non-synchronous situation that needs to be addressed through martingales.
W.l.o.g., we can assume that $\pv_n\models\Phi$ and $\pv'_n\models \neg\Phi$. From the upper-bound property of RSM-maps (see Theorem~\ref{thm:rsmmaps1} in Appendix~\ref{app:simplecase}), we derive that $\expv_{\pv_n}(T_{\pv_n})\leq\dfrac{\eta(\pv_n) -K}{\epsilon}$.
From the bounded-update condition (B2) and the triangle inequality of metrics, we have that $|\pv_n[z]-Z_{\pv_n}|\le \frac{1}{D_1}\cdot\dist(\pv_n,\mathbf{W}_{\pv_n})\le \frac{d}{D_1}\cdot T_{\pv_n}$.
Thus, we obtain that
\[
|\expv_{\pv_n}(Z_{\pv_n})-\pv_n[z]| \leq \expv_{\pv_n}(|\pv_n[z]-Z_{\pv_n}|) \leq  \expv_{\pv_n}\left(\frac{d}{D_1}\cdot T_{\pv_n}\right)   \leq \frac{d}{D_1}\cdot \dfrac{\eta(\pv_n) -K}{\epsilon}\enskip.
\]
From the non-expansiveness we have $\dist(\pv_n,\pv'_n)\leq \delta$. Then by the RSM-continuity (B3), we have $|\eta(\pv_n)-\eta(\pv'_n)|\leq M\cdot \delta$.
    Furthermore, from  (A2) we have $\eta(\pv'_n)\le 0$. So we obtain that $\eta(\pv_n)\leq M\cdot \delta$.
It follows that
\begin{eqnarray*}
|\expv_{\pv_n}(Z_{\pv_n})-\expv_{\pv'_n}(Z_{\pv'_n})| &=& |\expv_{\pv_n}(Z_{\pv_n})-\pv'_n[z]|  \\
&\leq &  |\expv_{\pv_n}(Z_{\pv_n})-\pv_n[z]|+ |\pv_n[z]-\pv'_n[z] |  \\
&\leq & \frac{d}{D_1}\cdot\dfrac{M\cdot  \delta-K}{\epsilon}+ \frac{\delta}{D_1} = \frac{d\cdot M+\epsilon}{\epsilon\cdot D_1}\cdot   \delta- \frac{d\cdot K}{\epsilon\cdot D_1}
\end{eqnarray*}
\item[{\em Case 3.}] Neither $\pv_n$ nor $\pv'_n$ violates the loop guard $\Phi$. In this case, the loop $Q$ will continue from both $\pv_n$ and $\pv'_n$.
Then in the next iteration, the same analysis can be carried out for the next program valuations $\pv_{n+1},\pv'_{n+1}$.
\end{compactitem}
From the termination property ensured by RSM-maps (Theorem~\ref{thm:rsmmaps1}),
the probability that the third case happens infinitely often equals zero.
Thus, the sensitivity analysis eventually reduces to the first two cases, and the first two cases derives the expected affine-sensitivity.

In the following, we demonstrate the detailed proof.
By Theorem~\ref{thm:intexp} (in Appendix~\ref{pf:integral}), for a program valuation $\pv''$ satisfying the loop guard $\Phi$, we can derive that
\[
\expv_{\pv''}(Z_{\pv''})
= \int \left[\sum_{\loc\in\mathbf{L}} p_\loc\cdot \expv_{\updf(\loc,\pv'',\sv)}( Z_{\updf(\loc,\pv'',\sv)}) \right]\,\mathrm{d}\sv
\]
On the other hand, if $\pv''\not\models \Phi$, then we obtain straightforwardly that $\expv_{\pv''}(Z_{\pv''})=\pv''[z]$.
Note that if for a particular sampled valuation $\mathbf{r}$ we have $\updf(\loc,\pv'',\sv)\models\Phi$, then we can use Theorem~\ref{thm:intexp} to expand the integrand further into an integral, i.e.,
\begin{equation}\label{eq:intexp1}
\expv_{\updf(\loc,\pv'',\sv)}(Z_{\updf(\loc,\pv'',\sv)})
= \int \left[ \sum_{\loc\in\mathbf{L}}p_\loc\cdot \expv_{\updf(\updf(\loc,\pv'',\sv),\sv')}(Z_{{\updf(\updf(\loc,\pv'',\sv),\sv')}}) \right]\,\mathrm{d}\sv'\enskip.
\end{equation}
Note that once we have $\pv,\pv'\models\Phi$, we can derive from the linearity of integral that
\begin{equation}\label{eq:intexp22}
\expv_{\pv}(Z_{\pv})-\expv_{\pv'}(Z_{\pv'})
= \int  \left[\sum_{\loc\in\mathbf{L}} p_{\loc}\cdot(\expv_{{\updf}(\loc,\pv,\sv)}(Z_{{\updf}(\loc,\pv,\sv)})- \expv_{\updf(\loc,\pv',\sv)}(Z_{{\updf}(\loc,\pv',\sv)}))\right]\,\mathrm{d}\,\sv\enskip.
\end{equation}
Then using (\ref{eq:intexp1}), we can expand the integral above into an arbitrary depth until we reach an end situation $\expv_{\pv_1}(Z_{\pv_1})- \expv_{\pv_2}(Z_{\pv_{2}})$
such that either $\pv_1\not\models \Phi$ or $\pv_2\not\models \Phi$. (This situation corresponds to the first two cases demonstrated previously and will eventually happen since we have $\probm(T_\pv<\infty)=\probm(T_{\pv'}<\infty)=1$.)

Below given any program valuation $\pv''$ and any infinite sequence $\rho=\{\sv_n\}$ where each $\sv_n$ represents the sampled valuation in the $(n+1)$-th loop iteration, we define the infinite sequence $\omega_{\pv'',\rho}$ as the unique execution that starts from $\pv''$ and follows the samplings in $\rho$.

Consider any program valuations $\pv,\pv'\models \Phi$ and any infinite sequence $\rho$ of sampled valuations such that either $T_\pv(\omega_{\pv,\rho})=m$ or $T_{\pv'}(\omega_{\pv',\rho})=m$, for a natural number $m$.
Denote $\omega_{\pv,\rho}=\{\pv_n\}_{n\ge 0}$ and $\omega_{\pv',\rho}=\{\pv'_n\}_{n\ge 0}$.
Then this sequence corresponds to an integral expansion path from $\expv_{\pv}(Z_{\pv})-\expv_{\pv'}(Z_{\pv'})$
to an end situation $\expv_{\pv_m}(Z_{\pv_m})-\expv_{\pv'_m}(Z_{\pv'_m})$ such that either
$\pv_m\not\models \Phi$ or $\pv'_m\not\models\Phi$. Since this situation falls in the first two cases discussed previously in the proof, we obtain that
\[
|\expv_{\pv_m}(Z_{\pv_m})-\expv_{\pv'_m}(Z_{\pv'_m})| \le \frac{d\cdot M+\epsilon}{\epsilon\cdot D_1}\cdot   \delta- \frac{d\cdot K}{\epsilon\cdot D_1}\enskip.
\]
Denote $A':=\frac{d\cdot M+\epsilon}{\epsilon\cdot D_1}$ and $B':=\frac{d\cdot K}{\epsilon\cdot D_1}$. Since the choice of $\rho$ is arbitrary,
we have that the total amount contributed to the value of the integral in (\ref{eq:intexp22})
under the situation ``either $T_{\pv}=m$ or $T_{\pv'}=m$'' is no more than
\[
\probm(T_{\pv}=m\vee T_{\pv'}=m) \cdot  (A'\cdot   \delta+B'),
\]
which is no greater than
\[
(\probm(T_{\pv}=m)+\probm(T_{\pv'}=m))\cdot  (A'\cdot   \delta+B').
\]
By summing up all steps $m$'s, we obtain that
\begin{eqnarray*}
|\expv_{\pv}(Z_{\pv})-\expv_{\pv'}(Z_{\pv'})| &\le & \sum_{m=1}^{\infty} (\probm(T_{\pv}=m)+\probm(T_{\pv'}=m))\cdot (A'\cdot   \delta+B')\\
& \le & 2\cdot  (A'\cdot \delta+B')\\
&= & A\cdot\delta+B
\end{eqnarray*}
where $A:= 2\cdot A'$ and $B:=2\cdot B'$.
\end{proof}

\section{Proofs for Section~\ref{sect:avsen}}\label{app:avsen}

\noindent{\bf Lemma~\ref{le:LipconT1}.}
Consider a simple while loop $Q$ in the form (\ref{eq:swl}) that satisfies the following conditions: 
\begin{compactenum}
\item both $\updf$ and $\Phi$ are affine and $\Phi$ is equivalent to some DNF $\bigvee_{i\in\mathcal{I}} (\mathbf{A}_i\cdot \pv\le \mathbf{d}_i)$;
\item all sampling variables are continuously-distributed whose probability density functions have bounded values;
\item for all $i\in\mathcal{I}$, $\loc\in\mathbf{L}$ and program valuations $\pv\models\Phi$, the coefficients for the sampling variables $\sv$ in $\mathbf{A}_i\cdot \updf(\loc, \pv, \sv)$  are not all zero at each row, \ie, the truth value of each disjunctive clause in $\Phi$ for $\updf(\loc,\pv,\sv)$ depends on $\sv$ at every row.
\end{compactenum}
Then the loop $Q$ has the 
Lipschitz continuity in next-step termination
w.r.t any metric $\dist$.
\begin{proof}
Denote the update function $\updf$ by
$\updf(\loc, \pv,\sv)=\mathbf{B}\cdot \pv+ \mathbf{C}\cdot \sv +\mathbf{c}$\enskip.
Consider any $\loc\in\mathbf{L}$ and any program valuations $\pv,\pv'$ that satisfy $\Phi$.
Then the probability $p$ that $\updf(\loc, \pv, (r_1,r_2))\models \Phi$ and $\updf(\loc, \pv', (r_1,r_2))\not\models \Phi$ is smaller than the probability that for some $i$, $\mathbf{A}_i\cdot \updf(\loc, \pv, (r_1,r_2))\le \mathbf{d}_i$ and $\mathbf{A}_i\cdot \updf(\loc, \pv', (r_1,r_2))\,{\not\le}\,\mathbf{d}_i$.
Let $p_i$ be the probability that $\mathbf{A}_i\cdot \updf(\loc, \pv, (r_1,r_2))\le \mathbf{d}_i$ and $\mathbf{A}_i\cdot \updf(\loc, \pv', (r_1,r_2))\,{\not\le}\,\mathbf{d}_i$. Then $p_i$ equals the probability of the event that
\begin{equation}\label{eq:continuityevent1}
\begin{cases}
\mathbf{A}_i\cdot\mathbf{C}\cdot \sv \le \mathbf{d}_i-\mathbf{A}_i\cdot \mathbf{B}\cdot \pv-\mathbf{A}_i\cdot \mathbf{c} & \\
\mathbf{A}_i\cdot\mathbf{C}\cdot \sv \,{\not\le}\, \mathbf{d}_i-\mathbf{A}_i\cdot \mathbf{B}\cdot \pv'-\mathbf{A}_i\cdot \mathbf{c} &
\end{cases}\enskip.
\end{equation}
Furthermore, the event (\ref{eq:continuityevent1}) implies that for some row $j$, the event
\begin{equation}\label{eq:continuityevent2}
\begin{cases}
(\mathbf{A}_i\cdot\mathbf{C})_j\cdot \sv \le (\mathbf{d}_i-\mathbf{A}_i\cdot \mathbf{c})_j-(\mathbf{A}_i\cdot \mathbf{B})_j\cdot \pv\\
(\mathbf{A}_i\cdot\mathbf{C})_j\cdot \sv > (\mathbf{d}_i-\mathbf{A}_i\cdot \mathbf{c})_j-(\mathbf{A}_i\cdot \mathbf{B})_j\cdot \pv'
\end{cases}
\end{equation}
holds. Denote the probability of the event (\ref{eq:continuityevent2}) by $p_{ij}$.
By the third condition in the statement of the lemma, we have that $(\mathbf{A}_i\cdot\mathbf{C})_j$ is not the zero vector.
Then following from the fact that all the sampling variables are continuously-distributed and have bounded probability density functions, the probability $p_{ij}$ is no greater than $L'_{ij}\cdot {\parallel}{\pv-\pv'}{\parallel}_\infty$ where $L'_{ij}$ is a constant.
To clarify this point, we can assume that there are only two sampling variables $r_1,r_2$.
The situation for more variables is similar.
Under this assumption, we have $(\mathbf{A}_i\cdot\mathbf{C})_j\cdot \sv=a_1\cdot r_1+a_2\cdot r_2$ (note that $a_1,a_2$ are not all zero) and the event (\ref{eq:continuityevent2}) becomes
\[
(\mathbf{d}_i-\mathbf{A}_i\cdot \mathbf{c})_j-(\mathbf{A}_i\cdot \mathbf{B})_j\cdot \pv'<a_1\cdot r_1+a_2\cdot r_2\le (\mathbf{d}_i-\mathbf{A}_i\cdot \mathbf{c})_j-(\mathbf{A}_i\cdot \mathbf{B})_j\cdot \pv.
\]
If both $a_1,a_2$ are non-zero and $a_2>0$, then we derive directly that
\[
p_{ij}\le \int_{-\infty}^\infty f_1(r_1)\cdot \int_{a}^b  f_2(r_2)\,\mathrm{d}r_2\mathrm{d}r_1
\]
where $f_1, f_2$ are probability density functions for respectively $r_1,r_2$ and
\[
a:=\frac{(\mathbf{d}_i-\mathbf{A}_i\cdot \mathbf{c})_j-(\mathbf{A}_i\cdot \mathbf{B})_j\cdot \pv'-a_1\cdot r_1}{a_2}
\]
\[
b:=\frac{(\mathbf{d}_i-\mathbf{A}_i\cdot \mathbf{c})_j-(\mathbf{A}_i\cdot \mathbf{B})_j\cdot \pv-a_1\cdot r_1}{a_2}\enskip.
\]
Thus we have $p_{ij}\le L'_{ij}\cdot {\parallel}{\pv-\pv'}{\parallel}_\infty$ where the constant $L'_{ij}$ is determined by $a_2, (\mathbf{A}_i\cdot \mathbf{B})_j$ and the bound for $f_2$.
The situation for other cases is similar.
So we have that $p\le \sum_{i,j} p_{ij}\le (\sum_{i,j} L'_{i,j})\cdot {\parallel}{\pv-\pv'}{\parallel}_\infty\le \frac{\sum_{i,j} L'_{i,j}}{D_1}\cdot \dist(\pv,\pv')$.
Finally, we sum up all the probabilistic choices $\loc\in\mathbf{L}$ and obtain the desired result.
\end{proof}

{\textbf{Proposition}~\ref{prop:linear}.}
A non-expansive simple while loop $Q$ in the form~(\ref{eq:swl}) has expected linear-sensitivity
over any neighbourhood $U_{\Phi,\dist}(\pv^*, \rho)$ of any given $\pv^*\in\Sat{\Phi}$ if $Q$ has (i) bounded update,
(ii) an RSM-map with RSM-continuity and (iii) the 
Lipschitz continuity in next-step termination.
\begin{proof}
The proof resembles the one for expected affine-sensitivity (Theorem~\ref{thm:affine}).
Consider runs $\omega=\{\pv_n\}_{n\ge 0}$, $\omega'=\{\pv'_n\}_{n\ge 0}$ starting from respectively program valuations $\pv,\pv'\in U_{\Phi,\dist}(\pv^*, \rho)$ and use the same sampled values for each loop iteration. We first analyze the case that $\pv=\pv^*$.
Consider that the event $\min\{T_{\pv}, T_{\pv'}\}\ge n$ holds.
We have exactly the three cases demonstrated in the expected affine-sensitivity analysis, and again the sensitivity analysis eventually reduces to the first two cases (see the proof for Theorem~\ref{thm:affine}).
As we enhance the conditions in Theorem~\ref{thm:affine} with the 
Lipschitz continuity in next-step termination,
we have a strengthened analysis for the second case (i.e., exactly one of $\pv_n,\pv'_n$ violates the loop guard) as follows.
W.l.o.g, we assume that $\pv_n\models\Phi$ and $\pv'_n\not\models\Phi$ in the second case.
As in the proof of Theorem~\ref{thm:affine}, we define $\delta := \dist(\pv,\pv')\le \rho$ and obtain that
$|\expv_{\pv_n}(Z_{\pv_n})-\expv_{\pv'_n}(Z_{\pv'_n})| \leq  A\cdot \delta+ B\le  A\cdot\rho+ B =:C$
where $A:=\frac{d\cdot M+\epsilon}{\epsilon\cdot D_1}$ and $B:=\frac{d\cdot K}{\epsilon\cdot D_1}$.
From the 
Lipschitz continuity in next-step termination,
we have that the second case happens with probability at most
$L'\cdot \dist(\pv_{n-1}, \pv'_{n-1})$, where $L'>0$ is from Definition~\ref{def:LipconT1}.
Thus, the difference contributed to the total sensitivity $|\expv_{\pv}(Z_{\pv_n})-\expv_{\pv'}(Z_{\pv'_n})|$ in the first two cases
is at most
\[
\probm(T_\pv=n\wedge T_{\pv'}=n)\cdot \dist(\pv_{n},\pv'_{n})+C\cdot L'\cdot\probm(T_\pv\ge n\wedge T_{\pv'}\ge n)\cdot \dist(\pv_{n-1},\pv'_{n-1})
\]
where the first summand is from the first case and the second is from the second case.
From the non-expansiveness, we have that $\dist(\pv_{n-1},\pv'_{n-1}),\dist(\pv_{n},\pv'_{n})\le\delta$.
By summing up all $n$'s together, using the fact that $\expv_{\pv}(T_{\pv})=\sum_{n=0}^\infty \probm(T_{\pv}>n)$, we obtain that
\[
|\expv_{\pv}(Z_{\pv})-\expv_{\pv'}(Z_{\pv'})|\le (C\cdot L'\cdot  \expv_{\pv}(T_\pv) +1) \cdot \delta\le (C\cdot L'\cdot \frac{\eta(\pv^*)-K}{\epsilon}+1) \cdot \delta\enskip.
\]
For the case $\pv\ne\pv^*$, we simply have
\begin{eqnarray*}
|\expv_{\pv}(Z_{\pv})-\expv_{\pv'}(Z_{\pv'})| &\le& |\expv_{\pv}(Z_{\pv})-\expv_{\pv^*}(Z_{\pv'})|+|\expv_{\pv}(Z_{\pv^*})-\expv_{\pv'}(Z_{\pv'})|\\
&\le& 2\cdot (C\cdot L'\cdot \frac{\eta(\pv^*)-K}{\epsilon}+1) \cdot \delta
\end{eqnarray*}
that implies the desired local linear sensitivity.
\end{proof}

{\textbf Lemma~\ref{lemm:pd}.}
If $\eta$ is a difference-bounded RSM-map with the parameters $\epsilon,c$ specified in Definition~\ref{def:RSMmaps1} and Definition~\ref{def:diffbounded}, then there exists a constant $p\in (0,1]$ such that
\begin{compactitem}
\item[\emph{(\dag)}] $\forall\pv: \left(\pv\models\Phi\Rightarrow \probm_{\sv,\loc}(\eta(\updf(\loc,\pv,\sv))-\eta(\pv)\leq -\frac{1}{2}\cdot\epsilon)\ge p\right)$
\end{compactitem}
where the probability $\probm_{\sv,\loc}(-)$ is taken w.r.t the sampled valuation $\sv$ and the resolution $\loc$ for probabilistic branches, and treats the program valuation $\pv$ as constant.
In particular, we can take $p:=\frac{\epsilon}{2\cdot c-\epsilon}$.
\begin{proof}
Consider any program valuation $\pv$ such that $\pv\models\Phi$. Denote $u:=\expv_\sv(\eta(\updf(\loc,\pv,\sv)))-\eta(\pv)$ and $q:=\probm_{\sv,\loc}(\eta(\updf(\loc,\pv,\sv))-\eta(\pv,\sv)\leq -\frac{1}{2}\cdot\epsilon)$.
By (A3) and (A4) we have that $u\le -\epsilon<0$ and $0<\epsilon\le c$. Then from (A4) and the Markov's inequality, we have that
\begin{eqnarray*}
q &= &\probm_{\sv,\loc}(\eta(\updf(\loc,\pv,\sv))-\eta(\pv)\leq -\frac{1}{2}\cdot\epsilon) \\
&= &\probm_{\sv,\loc}(\eta(\updf(\loc,\pv,\sv))-\eta(\pv)+c\leq -\frac{1}{2}\cdot\epsilon+c) \\
&\ge & 1-\frac{c+u}{c-\frac{1}{2}\cdot\epsilon} \qquad \mbox{(by Markov's inequality)}\\
&\ge & 1-\frac{c-\epsilon}{c-\frac{1}{2}\cdot\epsilon}\enskip.
\end{eqnarray*}
By taking $p:=1-\frac{c-\epsilon}{c-\frac{1}{2}\cdot\epsilon}=\frac{\epsilon}{2\cdot c-\epsilon}$, we obtain the desired result.
\end{proof}

\noindent{\bf Proposition~\ref{prop:Ak}.}
For any natural number $n\ge 1$, real numbers $C,D\ge 0$ and probability value $p\in (0,1]$, the following system of linear inequalities (with real variables $A_k$'s ($1\le k\le n$) and $A_\infty$)
\begin{eqnarray*}
(1-p)\cdot A_\infty+ C+ p\cdot A_0  \le  A_1 \\
(1-p)\cdot A_\infty+ C+ p\cdot A_1  \le  A_2 \\
\qquad\qquad\qquad \vdots \qquad\qquad\qquad\\
(1-p)\cdot A_\infty+ C+ p\cdot A_{n-1}  \le  A_n \\
D=A_0\le A_1\le  \dots \le A_n \le A_\infty \\
\end{eqnarray*}
has a solution.
\begin{proof}
We find a solution to the system of linear inequalities by equating each $(1-p)\cdot A_\infty+ C+ p\cdot A_k$ with $A_{k+1}$. After the equating, we have from induction on $k$ that
\begin{eqnarray*}
(1-p)\cdot A_\infty+ C+ p\cdot D  &=&  A_1 \\
(1-p^2)\cdot A_\infty+ (1+p)\cdot C+ p^2\cdot D  &=&  A_2 \\
\qquad\qquad\qquad \vdots \qquad\qquad\qquad\\
\textstyle (1-p^n)\cdot A_\infty+ (\sum_{m=1}^n p^{m-1})\cdot C+ p^n\cdot D  &=&  A_n \\
\textstyle (1-p^{n+1})\cdot A_\infty+ (\sum_{m=1}^{n+1} p^{m-1})\cdot C+ p^{n+1}\cdot D  &=&  A_\infty \enskip.\\
\end{eqnarray*}
By solving the last equation, we obtain that
\begin{compactitem}
\item $A_\infty= \frac{1}{p^{n+1}}\cdot (\sum_{m=1}^{n+1} p^{m-1})\cdot C+D= \frac{1}{p^{n+1}}\cdot \frac{1-p^{n+1}}{1-p}\cdot C+D$,
\item $A_k= (\frac{1-p^{k}}{p^{n+1}}\cdot (\sum_{m=1}^{n+1} p^{m-1})+ (\sum_{m=1}^k p^{m-1}))\cdot C + D=  (\frac{1-p^{k}}{p^{n+1}}\cdot \frac{1-p^{n+1}}{1-p} +  \frac{1-p^{k}}{1-p})\cdot C + D$ for $1\le k\le n$.
\end{compactitem}
Below we show that this solution (together with $A_0=D$) satisfies that $A_0\le A_1\le  \dots \le A_n \le A_\infty$. First we show that $A_0\le A_1$ and $A_n\le A_\infty$.
This follows directly from the fact that
\[
A_1=\frac{1}{p^{n+1}}\cdot C+D \mbox{ and } A_\infty-A_n= (\frac{1}{p}\cdot  \frac{1-p^{n+1}}{1-p}- \frac{1-p^{n}}{1-p})\cdot C\enskip.
\]
Then as $(1-p)\cdot A_\infty+ C+ p\cdot A_k=A_{k+1}$ for $1\le k< n$, we easily prove by induction on $k$ that $A_k\le A_{k+1}$ for all $1\le k<n$.
\end{proof}

\noindent{\bf Theorem~\ref{thm:linear}.}
A non-expansive simple while loop $Q$ in the form~(\ref{eq:swl}) has expected linear-sensitivity
over its loop guard $\Sat{\Phi}$ if $Q$ has (i) bounded update,
(ii) a difference-bounded RSM-map with RSM-continuity and (iii) the 
Lipschitz continuity in next-step termination. 
In particular, we can choose $\theta=\frac{1}{M}$ in (\ref{eq:expaffsen}) where the parameter $M$ is from the RSM-continuity (Definition~\ref{def:LipconRSM1}).
\begin{proof}
Choose any program variable $z$. Denote by $T,T'$ (resp. $Z_n,Z'_n$) the random variables for the number of loop iterations (resp. the value of $z$  at the $n$-th step),
from the input program valuations $\pv,\pv'$,
respectively.
For each natural number $n\ge 0$, define $\delta_n(\pv,\pv'):=\expv_{\pv}(Z_{T\wedge n})-\expv_{\pv'}(Z'_{T'\wedge n})$ for program valuations $\pv,\pv'$, where the random variable $T\wedge n$ is defined as $\min\{T,n\}$ and
$T'\wedge n$ likewise. We also define $\delta(\pv,\pv'):=\expv_{\pv}(Z_{T})-\expv_{\pv'}(Z'_{T'})$.
By Theorem~\ref{thm:rsmmaps1}, we have $\expv_\pv(T),\expv_{\pv'}(T')<\infty$.
Then from the bounded-update condition and Dominated Convergence Theorem, we have that $\lim\limits_{n\rightarrow \infty} \expv(Z_{T\wedge n})=\expv(Z_T)$ and the same holds for $Z_{T'\wedge n}$.
Thus, we have that $\lim\limits_{n\rightarrow \infty}\delta_n(\pv,\pv')=\delta(\pv,\pv')$.

Let $\eta$ be a difference-bounded RSM-continuous RSM-map with the parameters $c,\epsilon,M$ as specified in (A3), (A4), (B3).
By Lemma~\ref{lemm:pd}, we can obtain a probability value $p=\frac{\epsilon}{2\cdot c-\epsilon}$ such that the condition $(\dag)$ holds.
We also construct the regions $R_k$'s ($1\le k\le n^*$) and $R_\infty$ as in the paragraph below Lemma~\ref{lemm:pd}, and have the solution $A_k$'s and $A_\infty$ from Proposition~\ref{prop:Ak}, for which we choose $C=\max\{L'\cdot C',\frac{1}{D_1}\}$ and $D=\frac{1}{D_1}=A_0$, where the definition of $C'$ will be given below.
For the sake of convenience, we also define that $R_0:=\{\pv\mid \pv\not\models \Phi\}$.
Below we prove by induction on $n\ge 0$ that
\begin{compactitem}
\item[(*)] for all $k\in\{1,\dots,n^*,\infty\}$ and all program valuations $\pv,\pv'\models \Phi$, if $\dist(\pv,\pv')\le \frac{1}{M}=\theta$, then we have $|\delta_n(\pv,\pv')|\le A_k\cdot \dist(\pv,\pv')$ when $\pv\in R_k$ for $1\le k\le n^*$, and $|\delta_n(\pv,\pv')|\le A_\infty \cdot \dist(\pv,\pv')$ when $\pv\in R_{n^*}$.
\end{compactitem}

\noindent{\bf Base Step} $n=0$. By definition, we have that for all program valuations $\pv,\pv'\in R_k$, $\delta_0(\pv,\pv')=\pv[z]-\pv'[z]$ and $|\delta_0(\pv,\pv')|=|\pv[z]-\pv'[z]|\le \frac{1}{D_1}\cdot\dist(\pv,\pv')= A_0\cdot\dist(\pv,\pv') \le A_k\cdot \dist(\pv,\pv')$.

\noindent{\bf Inductive Step.} Suppose that the induction hypothesis (*) holds for $n$. We prove the case for $n+1$.
We first consider $1\le k\le n^*$ and program valuations $\pv,\pv'$ such that
$\pv,\pv'\models\Phi$,  $\dist(\pv,\pv')\le\theta$ and $\pv\in R_k$.
From the integral expansion (Theorem~\ref{thm:intexp}), we have that
\[
\expv_\pv(Z_{T\wedge (n+1)})=\int \left[\mathbf{1}_{\updf(\loc,\pv,\sv)\not\models \Phi}\cdot \updf(\loc,\pv,\sv)[z] + \mathbf{1}_{\updf(\loc,\pv,\sv)\models \Phi}\cdot \expv_{\updf(\loc,\pv,\sv)}(Z_{T\wedge n})\right]\,\mathrm{d}\loc \,\mathrm{d}\sv
\]
and similarly,
\[
\expv_{\pv'}(Z'_{T'\wedge (n+1)})=\int \left[\mathbf{1}_{\updf(\loc,\pv,\sv)\not\models \Phi}\cdot \updf(\loc,\pv,\sv)[z] + \mathbf{1}_{\updf(\loc,\pv,\sv)\models \Phi}\cdot \expv_{\updf(\loc,\pv,\sv)}(Z'_{T'\wedge n})\right]\,\mathrm{d}\loc\,\mathrm{d}\sv\enskip.
\]
From (\dag), we have that with probability at least $p=\frac{\epsilon}{2\cdot c-\epsilon}$, it happens that $\eta(\updf(\loc,\pv,\sv))-\eta(\pv)\le -\frac{1}{2}\cdot \epsilon$.
It follows that with probability at least $p$, $\updf(\loc,\pv,\sv)\in \bigcup_{m=0}^{k-1} R_{m}$.
Note that
\begin{eqnarray*}
& & \expv_\pv(Z_{T\wedge (n+1)})- \expv_{\pv'}(Z'_{T'\wedge (n+1)}) \\
&=& \int  \left[\mathbf{1}_{\updf(\loc,\pv,\sv)\not\models\Phi\wedge \updf(\loc,\pv',\sv)\not\models\Phi}\cdot (\updf(\loc,\pv,\sv)[z]-\updf(\loc,\pv',\sv))[z]\right]\,\mathrm{d}\loc\,\mathrm{d}\sv  \\
& &~~{}+\int \left[\mathbf{1}_{\updf(\loc,\pv,\sv)\not\models\Phi\wedge \updf(\loc,\pv',\sv)\models\Phi}\cdot (\updf(\loc,\pv,\sv)[z]-\expv_{\updf(\loc,\pv',\sv)}(Z'_{T'\wedge n}))\right]\,\mathrm{d}\loc\,\mathrm{d}\sv\\
& &~~{}+\int \left[\mathbf{1}_{\updf(\loc,\pv,\sv)\models\Phi\wedge \updf(\loc,\pv',\sv)\not\models\Phi}\cdot (\expv_{\updf(\loc,\pv,\sv)}(Z_{T\wedge n})-\updf(\loc,\pv',\sv)[z])\right]\,\mathrm{d}\loc\,\mathrm{d}\sv\\
& &~~{}+\int \left[\mathbf{1}_{\updf(\loc,\pv,\sv)\models\Phi\wedge \updf(\loc,\pv',\sv)\models\Phi}\cdot (\expv_{\updf(\loc,\pv,\sv)}(Z_{T\wedge n})-\expv_{\updf(\loc,\pv',\sv)}(Z'_{T'\wedge n}))\right]\,\mathrm{d}\loc\,\mathrm{d}\sv\\
\end{eqnarray*}
where the first integral corresponds to the case that the executions from $\pv,\pv'$ both terminate after one loop iteration,
the second and third integrals correspond to the case that one execution terminates but the other does not,
and the last integral correspond to the case that both the executions do not terminate.
In the first integral, we have from the non-expansiveness that
\[
|\updf(\loc,\pv,\sv)[z]-\updf(\loc,\pv',\sv)[z]|\le \frac{1}{D_1}\cdot \dist(\updf(\loc,\pv,\sv),\updf(\loc,\pv',\sv))\le \frac{1}{D_1}\cdot\dist(\pv,\pv')\enskip.
\]
In the second and third integral, we have from the second case in the proof of Theorem~\ref{thm:affine} and the fact $T\wedge n\le T$ that
\[
|\updf(\loc,\pv,\sv)[z]-\expv_{\updf(\loc,\pv',\sv)}(Z'_{T'\wedge n})|\le \left(\frac{d\cdot M+\epsilon}{D_1\cdot\epsilon}\right)\cdot \dist(\pv,\pv')- \frac{d\cdot K}{D_1\cdot \epsilon}\le \left(\frac{d\cdot M+\epsilon}{D_1\cdot\epsilon}\right)\cdot \frac{1}{M}- \frac{d\cdot K}{D_1\cdot \epsilon} =: C'
\]
\[
|\updf(\loc,\pv',\sv)[z]-\expv_{\updf(\loc,\pv,\sv)}(Z_{T\wedge n})|\le \left(\frac{d\cdot M+\epsilon}{D_1\cdot\epsilon}\right)\cdot \dist(\pv,\pv')- \frac{d\cdot K}{D_1\cdot \epsilon}\le  C'\enskip.
\]
Furthermore, we have
\begin{eqnarray*}
& &\int \mathbf{1}_{\updf(\loc,\pv,\sv)\models\Phi\wedge \updf(\loc,\pv',\sv)\models\Phi}\cdot (\expv_{\updf(\loc,\pv,\sv)}(Z_{T\wedge n})-\expv_{\updf(\loc,\pv',\sv)}(Z'_{T'\wedge n}))\,\mathrm{d}\loc \,\mathrm{d}\sv\\
&=&\int \mathbf{1}_{\updf(\loc,\pv,\sv)\models\Phi\wedge \updf(\loc,\pv',\sv)\models\Phi\wedge \eta(\updf(\loc,\pv,\sv))-\eta(\pv)\le-\frac{1}{2}\cdot\epsilon}\cdot (\expv_{\updf(\loc,\pv,\sv)}(Z_{T\wedge n})-\expv_{\updf(\loc,\pv',\sv)}(Z'_{T'\wedge n}))\,\mathrm{d}\loc\,\mathrm{d}\sv\\
& &~~{}+\int \mathbf{1}_{\updf(\loc,\pv,\sv)\models\Phi\wedge \updf(\loc,\pv',\sv)\models\Phi\wedge \eta(\updf(\loc,\pv,\sv))-\eta(\pv)>-\frac{1}{2}\cdot \epsilon}\cdot (\expv_{\updf(\loc,\pv,\sv)}(Z_{T\wedge n})-\expv_{\updf(\loc,\pv',\sv)}(Z'_{T'\wedge n}))\,\mathrm{d}\loc\,\mathrm{d}\sv\enskip.\\
\end{eqnarray*}
Denote
\begin{compactitem}
\item
$q_1:=\probm_{\loc,\sv}(\updf(\loc,\pv,\sv)\not\models\Phi\wedge \updf(\loc,\pv',\sv)\not\models\Phi)$;
\item
$q_2:=\probm_{\loc,\sv}((\updf(\loc,\pv,\sv)\not\models\Phi\wedge \updf(\loc,\pv',\sv)\models\Phi)\vee (\updf(\loc,\pv,\sv)\models\Phi\wedge \updf(\loc,\pv',\sv)\not\models\Phi))$;
\item
$p':=\probm_{\loc,\sv}(\updf(\loc,\pv,\sv)\models\Phi\wedge \updf(\loc,\pv',\sv)\models\Phi\wedge\eta(\updf(\loc,\pv,\sv))\le \eta(\pv)-\frac{\epsilon}{2})$;
\item
$\overline{p}:=\probm_{\loc,\sv}(\eta(\updf(\loc,\pv,\sv))\le \eta(\pv)-\frac{\epsilon}{2})$.
\end{compactitem}
Then $q_1+q_2+p'\ge \overline{p}\ge p$.
From 
the Lipschitz continuity in next-step termination,
we have that $q_2\le L'\cdot \dist(\pv,\pv')$.
Then from $A_0\le A_1\le\dots\le A_k\le A_\infty$ and the induction hypothesis,
we have
\begin{eqnarray*}
|\expv_\pv(Z_{T\wedge (n+1)})- \expv_{\pv'}(Z'_{T'\wedge (n+1)})| &\le& (p'\cdot A_{k-1}+ q_1\cdot \frac{1}{D_1}  + (1-(q_1+q_2+p'))\cdot A_\infty)\cdot \dist(\pv,\pv') + q_2\cdot C'\\
&\le& (p'\cdot A_{k-1}+ q_1\cdot \frac{1}{D_1} + L'\cdot C' + (1-(q_1+q_2+p'))\cdot A_\infty)\cdot \dist(\pv,\pv') \\
&\le & ((q_1+q_2+p')\cdot A_{k-1}+ C + (1-(q_1+q_2+p'))\cdot A_\infty)\cdot \dist(\pv,\pv')\\
&\le & (p\cdot A_{k-1}+ C + (1-p)\cdot A_\infty)\cdot \dist(\pv,\pv')\\
&=& A_k\cdot \dist(\pv,\pv')\enskip.
\end{eqnarray*}
Then we consider the case $\pv\in R_{\infty}$.
In this case, since $\dist(\pv,\pv')\le \frac{1}{M}$ and $\eta(\pv)>c+1$, we have from the RSM-continuity that $\eta(\updf(\loc,\pv,\sv)),\eta(\updf(\loc,\pv',\sv))>0$).
It follows that both $\updf(\loc,\pv,\sv),\updf(\loc,\pv',\sv)$ satisfy the loop guard .
Hence, we have from the induction hypothesis that $|\expv_\pv(Z_{T\wedge (n+1)})- \expv_{\pv'}(Z'_{T'\wedge (n+1)})| \le  A_\infty\cdot \dist(\pv,\pv')$.
Thus, the induction step is proved. By taking the limit $n\rightarrow\infty$, we obtain that the whole loop $Q$ is expected linear-sensitive in each $R_k$.
By taking the maximum constant $A_\infty$ for the global expected linear-sensitivity, we obtain the desired result.
\end{proof}

\section{Proofs for Section~\ref{sect:generalsensi}}\label{app:sect6}

\begin{theorem}~\cite{DBLP:journals/toplas/ChatterjeeFNH18}\label{thm:rsmmaps}
If there exists a difference-bounded RSM-map $\eta$ with $\epsilon, K,K'$
from Definition~\ref{def:RSMmaps1} and $c$ from Definition~\ref{def:diffbounded}, then for any initial program valuation $\pv$ we have $\probm(\tertime>n)\le  \mathrm{exp}({-\frac{(\epsilon\cdot n-\eta(\pv))^2}{2\cdot n\cdot c^2}})$ for all $n>\frac{\eta(\pv)}{\epsilon}$.
\end{theorem}

\noindent{\bf Theorem~\ref{thm:exp}.} Consider a simple while loop $Q$ in the form (\ref{eq:swl})
that satisfies the following conditions:
\begin{compactitem}
\item the loop body $P$ is Lipschitz continuous with a constant $L$ specified in Definition~\ref{def:LipconLB1}, and has bounded update;
\item there exists a difference-bounded RSM-map $\eta$ for $Q$ with RSM-continuity and parameters $\epsilon, K, c$ from Definition~\ref{def:RSMmaps1} and Definition~\ref{def:diffbounded} such that $L<\mathrm{exp}({\frac{3\cdot\epsilon^2}{8\cdot c^2}})$.
\end{compactitem}
Then for any program valuation $\pv^*$ such that $\pv^*\models \Phi$ and $\eta(\pv^*)>0$, there exists a radius $\rho>0$ such that the loop $Q$ is expected affine-sensitive
over $U_{\Phi,\dist}(\pv^*,\rho)$.
In particular, we can choose in Definition~\ref{def:EAS} that
\begin{eqnarray*}
& & A:= 2\cdot  A'\cdot L^N+2\cdot A'\cdot L^N\cdot \exp\left(-\frac{\epsilon\cdot \eta(\pv^*)}{8\cdot c^2}\right)\cdot\sum_{n=1}^\infty \left(L\cdot \exp\left(-\frac{3\cdot \epsilon^2}{8\cdot c^2}\right)\right)^{n}\\
& & B:= 2\cdot B'+ 2\cdot B'\cdot \exp\left(-\frac{\epsilon\cdot \eta(\pv^*)}{8\cdot c^2}\right)\cdot \sum_{n=1}^\infty \exp\left(-\frac{3\cdot \epsilon^2}{8\cdot c^2}\cdot n\right)
\end{eqnarray*}
where $A'=\frac{d\cdot M+\epsilon}{D_1\cdot\epsilon}$, $B'=-\frac{d\cdot K}{D_1\cdot\epsilon}$ and $N=\lfloor 4\cdot\frac{\eta(\pv^*)}{\epsilon}\rfloor+1$, for which the parameters
$d,M,\epsilon,K,D_1$ are from Definition~\ref{def:RSMmaps1}, Definition~\ref{def:bu1}, Definition~\ref{def:LipconRSM1} and (\ref{metric}).
\begin{proof}
The proof follows similar arguments as in the proof for Theorem~\ref{thm:affine}.
Choose an arbitrary program variable $z$.
Let $d$ be a bound for $z$ from Definition~\ref{def:bu1} and  $M$ be a constant from  Definition~\ref{def:LipconRSM1}.
Consider any $\pv^*\models\Phi$ such that $\eta(\pv^*)>0$. Choose $\rho=\frac{\eta(\pv^*)}{2\cdot M}$.
Then for all $\pv\in U_{\Phi,\dist}(\pv^*,\rho)$ we have $\eta(\pv)\ge \frac{\eta(\pv^*)}{2}>0$,
hence $\pv\models\Phi$.
Let $\pv,\pv'\in U_{\Phi,\dist}(\pv^*,\rho)$ be any input program valuations and define
$\delta := \dist(\pv,\pv')\le 2\cdot\rho$.
Consider any runs $\omega=\{\pv_n\}_{n\ge 0}$, $\omega'=\{\pv'_n\}_{n\ge 0}$ starting from respectively $\pv,\pv'$ such that $\omega,\omega'$ follow the same sequence of sampled valuations and $\min\{T_\pv(\omega),T_{\pv'}(\omega')\}\ge n$.
Similar to the proof of  Theorem~\ref{thm:affine}, we have three cases below.

\begin{compactitem}
\item[{\em Case 1.}] Both $\pv_n$ and $\pv'_n$ violate the loop guard $\Phi$, i.e., $\pv_n,\pv'_n\models \neg\Phi$.
From the condition (B1), we obtain directly that
	$\dist(\pv_n,\pv'_n)\leq L^n \cdot \delta$.
\item[{\em Case 2.}] One of $\pv_n, \pv'_n$ violates the loop guard $\Phi$ and the other does not. W.l.o.g., we assume that $\pv_n\models\Phi$ and $\pv'_n\models \neg\Phi$. From the analysis in the previous case we have
    $\dist(\pv_n,\pv'_n)\leq L^n\cdot \delta$.
    Furthermore, by the RSM continuity (B3), we have $|\eta(\pv_n)-\eta(\pv'_n)|\leq M\cdot L^{n}\cdot\delta$.
    Moreover, from the condition (A3) we have
    $\eta(\pv'_n)\le 0$, so we obtain that $\eta(\pv_n)\leq M\cdot L^{n}\cdot\delta$.
    By Theorem~\ref{thm:rsmmaps1}, we have that $\expv_{\pv_n}(T_{\pv_n})\leq\dfrac{\eta(\pv_n) -K}{\epsilon}$, which implies that
	\[
	|\expv_{\pv_n}(Z_{\pv_n})-\pv_n[z]| \leq \expv_{\pv_n}(|Z_{\pv_n}-\pv_n[z]|)
\leq   \frac{1}{D_1}\cdot\expv_{\pv_n}(\dist(Z_{\pv_n},\pv_n[z]))    \leq d\cdot \dfrac{\eta(\pv_n) -K}{\epsilon\cdot D_1}.
	\]
Hence we have
\begin{eqnarray*}
|\expv_{\pv_n}(Z_{\pv_n})-\expv_{\pv'_n}(Z_{\pv'_n})| &=& |\expv_{\pv_n}(Z_{\pv_n})-\pv'_n[z] |\\
&\leq &  |\expv_{\pv_n}(Z_{\pv_n})-\pv_n[z] |+ |\pv_n[z]-\pv'_n[z] |\\                                                                   \\
&\leq &d\cdot\dfrac{M\cdot L^{n}\cdot\delta-K}{D_1\cdot\epsilon}+L^n\cdot\frac{\delta}{D_1}\\
&=& \left(\frac{d\cdot M+\epsilon}{\epsilon\cdot D_1}\right)\cdot L^n\cdot \delta- \frac{d\cdot K}{D_1\cdot\epsilon}.
\end{eqnarray*}	
\item[{\em Case 3.}] Neither $\pv_n$ nor $\pv'_n$ violates the loop guard $\Phi$. In this case, the loop $Q$ will continue with valuations $\pv_n$ and $\pv'_n$.
Then in the next iteration, the same analysis can be carried out for the next program valuations $\pv_{n+1},\pv'_{n+1}$.
\end{compactitem}
Again, the situation that the third case happens infinitely often has probability zero, since our program is almost-surely terminating from the existence of an RSM-map.
Thus, the sensitivity analysis eventually reduces to the first two cases.
By taking into account the exponentially-decreasing property for Theorem~\ref{thm:rsmmaps} and a detailed calculation, we can obtain that if
the constant $L$ (i.e., the speed that the difference between program valuations grows larger) is less than
the exponential decreasing factor of program termination, then the loop $Q$ is expected affine-sensitive.

The detailed calculation is as follows.
Denote $A':=\frac{d\cdot M+\epsilon}{D_1\cdot\epsilon}$ and $B':=-\frac{d\cdot K}{D_1\cdot\epsilon}$.
Then the total amount contributed to the total sensitivity at the $n$th step 
is no more than
$\probm(T_\pv=n\vee T_{\pv'}=n) \cdot (A'\cdot L^n\cdot \delta+B')$,
which is no greater than
$(\probm(T_\pv=n)+\probm(T_{\pv'}=n))\cdot (A'\cdot L^n\cdot \delta+B')$.
Note that if $n\ge \frac{\eta(\pv)}{\epsilon}$, then by Theorem~\ref{thm:rsmmaps} we have
\[
\probm(T_\pv=n+1)\le \mathbb{P}(T_\pv> n)\le \mathrm{exp}\left({-\frac{(\epsilon\cdot n-\eta(\pv))^2}{2\cdot n\cdot c^2}}\right)=:p_n.
\]
Furthermore, if $n\ge 2\cdot \frac{\eta(\pv)}{\epsilon}$, then we can derive that $p_n\le \exp(-\frac{\epsilon^2\cdot n}{8\cdot c^2})\le  \exp(-\frac{\epsilon\cdot \eta(\pv)}{4\cdot c^2})$ and
\begin{eqnarray*}
\frac{p_{n+1}}{p_n} &=& \exp\left(-\frac{(\epsilon\cdot (n+1)-\eta(\pv))^2}{2\cdot (n+1)\cdot c^2} + \frac{(\epsilon\cdot n-\eta(\pv))^2}{2\cdot n\cdot c^2}\right)\\
&=& \exp\left(\frac{ -\epsilon^2\cdot n\cdot (n+1)+\eta^2(\pv)}{2\cdot n\cdot (n+1)\cdot c^2} \right)\\
&\le& \exp\left(\frac{ -\frac{3}{4}\cdot \epsilon^2\cdot n\cdot (n+1)}{2\cdot n\cdot (n+1)\cdot c^2} \right)\\
&\le & \exp\left(-\frac{3\cdot \epsilon^2}{8\cdot c^2} \right)\enskip.
\end{eqnarray*}
Similarly, for $n\ge \frac{\eta(\pv')}{\epsilon}$ we have
\[
\probm(T_{\pv'}=n+1)\le \mathbb{P}(T_{\pv'}> n)\le \mathrm{exp}\left({-\frac{(\epsilon\cdot n-\eta(\pv'))^2}{2\cdot n\cdot c^2}}\right)=:p'_n,
\]
and for $n \ge 2\cdot \frac{\eta(\pv')}{\epsilon}$, we have that $\frac{p'_{n+1}}{p'_n}\le\exp\left(-\frac{3\cdot \epsilon^2}{8\cdot c^2} \right)$ and $p'_n\le \exp(-\frac{\epsilon\cdot \eta(\pv')}{4\cdot c^2})$.
Since $|\eta(\pv')-\eta(\pv)|\le M\cdot \delta\le 2\cdot M\cdot \rho$,
we obtain that for $n\ge 2\cdot\frac{\eta(\pv)}{\epsilon}+4\cdot \frac{M\cdot\rho}{\epsilon}$,
the values of $p_n$'s and $p'_n$'s decrease exponentially with the factor $\exp\left(-\frac{3\cdot \epsilon^2}{8\cdot c^2} \right)$ and are no greater than $\exp(-\frac{\epsilon\cdot \eta(\pv^*)}{8\cdot c^2})$ (as $\eta(\pv),\eta(\pv')\ge\frac{\eta(\pv^*)}{2}$).
By summing up all $n$'s greater than $N:=\lfloor 2\cdot\frac{\eta(\pv^*)}{\epsilon}+4\cdot \frac{M\cdot\rho}{\epsilon}\rfloor+1=\lfloor 4\cdot\frac{\eta(\pv^*)}{\epsilon}\rfloor+1$, we obtain that
\begin{eqnarray*}
|\expv_{\pv}(Z_{\pv})-\expv_{\pv'}(Z_{\pv'})| &\le & \sum_{n=1}^{\infty} (\probm(T_{\pv}=n)+\probm(T_{\pv'}=n))\cdot (A'\cdot L^n\cdot \delta+B')\\
&=& \sum_{n=1}^{N} (\probm(T_\pv=n)+\probm(T_{\pv'}=n))\cdot (A'\cdot L^n\cdot \delta+B')\\
& & {}+ \sum_{n=N+1}^{\infty} (\probm(T_\pv=n)+\probm(T_{\pv'}=n))\cdot (A'\cdot L^n\cdot \delta+B')\\
& \le & 2\cdot (A'\cdot L^N\cdot \delta+B')\\
& & \quad{}+\left[\sum_{n=N+1}^{\infty} (\probm(T_\pv=n)+\probm(T_{\pv'}=n))\cdot(A'\cdot L^n\cdot \delta+B')\right]\\
& \le & 2\cdot (A'\cdot L^N\cdot \delta+B')\\
& & \quad{}+\left[\sum_{n=N+1}^{\infty} (p_{n-1}+p'_{n-1})\cdot(A'\cdot L^n\cdot \delta+B')\right]\\
& \le & 2\cdot (A'\cdot L^N\cdot \delta+B')\\
& & \quad{}+\left[\sum_{n=1}^{\infty} (p_{N}+p'_{N})\cdot \exp\left(-\frac{3\cdot \epsilon^2}{8\cdot c^2}\cdot n \right)   \cdot (A'\cdot L^{n+N}\cdot \delta+B')\right]\\
&\le & A\cdot\delta+B
\end{eqnarray*}
where we have
\[
A:= 2\cdot  A'\cdot L^N+2\cdot A'\cdot L^N\cdot \exp\left(-\frac{\epsilon\cdot \eta(\pv^*)}{8\cdot c^2}\right)\cdot\sum_{m=1}^\infty \left(L\cdot \exp\left(-\frac{3\cdot \epsilon^2}{8\cdot c^2}\right)\right)^{m}
\]
\[
B:= 2\cdot B'+ 2\cdot B'\cdot \exp\left(-\frac{\epsilon\cdot \eta(\pv^*)}{8\cdot c^2}\right)\cdot \sum_{m=1}^\infty \exp\left(-\frac{3\cdot \epsilon^2}{8\cdot c^2}\cdot m\right)\enskip.
\]
It follows that the loop $Q$ is expected affine-sensitive over the neighbourhood $U_{\Phi,\dist}(\pv^*,\rho)$.
\end{proof}

\section{Proofs for Section~\ref{sect:seqcom}}\label{app:seqcom}

\noindent{\bf Theorem~\ref{thm:seqcom}.}
Consider a non-expansive simple while loop $Q$ with  bounded-update
and an RSM-map with RSM-continuity, and a general program $Q'$ that has expected  affine-sensitivity
over a subset $U$ of input program valuations with threshold $\theta$ in (\ref{eq:expaffsen}).
If $\outp(Q)\subseteq U$ and assuming integrability in (\ref{eq:expaffsen}), then the sequential composition $Q;Q'$ is expected affine-sensitive
over the satisfaction set of the loop guard of $Q$ with threshold $\theta$. 
\begin{proof}
The proof is basically an extension to the one for Theorem~\ref{thm:affine}.
Suppose that $Q$ is in the form ~(\ref{eq:swl}) with the bound $d$
from Definition~\ref{def:bu1} and an RSM-map $\eta$ with the parameters $\epsilon, K$ from Definition~\ref{def:RSMmaps1} that has RSM-continuity with a constant $M$ from Definition~\ref{def:LipconRSM1}.
Suppose a general probabilistic program $Q'$ to be expected affine-sensitive over $U$ with the coefficients $A_{Q'},B_{Q'}$ in Definition~\ref{def:RSMmaps1}.
Consider two input program valuations $\pv,\pv'\in\Sat{\Phi}$ such that $\delta := \dist(\pv,\pv')\le \theta$. Let $\omega=\{\pv_n\}_{n\in\mathbb{N}}$, $\omega'=\{\pv'_n\}_{n\in \mathbb{N}}$ be any two runs under $Q$ that start from respectively $\pv,\pv'$ and follow the \emph{same} sampled values for each loop iteration of the loop $P$.
Consider for a step $n$ the event $\min\{T_{\pv}, T_{\pv'}\}\ge n$ holds (i.e., both the executions do not terminate before the $n$th step), where $T_{\pv}, T_{\pv'}$ are the random variables for the number of loop iterations of $Q$ when starting from $\pv,\pv'$, respectively. We have the following cases as in the proof for Theorem~\ref{thm:affine}:
\begin{compactitem}
\item[{\em Case 1.}]
	Both $\pv_n$ and $\pv'_n$ violate the loop guard $\Phi$, i.e., $\pv_n,\pv'_n\models \neg\Phi$.
	This case describes that the loop $Q$ terminates exactly after the $n$th iteration of the loop for both the executions. From the non-expansiveness, we obtain directly that
	$\dist(\pv_n,\pv'_n)\leq \delta\le \theta$. Then from $\pv_n,\pv'_n\in U$ and by the expected affine-sensitivity from $Q'$, we obtain that $|\expv_{\pv_n}(Z^{Q'}_{\pv_n})-\expv_{\pv'_n}(Z^{Q'}_{\pv'_n})|\le A_{Q'}\cdot \delta+B_{Q'}$, where $Z^{Q'}_{\pv_n},Z^{Q'}_{\pv'_n}$ are the random variables representing the value of $z$ after the execution of $Q'$ starting from the input program valuations $\pv_n,\pv'_n$.

\item[{\em Case 2.}] Exactly one of $\pv_n, \pv'_n$ violates the loop guard $\Phi$. W.l.o.g., we can assume that $\pv_n\not\models\Phi$ and $\pv'_n\models \Phi$. From the expected affine-sensitivity of $Q$, we obtain that
    \[
    |\expv_{\pv_n}(Z^{Q'}_{\pv_n})-\expv_{\pv''}(Z^{Q'}_{\pv''})|\le A_{Q'}\cdot \dist(\pv_n, \pv'')+B_{Q'}
    \]
where $\pv''$ is the random program valuation after the execution of $Q$ from $\pv'_n$.
    By the triangle inequality and the bounded update condition, we have that
\begin{eqnarray*}
\dist(\pv_n, \pv'')  \le  \dist(\pv_n, \pv'_n)+\dist(\pv'_n, \pv'') \le  \delta + T_{\pv'_n}\cdot d\enskip.\\
\end{eqnarray*}
Thus, we have that
\begin{eqnarray*}
|\expv_{\pv_n}(Z^{Q;Q'}_{\pv_n})-\expv_{\pv'_n}(Z^{Q;Q'}_{\pv'_n})| 
&=& |\expv_{\pv_n}(Z^{Q'}_{\pv_n})-\expv_{\pv'_n}(\expv_{\pv''}(Z^{Q'}_{\pv''}))| \\
&=& |\expv_{\pv'_n}(\expv_{\pv_n}(Z^{Q'}_{\pv_n})-\expv_{\pv''}(Z^{Q'}_{\pv''}))| \\
&\le& \expv_{\pv'_n}(|\expv_{\pv_n}(Z^{Q'}_{\pv_n})-\expv_{\pv''}(Z^{Q'}_{\pv''})|) \\
&\le& \expv_{\pv'_n}(A_{Q'}\cdot \dist(\pv_n, \pv'')+B_{Q'})\\
&\le& \expv_{\pv'_n}(A_{Q'}\cdot (\delta + T_{\pv'_n}\cdot d)+B_{Q'})\\
&=& A_{Q'}\cdot \delta + A_{Q'}\cdot \expv_{\pv'_n} (T_{\pv'_n})\cdot d+B_{Q'}\enskip.\\
\end{eqnarray*}
As shown in the second case of the proof of Theorem~\ref{thm:affine}, we have that
\[
\expv_{\pv'_n} (T_{\pv'_n})\le \frac{\eta(\pv'_n)-K}{\epsilon}\le \dfrac{M\cdot  \delta-K}{\epsilon}\enskip.
\]
Hence we have that
\begin{eqnarray*}
|\expv_{\pv_n}(Z^{Q;Q'}_{\pv_n})-\expv_{\pv'_n}(Z^{Q;Q'}_{\pv'_n})| &\le&  A_{Q'}\cdot \left(1+\frac{d\cdot M}{\epsilon}\right)\cdot \delta - A_{Q'}\cdot \frac{d\cdot K}{\epsilon}+B_{Q'} \\
&=:&A\cdot \delta + B
\end{eqnarray*}
where $A:=A_{Q'}\cdot \left(1+\frac{d\cdot M}{\epsilon}\right)$ and $B:=-A_{Q'}\cdot \frac{d\cdot K}{\epsilon}+B_{Q'}$.
\item[{\em Case 3.}] Neither $\pv_n$ nor $\pv'_n$ violates the loop guard $\Phi$. In this case, the loop $Q$ will continue from both $\pv_n$ and $\pv'_n$.
Then in the next iteration, the same analysis can be carried out for the next program valuations $\pv_{n+1},\pv'_{n+1}$, and so forth.
\end{compactitem}
From the termination property ensured by RSM-maps (Theorem~\ref{thm:rsmmaps1}),
the probability that the third case happens infinitely often equals zero.
Thus, the sensitivity analysis eventually reduces to the first two cases.
From the first two cases, the difference contributed to the total expected sensitivity $|\expv_{\pv}(Z_{\pv})-\expv_{\pv'}(Z_{\pv'})|$ when one of the runs terminates at a step $n$ is at most
\[
\probm(T_\pv=n\vee T_{\pv'}=n)\cdot (A\cdot   \delta+ B)\enskip.
\]
Then by a summation for all $n$, we derive the desired result that
$|\expv_{\pv}(Z_{\pv})-\expv_{\pv'}(Z_{\pv'})| \le  2\cdot A\cdot\delta+ 2\cdot B$
where $A,B$ are given as above.
\end{proof}

To prove Theorem~\ref{thm:seqcomlinear}, we need the following lemma.

\begin{lemma}\label{lemm:auxseqcomlinear}
Consider a sequential composition $Q=Q_1;\dots ;Q_n$ of non-expansive simple while loops that satisfies the condition (\ddag).
Let $\pv$ be a random program valuation that satisfies the loop guard of $Q_1$ a.s, and $Y$ be a non-negative random variable such that for all program variables $z$, $|\pv[z]|\le Y$ a.s. Then there exists a linear function $f$ determined by $\eta_1,\dots,\eta_n$ such that $\expv_\pv(|Z_Q|)\le f(Y)$ for all program variables $z$.
\end{lemma}
\begin{proof}
We prove by induction on $n\ge 1$. We denote by $d$ the bound for bounded update for all loops in $Q$, and by $K_i,\epsilon_i$ the parameters for $\eta_i$.
We also denote by $T_1$ the random variale for the number of loop iterations of $Q_1$.
Moreover, we denote by $|\pv|$ the vector obtained by taking the absolute value of every component in $\pv$.

\noindent{\bf Base Step} $n=1$, i.e., $Q=Q_1$. Then by Theorem~\ref{thm:rsmmaps1},
we have that $\expv_\pv(|Z_Q|)\le |\pv[z]|+\expv_\pv(\frac{d}{D_1}\cdot T_1)\le |\pv[z]|+\frac{d}{D_1}\cdot\frac{\eta_1(\pv)-K_1}{\epsilon_1}$. In this case, we can choose $f$ from $\eta_1$ and obtain the desired result.

\noindent{\bf Inductive Step} $Q=Q_1;Q'$ where $Q'$ is a sequential composition of simple while loops that satisfy ($\ddag$).
By Theorem~\ref{thm:rsmmaps1}, we have that $\expv_\pv(T_1)\le \frac{\eta_1(\pv)-K_1}{\epsilon_1}$.
Denote by $\pv'$ the random program valuation after the execution of $Q_1$ from the input random program valuation $\pv$.
Then from the triangle inequality and the bounded-update condition, we have that $\dist(\pv',\pv)\le d\cdot T_1$.
By induction hypothesis, we have that $\expv_{\pv'}(|Z_Q|)\le f'(|\pv|+\frac{d}{D_1}\cdot T_1)$ where $f'$ is a linear function determined by linear RSM-maps from $Q'$.
Hence, we have that $\expv_\pv(|Z_Q|)=\expv_\pv(\expv_{\pv'}(|Z_{Q'}|))\le \expv_\pv(f'(|\pv|+\frac{d}{D_1}\cdot T_1))\le f(|\pv|)$, where the last inequality is obtained by finding a linear function $f$
resulting from the expansion of the linear terms in $\expv_\pv(f'(|\pv|+\frac{d}{D_1}\cdot T_1))$ and the fact that $\expv_\pv(T_1)\le \frac{\eta_1(\pv)-K_1}{\epsilon_1}$. As $|\pv|\le Y$, we obtain the desired result.
\end{proof}

\noindent{\bf Theorem~\ref{thm:seqcomlinear}.}
Consider a non-expansive simple while loop $Q$ with loop guard $\Phi$ that has (i) bounded-update,
(ii) a difference-bounded linear RSM-map with RSM-continuity, and (iii) the 
Lipschitz continuity in next-step termination.
Then for any sequential composition $Q'$ of simple while loops that (a) satisfies the condition (\ddag) (defined right before the theorem) and (b) has expected linear-sensitivity
over a subset $U$ of input program valuations,
if $\Sat{\Phi}\cup\outp(Q)\subseteq U$, then the sequential composition $Q;Q'$ is expected linear-sensitive
over the satisfaction set of the loop guard of $Q$. 
\begin{proof}
The proof is an extension of that for Theorem~\ref{thm:linear}.
Let the sensitivity coefficient of $Q'$ be $A_{Q'}$ and $\theta'$ and the bound for bounded-update of $Q$ be $d$. 
Choose any program variable $z$. Denote by $T,T'$ (resp. $\pv_n,\pv'_n$) the random variables for the number of loop iterations (resp. the valuation at the $n$-th step),
from the input program valuations $\pv,\pv'$ and in the execution of the loop $Q$,
respectively.
Also denote by $Z_{Q''}$ the random variable that represents the value of $z$ after the execution of a proabilistic program $Q''$.
For each natural number $n\ge 0$, define
\[
\delta_n(\pv,\pv'):=\expv_{\pv}( \expv_{\pv_{T\wedge n}}(Z_{Q'}))-\expv_{\pv'}(\expv_{\pv'_{T'\wedge n}}(Z'_{Q'}))\enskip.
\]
for program valuations $\pv,\pv'$, where the random variable $T\wedge n$ is defined as $\min\{T,n\}$ and
$T'\wedge n$ likewise. We also define $\delta(\pv,\pv'):=\expv_{\pv}( \expv_{\pv_{T}}(Z_{Q'}))-\expv_{\pv'}(\expv_{\pv'_{T'}}(Z'_{Q'}))$.
By Lemma~\ref{lemm:auxseqcomlinear} and the fact that $|\pv_{T\wedge n}[z']|\le \pv[z']+\frac{d}{D_1}\cdot (T\wedge n)\le  \pv[z']+\frac{d}{D_1}\cdot T=:Y$,
we have that  $\expv_{\pv_{T\wedge n}}(|Z_{Q'}|)\le f(Y)$ and the same holds for $\pv'_{T\wedge n}$ with another linear function $f'$.
By Theorem~\ref{thm:rsmmaps1}, we have $\expv_\pv(f(Y))<\infty$ and the same holds for $\expv_{\pv'}(f'(Y))<\infty$.
Then
by the Dominated Convergence Theorem, we have that $\lim\limits_{n\rightarrow \infty}\delta_n(\pv,\pv')=\delta(\pv,\pv')$.

Let $\eta$ be a difference-bounded RSM-continuous RSM-map for $Q$ with the parameters $c,\epsilon,M,K$ as specified in (A3), (A4), (B3).
By Lemma~\ref{lemm:pd}, we can obtain a probability value $p=\frac{\epsilon}{2\cdot c-\epsilon}$ such that the condition $(\dag)$ holds.
We also construct the regions $R_k$'s ($1\le k\le n^*$) and $R_\infty$ as in the paragraph below Lemma~\ref{lemm:pd}, and have the solution $A_k$'s and $A_\infty$ from Proposition~\ref{prop:Ak}, for which we choose $C:=\max\{L'\cdot C',D\}$ and $D:=A_{Q'}=:A_0$, where the definition of $C'$ will be given below.
For the sake of convenience, we also define that $R_0:=\{\pv\mid \pv\not\models \Phi\}$.
Below we prove by induction on $n\ge 0$ that
\begin{compactitem}
\item[(*)] for all $k\in\{1,\dots,n^*,\infty\}$ and all program valuations $\pv,\pv'\models \Phi$ (where $\Phi$ is the loop guard of $Q$), if $\dist(\pv,\pv')\le \min\{\frac{1}{M},\theta'\}=:\theta$, then we have $|\delta_n(\pv,\pv')|\le A_k\cdot \dist(\pv,\pv')$ when $\pv\in R_k$ for $1\le k\le n^*$, and $|\delta_n(\pv,\pv')|\le A_\infty \cdot \dist(\pv,\pv')$ when $\pv\in R_{n^*}$.
\end{compactitem}

\noindent{\bf Base Step} $n=0$. By the expected linear sensitivity of $Q'$, we have that for all program valuations $\pv\in R_k$, $|\delta_0(\pv,\pv')|=|\expv_{\pv}(Z_{Q'})-\expv_{\pv'}(Z'_{Q'})|\le A_{Q'}\cdot\dist(\pv,\pv')= A_0\cdot\dist(\pv,\pv') \le A_k\cdot \dist(\pv,\pv')$.

\noindent{\bf Inductive Step.} Suppose that the induction hypothesis (*) holds for $n$. We prove the case for $n+1$.
We first consider $1\le k\le n^*$ and program valuations $\pv,\pv'$ such that
$\pv,\pv'\models\Phi$, $\dist(\pv,\pv')\le\theta$ and $\pv\in R_k$.
From the integral expansion (Theorem~\ref{thm:intexp}), we have that
\[
\expv_\pv( \expv_{\pv_{T\wedge (n+1)}}(Z_{Q'}))=\int \left[\mathbf{1}_{\updf(\loc,\pv,\sv)\not\models \Phi}\cdot \expv_{\updf(\loc,\pv,\sv)}(Z_{Q'}) + \mathbf{1}_{\updf(\loc,\pv,\sv)\models \Phi}\cdot \expv_{\updf(\loc,\pv,\sv)}( \expv_{\updf(\loc,\pv,\sv)_{T\wedge n}}(Z_{Q'}))\right]\,\mathrm{d}\loc \,\mathrm{d}\sv
\]
and similarly,
\[
\expv_{\pv'}( \expv_{\pv'_{T'\wedge (n+1)}}(Z'_{Q'}))=\int \left[\mathbf{1}_{\updf(\loc,\pv',\sv)\not\models \Phi}\cdot \expv_{\updf(\loc,\pv',\sv)}(Z'_{Q'}) + \mathbf{1}_{\updf(\loc,\pv',\sv)\models \Phi}\cdot \expv_{\updf(\loc,\pv',\sv)}( \expv_{\updf(\loc,\pv',\sv)_{T'\wedge n}}(Z'_{Q'}))\right]\,\mathrm{d}\loc \,\mathrm{d}\sv
\]
From (\dag), we have that with probability at least $p=\frac{\epsilon}{2\cdot c-\epsilon}$, it happens that $\eta(\updf(\loc,\pv,\sv))-\eta(\pv)\le -\frac{1}{2}\cdot \epsilon$.
It follows that with probability at least $p$, $\updf(\loc,\pv,\sv)\in \bigcup_{m=0}^{k-1} R_{m}$.
Note that
\begin{eqnarray*}
& & \expv_{\pv}( \expv_{\pv_{T\wedge (n+1)}}(Z_{Q'}))-\expv_{\pv'}(\expv_{\pv'_{T'\wedge (n+1)}}(Z'_{Q'})) \\
&=& \int  \left[\mathbf{1}_{\updf(\loc,\pv,\sv)\not\models\Phi\wedge \updf(\loc,\pv',\sv)\not\models\Phi}\cdot (\expv_{\updf(\loc,\pv,\sv)}(Z_{Q'})-\expv_{\updf(\loc,\pv',\sv)}(Z'_{Q'}))\right]\,\mathrm{d}\loc\,\mathrm{d}\sv  \\
& &~~{}+\int \left[\mathbf{1}_{\updf(\loc,\pv,\sv)\not\models\Phi\wedge \updf(\loc,\pv',\sv)\models\Phi}\cdot (\expv_{\updf(\loc,\pv,\sv)}(Z_{Q'})-\expv_{\updf(\loc,\pv',\sv)}( \expv_{\updf(\loc,\pv',\sv)_{T'\wedge n}}(Z'_{Q'})))\right]\,\mathrm{d}\loc\,\mathrm{d}\sv\\
& &~~{}+\int \left[\mathbf{1}_{\updf(\loc,\pv,\sv)\models\Phi\wedge \updf(\loc,\pv',\sv)\not\models\Phi}\cdot (\expv_{\updf(\loc,\pv,\sv)}(\expv_{\updf(\loc,\pv,\sv)_{T\wedge n}}(Z_{Q'}))-\expv_{\updf(\loc,\pv',\sv)}(Z'_{Q'}))\right]\,\mathrm{d}\loc\,\mathrm{d}\sv\\
& &~~{}+\int \left[\mathbf{1}_{\updf(\loc,\pv,\sv)\models\Phi\wedge \updf(\loc,\pv',\sv)\models\Phi}\cdot (\expv_{\updf(\loc,\pv,\sv)}(\expv_{\updf(\loc,\pv,\sv)_{T\wedge n}}(Z_{Q'}))-\expv_{\updf(\loc,\pv',\sv)}( \expv_{\updf(\loc,\pv',\sv)_{T'\wedge n}}(Z'_{Q'})))\right]\,\mathrm{d}\loc\,\mathrm{d}\sv\\
\end{eqnarray*}
where the first integral corresponds to the case that the executions of $Q$ from $\pv,\pv'$ both terminate after one loop iteration,
the second and third integrals correspond to the case that one execution terminates but the other does not,
and the last integral correspond to the case that both the executions do not terminate.
In the first integral, we have from the non-expansiveness that
\[
|\expv_{\updf(\loc,\pv,\sv)}(Z_{Q'})-\expv_{\updf(\loc,\pv',\sv)}(Z'_{Q'})|\le A_{Q'}\cdot \dist(\updf(\loc,\pv,\sv),\updf(\loc,\pv',\sv))\le A_{Q'}\cdot\dist(\pv,\pv')\enskip.
\]
In the second case, w.l.o.g. we can assume that $\updf(\loc,\pv,\sv)\not\models\Phi$ and $\updf(\loc,\pv',\sv)\models\Phi$.
From the expected linear-sensitivity of $Q'$, we obtain that
    \[
    |\expv_{\updf(\loc,\pv,\sv)}(Z_{Q'})-\expv_{\updf(\loc,\pv',\sv)_{T'\wedge n}}(Z'_{Q'})|\le A_{Q'}\cdot \dist(\updf(\loc,\pv,\sv), \updf(\loc,\pv',\sv)_{T'\wedge n})\enskip.
    \]
    By the triangle inequality and the bounded update condition, we have that
\begin{eqnarray*}
\dist(\updf(\loc,\pv,\sv), \updf(\loc,\pv',\sv)_{T'\wedge n}) & \le & \dist(\updf(\loc,\pv,\sv), \updf(\loc,\pv',\sv))+\dist(\updf(\loc,\pv',\sv),\updf(\loc,\pv',\sv)_{T'\wedge n}) \\
&\le & \dist(\pv,\pv') + (T_{\updf(\loc,\pv',\sv)}\wedge n)\cdot d \\
&\le & \dist(\pv,\pv') + T_{\updf(\loc,\pv',\sv)}\cdot d \enskip.\\
\end{eqnarray*}
Thus, we have that
\begin{eqnarray*}
|\expv_{\updf(\loc,\pv,\sv)}(Z_{Q'})-\expv_{\updf(\loc,\pv',\sv)}( \expv_{\updf(\loc,\pv',\sv)_{T'\wedge n}}(Z'_{Q'}))| 
&=& |\expv_{\updf(\loc,\pv',\sv)}(\expv_{\updf(\loc,\pv,\sv)}(Z_{Q'})- \expv_{\updf(\loc,\pv',\sv)_{T'\wedge n}}(Z'_{Q'}))| \\
&\le & \expv_{\updf(\loc,\pv',\sv)}(|\expv_{\updf(\loc,\pv,\sv)}(Z_{Q'})- \expv_{\updf(\loc,\pv',\sv)_{T'\wedge n}}(Z'_{Q'})|)  \\
&\le& \expv_{\updf(\loc,\pv',\sv)}(A_{Q'}\cdot (\dist(\pv,\pv') + T_{\updf(\loc,\pv',\sv)}\cdot d)) \\
&=& A_{Q'}\cdot \dist(\pv,\pv') + A_{Q'}\cdot \expv_{\updf(\loc,\pv',\sv)} (T_{\updf(\loc,\pv',\sv)})\cdot d\\
&\le& A_{Q'}\cdot \dist(\pv,\pv') + A_{Q'}\cdot d\cdot \frac{\eta(\updf(\loc,\pv',\sv))-K}{\epsilon}\\
&\le& A_{Q'}\cdot \dist(\pv,\pv') + A_{Q'}\cdot d\cdot \frac{M\cdot\dist(\pv,\pv')-K}{\epsilon}\\
&\le& A_{Q'}\cdot \theta + A_{Q'}\cdot d\cdot \frac{M\cdot\theta-K}{\epsilon}=: C'\\
\end{eqnarray*}
Furthermore, we have
\begin{eqnarray*}
& &\int \mathbf{1}_{\updf(\loc,\pv,\sv)\models\Phi\wedge \updf(\loc,\pv',\sv)\models\Phi}\cdot \alpha(\loc,\sv)\,\mathrm{d}\loc \,\mathrm{d}\sv\\
&=&\int \mathbf{1}_{\updf(\loc,\pv,\sv)\models\Phi\wedge \updf(\loc,\pv',\sv)\models\Phi\wedge \eta(\updf(\loc,\pv,\sv))-\eta(\pv)\le-\frac{1}{2}\cdot\epsilon}\cdot  \alpha(\loc,\sv)\,\mathrm{d}\loc\,\mathrm{d}\sv\\
& &~~{}+\int \mathbf{1}_{\updf(\loc,\pv,\sv)\models\Phi\wedge \updf(\loc,\pv',\sv)\models\Phi\wedge \eta(\updf(\loc,\pv,\sv))-\eta(\pv)>-\frac{1}{2}\cdot \epsilon}\cdot \alpha(\loc,\sv)\,\mathrm{d}\loc\,\mathrm{d}\sv\enskip.\\
\end{eqnarray*}
where $\alpha=\expv_{\updf(\loc,\pv,\sv)}(\expv_{\updf(\loc,\pv,\sv)_{T\wedge n}}(Z_{Q'}))-\expv_{\updf(\loc,\pv',\sv)}( \expv_{\updf(\loc,\pv',\sv)_{T'\wedge n}}(Z'_{Q'}))$. Denote
\begin{compactitem}
\item
$q_1:=\probm_{\loc,\sv}(\updf(\loc,\pv,\sv)\not\models\Phi\wedge \updf(\loc,\pv',\sv)\not\models\Phi)$;
\item
$q_2:=\probm_{\loc,\sv}((\updf(\loc,\pv,\sv)\not\models\Phi\wedge \updf(\loc,\pv',\sv)\models\Phi)\vee (\updf(\loc,\pv,\sv)\models\Phi\wedge \updf(\loc,\pv',\sv)\not\models\Phi))$;
\item
$p':=\probm_{\loc,\sv}(\updf(\loc,\pv,\sv)\models\Phi\wedge \updf(\loc,\pv',\sv)\models\Phi\wedge\eta(\updf(\loc,\pv,\sv))\le \eta(\pv)-\frac{\epsilon}{2})$;
\item
$\overline{p}:=\probm_{\loc,\sv}(\eta(\updf(\loc,\pv,\sv))\le \eta(\pv)-\frac{\epsilon}{2})$.
\end{compactitem}
Then $q_1+q_2+p'\ge\overline{p}\ge p$.
From the 
Lipschitz continuity in next-step termination, 
we have that $q_2\le L'\cdot \dist(\pv,\pv')$.
Then from $A_0\le A_1\le\dots\le A_k\le A_\infty$ and the induction hypothesis,
we have
\begin{eqnarray*}
|\delta_{n+1}(\pv,\pv')| &\le& (p'\cdot A_{k-1}+ q_1\cdot A_{Q'}  + (1-(q_1+q_2+p'))\cdot A_\infty)\cdot \dist(\pv,\pv') + q_2\cdot C'\\
&\le& (p'\cdot A_{k-1}+ q_1\cdot A_{Q'} + L'\cdot C' + (1-(q_1+q_2+p'))\cdot A_\infty)\cdot \dist(\pv,\pv') \\
&\le & ((q_1+q_2+p')\cdot A_{k-1}+ C + (1-(q_1+q_2+p'))\cdot A_\infty)\cdot \dist(\pv,\pv')\\
&\le & (p\cdot A_{k-1}+ C + (1-p)\cdot A_\infty)\cdot \dist(\pv,\pv')\\
&=& A_k\cdot \dist(\pv,\pv')\enskip.
\end{eqnarray*}
Then we consider the case $\pv\in R_{\infty}$.
In this case, since $\dist(\pv,\pv')\le \theta\le \frac{1}{M}$ and $\eta(\pv)>c+1$, we have from the RSM-continuity that $\eta(\updf(\loc,\pv',\sv)),\eta(\updf(\loc,\pv',\sv))>0$).
It follows that both $\updf(\loc,\pv,\sv),\updf(\loc,\pv',\sv)$ satisfy the loop guard.
Hence, we have from the induction hypothesis that $|\delta_n(\pv,\pv')| \le  A_\infty\cdot \dist(\pv,\pv')$.
Thus, the induction step is proved. By taking the limit $n\rightarrow\infty$, we obtain that the whole loop $Q$ is expected linear-sensitive in each $R_k$.
By taking the maximum constant $A_\infty$ for the global expected linear-sensitivity, we obtain the desired result.
\end{proof}

\section{Details for Experimental results}\label{app:running}

We consider examples and their variants from the literature~\cite{DBLP:conf/cav/ChatterjeeFG16,DBLP:conf/ijcai/ChatterjeeFGO18,DBLP:journals/toplas/ChatterjeeFNH18,DBLP:conf/pldi/NgoC018}.
Below we show all the experimental examples. 

\lstset{language=prog}
\lstset{tabsize=3}
\newsavebox{\onerobot}
\begin{lrbox}{\onerobot}
	\begin{lstlisting}[mathescape]
	 while $x\le 1000$ do
	    if prob(0.6) then
           $x:=x+1$
        else
           $x:=x-1$
        fi
	 od
	\end{lstlisting}
\end{lrbox}


\lstset{language=prog}
\lstset{tabsize=3}
\newsavebox{\oneDrobot}
\begin{lrbox}{\oneDrobot}
	\begin{lstlisting}[mathescape]
     $r\sim unif(1,3)$;
	 while $x\le 1000$ do
	    if prob(0.6) then
           $x:=x+r$
        else
           $x:=x-r$
        fi
	 od
	\end{lstlisting}
\end{lrbox}

\begin{figure}
\begin{minipage}{0.4\textwidth}
\usebox{\onerobot}
\caption{rdwalk}
\label{fig:running1}
\end{minipage}
\quad\quad\quad\quad\begin{minipage}{0.4\textwidth}
\usebox{\oneDrobot}
\caption{A variant of rdwalk}
\label{fig:running2}
\end{minipage}
\end{figure}

\lstset{language=prog}
\lstset{tabsize=3}
\newsavebox{\prdwalk}
\begin{lrbox}{\prdwalk}
	\begin{lstlisting}[mathescape]
     $r_1\sim unif(0,2),r_2\sim unif(0,5)$;
	 while $x\le 1000$ do
	    if prob(0.5) then
           $x:=x+r_1$
        else
           $x:=x+r_2$
        fi
	 od
	\end{lstlisting}
\end{lrbox}


\lstset{language=prog}
\lstset{tabsize=3}
\newsavebox{\prdwalkk}
\begin{lrbox}{\prdwalkk}
	\begin{lstlisting}[mathescape]
	 while $x\le 1000$ do
	    if prob(0.5) then
           $x:=x+2$
        else
           $x:=x+5$
        fi
	 od
	\end{lstlisting}
\end{lrbox}

\begin{figure}
\begin{minipage}{0.4\textwidth}
\usebox{\prdwalk}
	\caption{prdwalk}
	\label{fig:running4}
\end{minipage}
\quad\quad\quad\quad\quad\begin{minipage}{0.4\textwidth}
	\usebox{\prdwalkk}
	\caption{A Variant of prdwalk}
	\label{fig:running18}
\end{minipage}
\end{figure}

\lstset{language=prog}
\lstset{tabsize=3}
\newsavebox{\prspeed}
\begin{lrbox}{\prspeed}
	\begin{lstlisting}[mathescape]
	 while $x\le 1000$ do
           if prob(0.75) then
              $x:=x+0$
           else
              if prob(2/3) then
                 $x:=x+2$
              else
                 $x:=x+3$
              fi
           fi
	 od
	\end{lstlisting}
\end{lrbox}


\lstset{language=prog}
\lstset{tabsize=3}
\newsavebox{\prspeedd}
\begin{lrbox}{\prspeedd}
	\begin{lstlisting}[mathescape]
     $r_1\sim unif(1,2),r_2\sim unif(2,3),$
     $r_3\sim unif(3,5)$
	 while $x\le 1000$ do
           if prob(0.75) then
              $x:=x+r_1$
           else
              if prob(2/3) then
                 $x:=x+r_2$
              else
                 $x:=x+r_3$
              fi
           fi
	 od
	\end{lstlisting}
\end{lrbox}

\begin{figure}
\begin{minipage}{0.4\textwidth}
\usebox{\prspeed}
	\caption{prspeed}
	\label{fig:running5}
\end{minipage}
\quad\quad\quad\quad\quad\quad\begin{minipage}{0.4\textwidth}
	\usebox{\prspeedd}
	\caption{A Variant of prspeed}
	\label{fig:running6}
\end{minipage}
\end{figure}

\lstset{language=prog}
\lstset{tabsize=3}
\newsavebox{\race}
\begin{lrbox}{\race}
	\begin{lstlisting}[mathescape]
     $r\sim unif(2,4)$;
	 while $h\le t$ do
        $t:=t+1$;
	    if prob(0.5) then
           $h:=h+r$
        else
           skip
        fi
	 od
	\end{lstlisting}
\end{lrbox}


\lstset{language=prog}
\lstset{tabsize=3}
\newsavebox{\racee}
\begin{lrbox}{\racee}
	\begin{lstlisting}[mathescape]
	 while $h\le t$ do
        $t:=t+1$;
	    if prob(0.5) then
           $h:=h+1$
        else
           if prob(0.6)
              $h:=h+4$
           else
              skip
           fi
        fi
	 od
	\end{lstlisting}
\end{lrbox}

\begin{figure}
\begin{minipage}{0.4\textwidth}
\usebox{\race}
	\caption{race}
	\label{fig:running7}
\end{minipage}
\quad\quad\quad\begin{minipage}{0.4\textwidth}
	\usebox{\racee}
	\caption{A Variant of race}
	\label{fig:running8}
\end{minipage}
\end{figure}

\lstset{language=prog}
\lstset{tabsize=3}
\newsavebox{\progbb}
\begin{lrbox}{\progbb}
	\begin{lstlisting}[mathescape]
     $r\sim unif(0,1)$
	 while $x\le 1000$ do
	   $x:=x+r$
	 od
	\end{lstlisting}
\end{lrbox}


\lstset{language=prog}
\lstset{tabsize=3}
\newsavebox{\pollutant}
\begin{lrbox}{\pollutant}
	\begin{lstlisting}[mathescape]
     $r_1\sim unif(2,4),r'_1\sim unif(1,3),$
     $r_2\sim unif(1,3),r'_2\sim unif(0,2)$;
	 while $n\ge 5$ do
	    if prob(0.6) then
           $n:=n-r_1$;
           $n:=n+r'_1$
        else
           $n:=n-r_2$;
           $n:=n+r'_2$
        fi
	 od
	\end{lstlisting}
\end{lrbox}

\begin{figure}
\begin{minipage}{0.4\textwidth}
\usebox{\progbb}
	\caption{A Simple Probabilistic While Loop}
	\label{fig:running9}
\end{minipage}
\quad\quad\quad\quad\quad\begin{minipage}{0.4\textwidth}
	\usebox{\pollutant}
	\caption{Pollutant Disposal}
	\label{fig:running3}
\end{minipage}
\end{figure}

\lstset{language=prog}
\lstset{tabsize=3}
\newsavebox{\adversarial}
\begin{lrbox}{\adversarial}
	\begin{lstlisting}[mathescape]
	 while $x\le y$ do
       if prob(0.5) then
          if prob(0.7) then
             $x:=x+3$
          else
             $y:=y+2$
          fi
       else
          if prob(0.7) then
             $x:=x+2$
          else
             $y:=y+1$
          fi
	   fi
	 od
	\end{lstlisting}
\end{lrbox}


\lstset{language=prog}
\lstset{tabsize=3}
\newsavebox{\adversariall}
\begin{lrbox}{\adversariall}
	\begin{lstlisting}[mathescape]
     $r_1,r_2\sim unif(2,4),r'_1,r'_2\sim unif(1,2)$
	 while $x\le y$ do
       if prob(0.5) then
          if prob(0.7) then
             $x:=x+r_1$
          else
             $y:=y+r'_1$
          fi
       else
          if prob(0.7) then
             $x:=x+r_2$
          else
             $y:=y+r'_2$
          fi
	   fi
	 od
	\end{lstlisting}
\end{lrbox}

\begin{figure}
\begin{minipage}{0.4\textwidth}
\usebox{\adversarial}
	\caption{Adversarial random walk in two dimensions}
	\label{fig:running10}
\end{minipage}
\quad\quad\quad\quad\quad\begin{minipage}{0.4\textwidth}
	\usebox{\adversariall}
	\caption{A Variant of adversarial random walk in two dimensions}
	\label{fig:running11}
\end{minipage}
\end{figure}

\lstset{language=prog}
\lstset{tabsize=3}
\newsavebox{\gambler}
\begin{lrbox}{\gambler}
	\begin{lstlisting}[mathescape]
     $r\sim unif(-1,1)$
	 while $1\le x \wedge x\le 10$ do
           $x:=x+r$
	 od
	\end{lstlisting}
\end{lrbox}


\lstset{language=prog}
\lstset{tabsize=3}
\newsavebox{\gamblerr}
\begin{lrbox}{\gamblerr}
	\begin{lstlisting}[mathescape]
	 while $1\le x \wedge x\le 10$ do
        if prob(0.5) then
           $x:=x+1$
        else
           $x:=x-1$
        fi
	 od
	\end{lstlisting}
\end{lrbox}

\begin{figure}
\begin{minipage}{0.4\textwidth}
\usebox{\gambler}
	\caption{Gambler's Ruin}
	\label{fig:running12}
\end{minipage}
\quad\quad\quad\begin{minipage}{0.4\textwidth}
	\usebox{\gamblerr}
	\caption{A Variant of Gambler's Ruin}
	\label{fig:running13}
\end{minipage}
\end{figure}

\lstset{language=prog}
\lstset{tabsize=3}
\newsavebox{\oneadversarial}
\begin{lrbox}{\oneadversarial}
	\begin{lstlisting}[mathescape]
     $r\sim unif(-1,1)$
	 while $x\ge 0$ do
           $x:=x+r$;
           if prob(0.5) then
              if prob(0.9) then
                 $x:=x-1$
              else
                 $x:=x+1$
              fi
           else
              $x:=x-1$
           fi
	 od
	\end{lstlisting}
\end{lrbox}


\lstset{language=prog}
\lstset{tabsize=3}
\newsavebox{\oneadversariall}
\begin{lrbox}{\oneadversariall}
	\begin{lstlisting}[mathescape]
	 while $x\ge 0$ do
           $x:=x+1$;
           if prob(0.5) then
              if prob(0.9) then
                 $x:=x-2$
              else
                 $x:=x+1$
              fi
           else
              $x:=x-1$
           fi
	 od
	\end{lstlisting}
\end{lrbox}

\begin{figure}
\begin{minipage}{0.4\textwidth}
\usebox{\oneadversarial}
	\caption{Adversarial random walk in one dimension}
	\label{fig:running14}
\end{minipage}
\quad\quad\quad\quad\quad\begin{minipage}{0.4\textwidth}
	\usebox{\oneadversariall}
	\caption{A Variant of Adversarial random walk in one dimension}
	\label{fig:running15}
\end{minipage}
\end{figure}

\lstset{language=prog}
\lstset{tabsize=3}
\newsavebox{\AmericanRoulette}
\begin{lrbox}{\AmericanRoulette}
\begin{lstlisting}[mathescape]
while $x\ge 1$ do
   if prob(1/304) then
      $x:=x+35$;$w:=w+35$
   else if prob(2/303) then
      $x:=x+17$;$w:=w+17$
   else if prob(3/301) then
      $x:=x+11$;$w:=w+11$
   else if prob(2/149) then
      $x:=x+8$;$w:=w+8$
   else if prob(5/294) then
      $x:=x+6$;$w:=w+6$
   else if prob(6/289) then
      $x:=x+5$;$w:=w+5$
   else if prob(12/283) then
      $x:=x+2$;$w:=w+2$
   else if prob(2/271) then
      $x:=x-0.5$
   else if prob(18/269) then
      $x:=x+1$;$w:=w+1$
   else if prob(2/251) then
      $x:=x-0.5$
   else
      $x:=x-1$
   fi fi fi fi fi fi fi fi fi fi
od
\end{lstlisting}
\end{lrbox}
%


\lstset{language=prog}
\lstset{tabsize=3}
\newsavebox{\AmericanRoulettee}
\begin{lrbox}{\AmericanRoulettee}
\begin{lstlisting}[mathescape]
$r_1\sim unif(30,35),r_2\sim unif(12,17),$
$r_3\sim unif(9,11),r_4\sim unif(7,8),$
$r_5\sim unif(5,6),r_6\sim unif(3,5)$
$r_7\sim unif(2,3),r'_7\sim unif(0.5,1),$
$r_8\sim unif(1,2),r'_8\sim unif(0.5,1),$
$r_9\sim unif(1,2)$
while $x\ge 1$ do
   if prob(1/304) then
      $x:=x+r_1$;$w:=w+35$
   else if prob(2/303) then
      $x:=x+r_2$;$w:=w+17$
   else if prob(3/301) then
      $x:=x+r_3$;$w:=w+11$
   else if prob(2/149) then
      $x:=x+r_4$;$w:=w+8$
   else if prob(5/294) then
      $x:=x+r_5$;$w:=w+6$
   else if prob(6/289) then
      $x:=x+r_6$;$w:=w+5$
   else if prob(12/283) then
      $x:=x+r_7$;$w:=w+2$
   else if prob(2/271) then
      $x:=x-r'_7$
   else if prob(18/269) then
      $x:=x+r_8$;$w:=w+1$
   else if prob(2/251) then
      $x:=x-r'_8$
   else
      $x:=x-r_9$
   fi fi fi fi fi fi fi fi fi fi
od
\end{lstlisting}
\end{lrbox}

\begin{figure}
\begin{minipage}{0.4\textwidth}
\usebox{\AmericanRoulette}
	\caption{American Roulette}
	\label{fig:running16}
\end{minipage}
\quad\quad\quad\quad\quad\quad\quad\begin{minipage}{0.4\textwidth}
	\usebox{\AmericanRoulettee}
	\caption{A Variant of American Roulette}
	\label{fig:running17}
\end{minipage}
\end{figure}

\end{document}